\documentclass[10pt,english]{article} 
\usepackage{xcolor} 
\usepackage{amsmath, amsfonts, amsthm, amssymb, amscd}
\usepackage{xcolor} 
\usepackage{amsmath, amsfonts, amsthm, amssymb, amscd}
\usepackage[mathcal]{euscript}
\usepackage{mathrsfs}
\usepackage{a4wide}
\usepackage{graphicx}
\usepackage{ulem}  
\usepackage{subfig}
\usepackage{mathrsfs}
\usepackage{rotating}
\usepackage{verbatim}
\usepackage{slashed}
\usepackage{perpage}  
\usepackage{tensor}       
\renewcommand{\emph}[1]{{\sl #1}}
 
\newtheorem{theorem}{Theorem}[section]
\newtheorem{proposition}[theorem]{Proposition}
\newtheorem{lemma}[theorem]{Lemma}

\newtheorem{definition}[theorem]{Definition}

\newtheorem{remark}[theorem]{Remark}

\numberwithin{equation}{section}
\DeclareFontFamily{U}{BOONDOX-calo}{\skewchar\font=45 }
\DeclareFontShape{U}{BOONDOX-calo}{m}{n}{
<-> s*[1.05] BOONDOX-r-calo}{}
\DeclareFontShape{U}{BOONDOX-calo}{b}{n}{
<-> s*[1.05] BOONDOX-b-calo}{}
\DeclareMathAlphabet{\mathcalboondox}{U}{BOONDOX-calo}{m}{n}
\SetMathAlphabet{\mathcalboondox}{bold}{U}{BOONDOX-calo}{b}{n}
\DeclareMathAlphabet{\mathbcalboondox}{U}{BOONDOX-calo}{b}{n} 

\newcommand \epss {\eps_\star}

\renewcommand {\Bbb}{\mathbb{B}}
\newcommand {\Cbb}{\mathbb C}
\newcommand {\Gbb}{\mathbb G}
\newcommand {\Mbb}{\mathbb M}

\newcommand \wR {\tensor[^{(w)}]{}{}  R}

\newcommand {\us}{\slashed{u}}
\newcommand \Rwave {\tensor[^{(w)} ]{R}{^\star}}

\newcommand \gd g  

\newcommand \gconf g                 
\newcommand \kd {{k}}                     

\newcommand \nablad    \nabla 

\newcommand \givenCurv {\widetilde S}

\newcommand \givennabla {\widetilde \nabla} 

\newcommand \givenBox {\widetilde \Box}

\newcommand \givenVol {\widetilde{\textrm{Vol}}}
\newcommand \givenT {\widetilde{T}}
\newcommand \givenG {\widetilde{\text{G}}}

\newcommand \givenRic {\widetilde{\textrm{Ric}}}
\newcommand \givenScal {\widetilde{\text{R}}}
\newcommand \givenR {\widetilde{\text{R}}}

\newcommand \FCenergy {\textbf{ FC}}
\newcommand \FMenergy {\textbf{ FM}}
\newcommand \FSenergy {\textbf{ FS}}

\newcommand \NMenergy {\textbf{ NM}}
\newcommand \NCenergy {\textbf{ NC}}
\newcommand \NSenergy {\textbf{ NS}}

\newcommand \fTime T        
\newcommand \Minsk \eta
\newcommand \rhoH    {r^{\mathcal H}}    
\newcommand \rhoE     {r^{\mathcal E}}   

\newcommand \Fenergy      {\mathbf F}              
\newcommand \Eenergy      {\mathbf E}           

\newcommand \EMenergy{\mathbf {EM}} 
\newcommand \ECenergy{\mathbf {EC}} 
\newcommand \ESenergy{\mathbf {ES}} 

\newcommand \expeta \sigma 
 
\newcommand \E {\mathcal{E}} 

\newcommand \ME {\mathcal{EM}} 
\newcommand \EM {\mathcal{EM}} 

\renewcommand \H {\mathcal H} 

\newcommand \M {\mathcal M} 

\newcommand \N {\mathcal N} 

\newcommand \Mscr       {\mathscr M} 

\newcommand \MH 		{\Mcal^{\mathcal{H}}} 

\newcommand \MME 		{\Mcal^{\mathcal{EM}}}  

\newcommand \MM 		{\Mcal^{\mathcal{M}}} 

\newcommand \MMEfar 	{\Mcal^{\textbf{far}}} 
\newcommand \MMEnear  	{\Mcal^{\textbf{near}}} 
 
\newcommand \Mnear {\Mcal^\textbf{near}}
\newcommand \Mfar{\Mcal^\textbf{far}}
\newcommand{\RR}{\mathbb{R}}
\newcommand{\del}{\partial}
\newcommand{\Mcal}{\mathcal{M}}
\newcommand{\Tr}{\text{Tr}}

\newcommand{\Gammaj}{\widetilde{\Gamma}}

\newcommand{\hfrak}{\mathfrak{h}}
\newcommand{\ffrak}{\mathfrak{f}}
\newcommand \Ufrak {\mathfrak U}

\newcommand{\Ebf}{{\bf E}}

\newcommand{\nablab}{\overline{\nabla}}
 
\newcommand \sourterm   H 
\newcommand \coeffi \kappa 

\newcommand \bse {\begin{subequations}}
\newcommand \ese {\end{subequations}}

\newcommand \giveng {{\widetilde{g}}}

\newcommand \Sch {S}

\newcommand \Boxt {\widetilde \Box}

\newcommand \Lscr{\mathscr{L}}

\newcommand \Sbb {\mathbb S}

\newcommand \la \langle
\newcommand \ra \rangle

\newcommand \bei  {\begin{itemize}}
\newcommand \eei {\end{itemize}}

\newcommand \BoxChapeau {\widehat \Box}

\newcommand \BoxChapeaud {\BoxChapeau} 
\newcommand \Lcal {\mathcal L}

\newcommand \Hb{\overline H}

\newcommand \Hcal {\mathcal H}

\newcommand \delu {\underline{\del}}

\newcommand \Hu {\underline{H}}

\newcommand {\gu}{\underline{g}}

\newcommand {\eps} \epsilon

\let\oldmarginpar\marginpar
\renewcommand\marginpar[1]{\- \oldmarginpar[\raggedleft\footnotesize #1]%
{\raggedright\footnotesize #1}}
\newcommand \phizerod {{\phi_0}}
\newcommand \phioned {{\phi_1}}

\newcommand \TJ T

\newcommand \Gammad  \Gamma    

\newcommand \Ncal {\mathcal N}

\newcommand \Mtran {\Mcal^\text{tran}}
\newcommand \Mext {\Mcal^\text{ext}}

\newcommand \crochet \vartheta  

\newcommand \delH {\del^{\mathcal H}} 
\newcommand \delsH {\slashed \del^{\mathcal H}}
\newcommand \PsiH {{\Psi^{\mathcal H}}{}}

\newcommand {\usH} {\slashed{u}^{\mathcal{H}}}

\newcommand \delN {\del^{\mathcal{N}}}
\newcommand \delsN {\slashed \del^{\mathcal{N}}}
\newcommand \PsiN {{\Psi^{\mathcal N}{}}}  
\newcommand \PhiN {{\Phi^{\mathcal N}}{}}

\newcommand \HN {H^{\mathcal{N}}{}}
\newcommand \hN {h^{\mathcal{N}}}

\newcommand \delsEH {\slashed {\del}^{\mathcal{EH}}}
\newcommand \delsME {\slashed {\del}^{\mathcal{EM}}} 

\newcommand \Fbb {\mathbb F}
\newcommand \Pbb {\mathbb P}
\newcommand \Qbb {\mathbb Q}

\newcommand \SbbME{\mathbb S^{\mathcal{EM}}}

\newcommand {\Lbb} {\mathbb{L}}
\newcommand \Ibb	{\mathbb I}  	
\newcommand \init {\textbf{init}}
\newcommand \dive {\text{div} \hskip.05cm}
\newcommand \delEH {\del^{\mathcal{E \hskip-.06cm H}}} 
\newcommand \uts   {{\slashed u}^\Ncal}

\newcommand \delus  \delsH

\newcommand \delts  \delsN

\renewcommand \gu {g^\Hcal{}}

\renewcommand \Hu {H^\Hcal{}}

\newcommand \lbf {\mathbf l}
\newcommand \ord {\textbf{ord}}

\renewcommand \deg {\textbf{deg}}

\newcommand \Mink {\textbf{Mink}}
\newcommand \rank {\textbf{rank}}
\newcommand \gMink {g_\textbf{Mink}}
\renewcommand \Sch {\textbf{Sch}}

\newcommand \near {\textbf{near}}

\newcommand \rr {{rr}}

\newcommand \glue {{}}                 

\newcommand \easy {\textbf{easy}}
\newcommand \hard {\textbf{hard}}

\newcommand \LOmega {Y_\text{rot}} 

\newcommand \exposant \rho
\newcommand \vecnnu \nu
\newcommand \notrelapse L

\newcommand \sour {\textbf{sour}}

\newcommand \eff {\textbf{eff}}

\newcommand \be {\begin{equation}}
\newcommand \ee {\end{equation}} 
\newcommand \bel {\begin{equation}\label} 

\begin{document}

\title{The Euclidean-hyperboloidal foliation method. 
\\
Application to f(R) modified gravity\footnotetext{$^1$Laboratoire Jacques-Louis Lions, Sorbonne Universit\'e 
and Centre National de la Recherche Scientifique, 4~Place Jussieu, 75252 Paris, France. 
Email: {\sl contact@philippelefloch.org}.
\newline 
$^2$School of Mathematics and Statistics, Xi'an Jiaotong University, Xi'an, 710049 Shaanxi, People's Republic of China. 
Email: {\sl yuemath@mail.xjtu.edu.cn} 
\hfill 
\newline
To appear in: {\bf General Relativity and Gravitation} for the Issue \it ``Hyperboloidal foliations in the era of gravitational wave astronomy: from mathematical relativity to astrophysics''. 
}
} 

\author{\large Philippe G. LeFloch$^1$ and Yue Ma$^2$} 

\date{} 

\maketitle 

\begin{abstract}    
This paper is a part of a series devoted to the Euclidean-hyperboloidal foliation method introduced by the authors for investigating the global existence problem associated with nonlinear systems of coupled wave-Klein-Gordon equations with small data. This method was developed especially for investigating the initial value problem for the Einstein-massive field system in wave gauge. Here, we study the (fourth-order) field equations of $f(R)$ modified gravity and investigate the global dynamical behavior of the gravitational field in the near-Minkowski regime. We establish the existence of a globally hyperbolic Cauchy development approaching Minkowski spacetime (in spacelike, null, and timelike directions), when the initial data set is sufficiently close to an asymptotically Euclidean and spacelike hypersurface in Minkowski spacetime. We cast the (fourth-order) $f(R)$-field equations in the form of a second-order wave-Klein-Gordon system, which has an analogous structure to the Einstein-massive field system but, in addition, involves a (possibly small) effective mass parameter. We establish the nonlinear stability of the Minkowski spacetime in the context of $f(R)$ gravity, when the integrand $f(R)$ in the action functional can be taken to be arbitrarily close to the integrand $R$ of the standard Hilbert-Einstein action. 
\end{abstract}


{\small

\setcounter{secnumdepth}{2}
\setcounter{tocdepth}{2} 

\tableofcontents
}
  

\section{Introduction} 
\label{section=1}

\paragraph{Global evolution for a model of modified gravity.}

The field equations of $f(R)$ gravity, an extension of Einstein's gravity theory, have received only limited attention by the mathematical community, as far as the global existence and stability theory of spacetimes is concerned. In the present paper, our aim is to initiate the study of the global evolution problem and put this theory on solid mathematical grounds. Building upon recent advances on the analysis of nonlinear wave-Klein-Gordon systems, we apply here the {\sl Euclidean-hyperboloidal foliation method} introduced by the authors in recent years (see references below). More generally, this method is relevant in order to deal with coupled systems of nonlinear wave and Klein-Gordon equations. Our main task in the present paper is to formulate the $f(R)$-gravity field equations in the form of a second-order system of coupled partial differential equations, investigate the differential structure of this system, and control its nonlinearities in suitably weighted Sobolev norms. 

In seeking global existence and stability results, a primary challenge lies in the coupling between wave equations and Klein-Gordon equations, for which quite different arguments of proof were proposed in the literature. The partial differential equations under consideration are nonlinear in nature, and global existence results can only be achieved if phenomena such as formation of singularities, and gravitational collapse or coordinate singularities in the context of gravity, are avoided. A major observation is that, in the evolution in the vicinity of the Minkowski spacetime, sufficiently strong {\sl dispersion effects} take place, so that global-in-time solutions exist. The Euclidean-hyperboloidal foliation method provides one with the setting in which to address these issues and, in turn, prove  (almost) {\sl sharp decay estimates} in timelike, null, and spacelike directions. 
 
While Einstein's theory goes back to 1915, the $f(R)$-theory was formulated much later in 1970 by Buchdahl \cite{Buchdahl} and attracted a lot of attention by physicists; cf.~\cite{Bhattacharyya-2022,BENO,CNO,CNOT,Garcia,MagnanoSoko,Nojiri,SalgadoMartinez,Sotiriou-2010,Woodard-2007} and the earlier work by Brans and Dicke~\cite{BransDicke} and Starobinsky~\cite{Starobinsky,Starobinsky-1980}. In this theory, the integrand of the Hilbert-Einstein functional, i.e.~the spacetime scalar curvature $R$  is replaced by a nonlinear function $f(R)$ (cf.~\eqref{equa:action1}, below).  We especially refer the reader to Nojiri et al.~\cite{Nojiri} and Sotiriou and Faraoni~\cite{Sotiriou-2010} for a review from a physics perspective, as well as LeFloch and Ma's monograph~\cite{PLF-YM-memoir} for the local-in-time theory for the Cauchy problem. 


\paragraph{Global stability of Minkowski spacetime.}

We are thus interested in the global future evolution from a spacelike hypersurface which is assumed to be a small perturbation of an asymptotically flat and spacelike slice of Minkowski spacetime. Suitable data are imposed on such a hypersurface and we seek to prove that these small initial disturbances decay in time,  
so that the curved spacetime geometry eventually converges to the (flat) Minkowski geometry.
 This problem for the vacuum Einstein equations was first solved by Christodoulou and Klainerman \cite{CK}. Lindblad and Rodnianski later provided an alternative proof using wave coordinates \cite{LR1,LR2}. Bieri addressed solutions with reduced decay at spacelike infinity \cite{Bieri,BieriZipser}, and, most recently, further important advances are due to Hintz and Vasy \cite{HintzVasy1,HintzVasy2} as well as Shen~\cite{Shen}. The techniques of proof in these papers apply to \emph{vacuum} spacetimes as well \emph{massless} matter fields.

Interestingly, the global evolution of \emph{self-gravitating massive fields} was tackled by mathematicians only in recent years, and there is now a growing body of literature on this subject. Regarding Einstein's field equations, the nonlinear stability of Minkowski spacetime in presence of a Klein-Gordon scalar field was achieved first for perturbations that coincide with the Schwarzschild spacetime outside a (large) light cone determined from the initial hypersurface and ``propagating'' within the spacetime. To this end, LeFloch and Ma~\cite{PLF-YM-book,PLF-YM-two} introduced a new methodology (the hyperboloidal foliation method) based on wave gauge, while in simultaneous and independent work Wang~\cite{Wang} introduced a fully geometric approach. It is only more recently that the resolution of the nonlinear stability problem for unrestricted perturbations 
was achieved in two independent works by Ionescu and Pausader~\cite{IP-two, IP3} and by LeFloch and Ma~\cite{PLF-YM-main}. 


\paragraph{Related works on hyperboloidal foliations.} 

The use of hyperbolic hypersurfaces for the analysis of partial differential equations goes back to Klainerman~\cite{Klainerman85}, H\"ormander~\cite{Hormander}, and Tataru~\cite{Tataru96}. In the context of general relativity, we refer to earlier work by Friedrich~\cite{Friedrich81, Friedrich83}. Concerning the evolution of various models of matter fields, we also refer to the works~\cite{HuneauStingo, IfrimStingo,KauffmanLindblad}. Kinetic equations were investigated by Fajman, Joudioux, and Smulevici \cite{FJS,FJS3} and Bigorgne et al. \cite{Bigorgne,Bigorgne2}. For the massive Maxwell Klein-Gordon system, see~\cite{KlainermanWangYang}. Importantly, this strategy is also a very active domain of research in numerical relativity, pursued by Moncrief and Rinne \cite{MoncriefRinne}, Zenginoglu~\cite{Zenginoglu,Zenginoglu-5}, 
Va\~n\'o-Vi\~nuales et al. \cite{Vano1,Vano2,Vano3}, Gasperini et al. \cite{Gasperin}, Gautam et al.~\cite{Gautam} and followers. 
 
We point out that our use of ``hyperboloidal hypersurfaces'' in \cite{PLF-YM-book}--\cite{PLF-YM-model} and in the present paper is very different from the one in Helmut Friedrich's hyperboloidal framework~\cite{Friedrich81,Friedrich83}, as we are  motivated by a different perspective. While the notion of ``hyperboloidal hypersurface'' in~\cite{Friedrich81,Friedrich83} is determined by an asymptotic condition (and, in the interior of the spacetime, hyperboloidal surfaces are nothing but ordinary spacelike surfaces), in the present paper the term "hyperboloidal" is used in the elementary sense that in the coordinates under consideration, the (hyperboloidal) ``interior domain'' $\{r<t-1\}$ is foliated by (a subset of) the standard hyperboloids $\{t^2-r^2=s^2\}$
of the Minkowski spacetime. Moreover, in an ``exterior domain'' these slices are merged together with a standard constant $t$-foliation of Minkowski spacetime. As far as the constructed spacetime is concerned, our construction produces “asymptotically Euclidean” hypersurfaces, understood in a standard sense that is in use in general relativity. We find it convenient to refer to the interior foliation as being "asymptotically hyperboloidal''. We emphasize that 
our slices in the interior domain are all asymptotic to the {\sl same cone.} Our foliation serves to analyze the decay in time-like and space-like directions, while Helmut Friedrich's hyperboloidal framework addresses the problem of the structure of the spacetimes along null infinity. 


\paragraph{Outline of this paper.}

This paper is organized as follows. In Section~\ref{section=3}, we present the field equations of $f(R)$ gravity which we express first in the so-called ``Jordan metric'' and then recast in the ``Einstein conformal metric'', as we call here it, which is defined by a rescaling based on the scalar curvature. In Section~\ref{section=4}, we formulate the initial value problem and then present our main result, stated in a geometric form. In Section~\ref{section=5}, we define the Euclidean-Hyperboloidal foliation of interest, together with a fundamental weighted energy inequality. In~Section~\ref{section=6}, we are in a position to formulate the $f(R)$-gravity model in conformal wave gauge in which the scalar curvature is then regarded as an independent field. Next, in Section~\ref{section=7} we reformulate our main stability result in conformal wave gauge, in a form that is directly amenable to our Euclidean-Hyperboloidal foliation method. In Section~\ref{section-Euclidean}, we present our analysis in the Euclidean-merging domain and finally, in~Section~\ref{section-hyperb}, our analysis in the hyperboloidal domain. Technical identities are collected in Appendix~\ref{appendix-A} (conformal identities) and Appendix~\ref{appendix-B} (structural properties).  


\section{Structure of the f(R) gravity equations} 
\label{section=3}

\subsection{Formulation in the Jordan metric}

\paragraph{Aim of this section.}

Our first task is to present the equations of interest, which are stated first with a metric denoted by $\giveng$ and referred to as the {\sl Jordan metric,} and next in a conformal metric denoted by $\gconf$ and referred to as the {\sl Einstein conformal metric.} In the latter, we introduce a suitable rescaling of the scalar curvature, which is then viewed as an independent unknown, which we propose to call the {\sl effective curvature field}. To facilitate comparison later on, we first recall the formulation of Einstein's gravity equations. From the metric $\giveng$, we define geometric objects such as the Ricci curvature by conveniently including a tilde symbol in our notation. The equations of interest will next be stated in term of the conformal metric $\gconf$.


\paragraph{Einstein's gravity equations.}

We are interested in a four-dimen\-sional spacetime $(\Mcal, \giveng)$, that is, a manifold  $\Mcal$ (which we are going to take as $\Mcal \simeq [0, +\infty) \times \RR^3$) endowed with a Lorentzian metric $\giveng$ with signature $(-, +, +, +)$. The  Levi-Civita connection of  $\giveng$ is denoted by $\givennabla = \nabla_\giveng$, from which we define the Ricci tensor and scalar curvature of the spacetime denoted by $\givenRic$ and $\givenScal$, respectively. We also write $\givenRic = (\givenR_{\alpha\beta})$ in components. Throughout, Greek indices describe $0, 1,2,3$ and we use the standard convention of implicit summation over repeated indices, as well as raising and lowering indices with respect to the metric $\giveng_{\alpha\beta}$ and its inverse denoted by $\giveng^{\alpha\beta}$. For instance, we write $X_\alpha = \giveng_{\alpha\beta} X^\beta$ for the duality between vectors and $1$-forms.  

The standard Hilbert-Einstein action is related to the integral of the scalar curvature, that is, 
\bel{equa:action1}
\int_\Mcal \Big( \givenScal + 16 \pi \, L[\Psi, \giveng] \Big) \, d\givenVol, 
\ee 
where $d\givenVol$ denotes the canonical volume form on $(\Mcal, \giveng)$ and the Lagrangian 
$L[\Psi, \giveng]$ describes the matter content and is specified later in this text.
Critical points of this action are well-known to satisfy {\sl Einstein's field equations} 
\bel{equa-1:einstein-massif}
\givenG_{\alpha\beta} = 8\pi \,\givenT_{\alpha\beta} \quad \text{ in } (\Mcal,\giveng), 
\ee
in which Einstein's curvature tensor (associated with $\giveng$) reads 
\bel{equa:Eins9}
\givenG_{\alpha\beta} := \givenR_{\alpha\beta} - {1 \over 2} \givenScal \, \giveng_{\alpha\beta}.
\ee 

Furthermore, the Lagrangian $L[\Psi, \giveng]$ in \eqref{equa:action1} models the matter content of the spacetime and may involve a collection $\Psi$ of unknown fields. This Lagrangian allows us to determine the right-hand side of \eqref{equa-1:einstein-massif}, namely the {\sl energy-momentum tensor} 
\be
\givenT_{\alpha\beta}
:= - 2 \, {\delta L \over \delta \giveng^{\alpha\beta}} [\Psi,\giveng]  + \giveng_{\alpha\beta} \, L[\Psi,\giveng].
\ee
In addition, the evolution equations for the matter fields are given by the twice-contracted Bianchi identities 
$\givennabla^\beta \givenG_{\alpha\beta} = \nabla^\beta \givenR_{\alpha\beta} - {1 \over 2} \givennabla_\alpha \givenR$,  
which are equivalent to saying 
\begin{equation}\label{equa-evolution-matter} 
\givennabla^\beta \givenT_{\alpha\beta}= 0  \quad \text{ in } (\Mcal,\giveng). 
\end{equation}


\paragraph{Modified gravity equations.}

The modified theory of gravity we are interested in is the generalization of Einstein's theory that is based on the action 
\bel{equa:action} 
\int_\Mcal \Big( f(\givenScal) + 16 \pi \, L[\Psi, \giveng] \Big) \, d\givenVol,
\ee
in which $f: \RR \to \RR$ is a given (sufficiently regular) function. In order for the modified theory to be a formal extension of Einstein's theory, we assume that $f(R) \simeq R$ in the zero curvature limit $R \to 0$ which will be the regime studied in this paper. For instance, the choice of the quadratic function $R+ {\coeffi \over 2} R^2$ was proposed by Starobinsky \cite{Starobinsky} for the description of the early universe and is the main example of interest. 

Critical points of the action \eqref{equa:action} are easily found to satisfy the {\sl field equations of $f(R)$ gravity}
\bel{equa--FieldEquations}
\givenCurv_{\alpha\beta} = 8 \pi \, \givenT_{\alpha\beta}[F,\giveng] \quad \text{ in } (\Mcal,\giveng), 
\ee
in which the {\sl $f(R)$-gravity tensor} is defined as 
\bel{equa--ModifiedCurv}
\givenCurv_{\alpha\beta} :=
f'(\givenR) \, \givenG_{\alpha\beta} -  {1 \over 2}  \Big( f(\givenR) - \givenR f'(\givenR) \Big) \giveng_{\alpha\beta}
+ \Big( \giveng_{\alpha\beta} \, \givenBox   - \givennabla_\alpha \givennabla_\beta \Big) \big( f'(\givenR) \big).
\ee
Here, $\givenBox := \givennabla_\alpha \givennabla^\alpha$ is the wave operator associated with the metric $\giveng$, while $\givenG$ denotes Einstein's curvature tensor \eqref{equa:Eins9}. Observe that, in contrast with $\givenG$ which involves up to second-order derivatives of the metric,  the tensor $\givenCurv$ involves up to {\sl fourth-order derivatives} of the spacetime metric.

Interestingly, by relying again on the twice-contracted Bianchi identities together with the differential structure of the tensor $\givenCurv$, the modified gravity tensor $\givenCurv$ is also found to be {\sl divergence-free,} that is,
\bel{equa:1112} 
\givennabla^\alpha \givenCurv_{\alpha\beta} = 0 \quad \text{ in } (\Mcal,\giveng). 
\ee
Consequently, the evolution equations for the matter field still take the form \eqref{equa-evolution-matter}. 


\paragraph{Evolution of the spacetime scalar curvature.}

In view of the expression \eqref{equa--ModifiedCurv} of the tensor $\givenCurv$, which depends upon the Hessian of the scalar curvature, it is natural to seek an evolution equation for $\givenScal$. Namely, taking the trace of the curvature tensor \eqref{equa--ModifiedCurv} we obtain 
\bel{equa--trace-gravity}
\Tr_{\giveng}(\givenCurv) = 3 \, \givenBox f'(\givenR) + \big(  - 2 f(\givenR) + \givenR f'(\givenR) \big). 
\ee
Consequently, after taking the trace of the field equations \eqref{equa--FieldEquations}  we find 
$\Tr_{\giveng}(\givenCurv) = 8\pi \, \Tr_\giveng(\givenT)$, therefore  
\bel{Eq1-15}
\aligned
3 \, \givenBox f'(\givenR) + \big(  - 2 f(\givenR) + \givenR f'(\givenR) \big)
= 8 \pi \, \giveng^{\alpha\beta} \givenT_{\alpha\beta}
\quad \text{ in } (\Mcal,\giveng). 
\endaligned
\ee
This is a second-order evolution equation of Klein-Gordon type satisfied by the spacetime curvature $\givenR$. It plays a central role in our analysis.


\subsection{Formulation in the Einstein conformal metric}

\paragraph{Conditions on the function $f(R)$.}

We now rewrite the equations of $f(R)$-gravity in a form that will be amenable to techniques of geometric analysis. We are going to introduce a new metric that is conformally equivalent to $\giveng$, and then  
formulate a second-order hyperbolic system
 instead of a fourth-order system in $\giveng$. Moreover, in order for our presentation to encompass the formal limit $f(R) \to R$ (and eventually establish uniform estimates in this limit), it is convenient to assume
\bel{Equa--condi-coeffi} 
f(0) = 0 
 \qquad 
 f'(0) = 1, 
 \qquad
f''(0) =  \coeffi,
\qquad 
\sup_r |f^{(j)}(r)| \leq C \coeffi^{j-1},\quad  j \geq 3, 
\ee 
in which $\coeffi>0$ that can approach zero and $C>0$ is a fixed constant.  The supremum is taken over a small neighborhood of the origin and for sufficiently large order $j \geq 2$. 
In particular, we use that $f'(r) \simeq 1$ and thus does not vanish for all sufficiently small $r$, which is a consequence of \eqref{Equa--condi-coeffi} for the regime of small curvature under consideration in the present paper. 

We emphasize that the positivity of $\coeffi >0$ is standard in the physics literature (as it can be guaranteed by the strong energy condition) and is also essential for the global stability of solutions, while it is irrelevant as far the local-in-time existence theory is concerned. The condition \eqref{Equa--condi-coeffi} is sufficient in order to proceed with the rewriting of the field equations, but later on we will reformulate it in a form that is more suitable for our global  stability theory. We rely on the change of variable $r \mapsto f'(r)$ and, more precisely, we rescale the scalar curvature variable by defining the function 
\be
\hfrak (s) := \Big({1 \over\coeffi} \log (f') \Big)^{-1}(s), \quad s \in \RR, 
\ee
in which the symbol $^{-1}$ stands here for the inverse of a function. It is convenient to introduce the {\sl potential} 
\be
U(r) : = \frac{r \, f'(r) - f(r)}{f''(0)(f'(r))^2}, \quad r \in \RR. 
\ee
Under the change of variable $s \mapsto \hfrak(s)$, the function $f$ and the potential $U$ can both be transformed and we write  
\be
\ffrak := f\circ\hfrak, 
\quad 
\Ufrak :=  U\circ\hfrak.
\ee

For instance, the Starobinsky model corresponds to the choice $f(r) = r + (\coeffi/2) r^2$ for some $\coeffi>0$, in which case $f'(r) = 1 + \coeffi \, r$ and $U(r) = \frac{r^2}{2(1+\coeffi r)^2}$. Consequently, in this case we have 
\begin{equation}\label{eq1-09-Mars-2023}
\hfrak(s) = {1 \over \coeffi} (e^{\coeffi s} - 1),
\quad
\Ufrak(s) = \frac{1}{2 \coeffi^2 e^{2 \coeffi s}}(e^{\coeffi s} - 1)^2
\quad 
\text{ (Starobinsky model).} 
\end{equation}
Clearly, when $s \to 0$ we have $\hfrak(s) \simeq s$ and $\Ufrak(s) \simeq  (1/2) s^2$. 


\paragraph{A new metric.}

We now introduce a new metric together with a new field $\phi$ whose normalization is chosen so that, in view of \eqref{Equa--condi-coeffi}, $\phi$ approaches the scalar curvature in the limit $\coeffi \to 0$. 

\begin{definition}
The {\bf Einstein conformal metric} $\gconf$  associated with a solution $(\Mcal, \giveng)$ to the  $f(R)$-field equations \eqref{equa--FieldEquations} 
is defined by  
\bel{equa:conformgdef}
\gconf_{\alpha\beta} =: e^{\coeffi \phi} \, \giveng_{\alpha\beta}, 
\qquad 
\phi :=\hfrak^{-1}(\givenR) =  {1 \over \coeffi} \ln(f'(\givenR)), 
\ee
in which the conformal factor $\phi: \Mcal \to \RR$ is expressed in terms of the scalar curvature $\givenScal$ and is referred to as the {\bf effective curvature field.} 
 \end{definition}

The motivations for the above terminology will become clear shortly, after we complete the transformation of the modified gravity equations. We emphasize that some straightforward calculations are postponed to Appendix~\ref{appendix-A}.  First of all, thanks to the notation \eqref{equa:conformgdef}, the gravity tensor $\givenCurv_{\alpha\beta}$ and its trace, as defined in \eqref{equa--ModifiedCurv} and \eqref{equa--trace-gravity}, now take the form (cf.~Appendix~\ref{appendix-A}) 
\begin{equation}\label{eq1-12-nov-2023}
\givenCurv_{\alpha\beta} - \frac{1}{2}\Tr_g(\givenCurv)g_{\alpha\beta} 
= e^{\coeffi \phi}R_{\alpha\beta} - {3  \over 2} \coeffi^2 e^{\coeffi\phi}\nabla_\alpha \phi \nabla_\beta \phi 
- {\coeffi \over 2} \, g_{\alpha\beta}e^{\coeffi\phi}\Ufrak(\phi),
\end{equation}
\begin{equation}\label{eq2-10-mars-2023}
\Tr_{\giveng}(\givenCurv) 
= 3 \coeffi \, e^{2 \coeffi \phi}\Box_g\phi - \ffrak(\phi) + \coeffi e^{2 \coeffi \phi}\Ufrak(\phi). 
\end{equation}
Hence, following by the physics literature (as overviewed in the introduction), it is natural to introduce the stress-energy tensor associated with the field $\phi$, that is, the {\bf effective tensor} 
\begin{equation}\label{eq3-26-mars-2023}
T^\eff[\phi,g]_{\alpha\beta} 
:= {3 \over 2} \coeffi^2 \nabla_{\alpha}\phi\nabla_\beta \phi - {3 \over 4} \coeffi^2 g(\nabla\phi,\nabla\phi)g_{\alpha\beta} - \frac{\coeffi}{2} \Ufrak(\phi) \, g_{\alpha\beta}. 
\end{equation}
Consequently, the relation \eqref{eq1-12-nov-2023} reads 
\begin{equation}\label{eq2-12-nov-2023}
G_{\alpha\beta} = e^{-\kappa\phi}\givenCurv_{\alpha\beta} + T^\eff[\phi,g]_{\alpha\beta}
\end{equation} 
and the field equation of the modified gravity is thus found to be equivalent to 
\begin{equation}\label{eq4-26-mars-2023}
G_{\alpha\beta} = e^{- \coeffi\phi} 8 \pi \, \givenT_{\alpha\beta}[\Psi,\giveng]+ T^\eff[\phi,g]_{\alpha\beta}
\qquad \text{ in } (\Mcal,g).
\end{equation}
Here, $G_{\alpha\beta}$ is the Einstein gravity tensor associated with the new metric $g$. In other words, by introducing the additional field $\phi$ (which is determined by the Jordan metric $\giveng$) we have arrived at a system that is formally identical to the Einstein-massive field system, in which the scalar curvature is regarded as an ``effective'' curvature field, which looks formally like a real massive scalar field.  


\paragraph{Evolution of the matter field.}

As stated in~\eqref{eq2-26-mars-2023} in Lemma~\ref{appendix-propositionA1}, with the property \eqref{equa:1112}, we have 
\be
\nabla^{\alpha}\givenCurv_{\alpha\beta} 
=  \coeffi \, \givenCurv_{\mu\beta} \, \nabla^{\mu}\phi - {\coeffi \over 2} \Tr_g(\givenCurv)\nabla_\beta \phi.
\ee
Together with the coupling $\givenCurv_{\alpha\beta} = 8\pi \givenT_{\alpha\beta}$, we obtain the evolution of the matter field:
\begin{equation}\label{eq-matter-general}
\nabla^{\alpha}\givenT_{\alpha\beta} = \coeffi \givenT_{\mu\beta}\nabla^{\mu}\phi - \frac{\kappa}{2}\Tr_g(\givenT)\nabla_{\beta}\phi.
\end{equation}


\paragraph{Evolution of the effective curvature field.}

By taking the divergence of \eqref{eq2-12-nov-2023} with respect to the conformal metric $g$ and using \eqref{eq2-26-mars-2023} and the observation that $G_{\alpha\beta}$ is divergence-free with respect to the metric $g$, we find 
\begin{equation}
\nabla^{\alpha} \big( T^\eff[\phi,g]_{\alpha\beta}\big) = {\kappa \over 2}e^{-\coeffi\phi} \Tr_g(\givenCurv)\nabla_\beta \phi. 
\end{equation}
Recalling \eqref{eq3-26-mars-2023}, for the wave operator $\Box_g $ associated with the metric $g$ we deduce that, for the above vector-valued equation to hold, it is sufficient to solve the scalar equation 
\begin{equation}\label{eq5-26-mars-2023}
3 \coeffi \, \Box_g \phi -  \Ufrak'(\phi) 
= e^{-\coeffi\phi}\Tr_g(\givenCurv). 
\end{equation}
In combination with the relation $\Tr_g(\givenCurv) = 8\pi \, \Tr_g(\givenT)$, this gives us the evolution law for the effective field $\phi$ and find 
\begin{equation}\label{equa--effective-phi-equation}
3 \coeffi \, \Box_g \phi -  \Ufrak'(\phi) 
= 8 \pi e^{-\coeffi\phi}\, \Tr_g(\givenT).
\end{equation}
Recalling \eqref{eq2-10-mars-2023}, it may seem to be strange that we have got two evolution equations on the same scalar field. However, in Appendix~\ref{appendix-A}, we show the equivalence between them.

If we suppress the relation $f'(\givenR) = e^{\coeffi\phi}$ and regard $\phi$ as an {\sl independent unknown variable,} then \eqref{eq4-26-mars-2023} and \eqref{equa--effective-phi-equation} together with \eqref{equa-evolution-matter} form a Einstein-scalar-matter coupling system. Interestingly, if \eqref{eq4-26-mars-2023} holds, then $f'(\givenR) = e^{\coeffi\phi}$ holds automatically. This is known in the physics literature and was rigorously proved in \cite{PLF-YM-one} in a slightly different context. Here, we revisit this property, as follows. We point out that \eqref{eq3c-12-nov-2013} is not directly used in the proof of the following property. 

\begin{proposition}[Derivation of the curvature constraint]
\label{prop1-12-nov-2023}
Suppose that $(\Mcal,g,\phi,\givenT)$ is a sufficiently regular solution to the system
\begin{subequations}\label{eq3-12-nov-2013}
\begin{equation}\label{eq3a-12-nov-2013}
G_{\alpha\beta} = e^{- \coeffi\phi} 8 \pi \, \givenT_{\alpha\beta}[\Psi,\giveng]+ T^\eff[\phi,g]_{\alpha\beta}
\qquad \text{ in } (\Mcal,g), 
\end{equation}
\begin{equation}\label{eq3b-12-nov-2013}
3 \coeffi \, \Box_g \phi -  \Ufrak'(\phi) 
= 8\pi e^{-\coeffi\phi}\,\Tr_g(\givenT)
\qquad \text{ in } (\Mcal,g), 
\end{equation}
\begin{equation}\label{eq3c-12-nov-2013}
\nabla^{\alpha}\givenT_{\alpha\beta} = \coeffi \givenT_{\mu\beta}\nabla^{\mu}\phi - \frac{\kappa}{2}\Tr_g(\givenT)\nabla_{\beta}\phi
\qquad \text{ in } (\Mcal,g), 
\end{equation}
\end{subequations}
then the scalar curvature $\givenR$ of the metric $\giveng_{\alpha\beta}:= e^{-\coeffi\phi}g_{\alpha\beta}$ is computed 
from the field $\phi$ explicitly by  
\begin{equation}\label{eq4-12-nov-2013}
f'(\givenR) = e^{\coeffi\phi}. 
\end{equation} 
\end{proposition}

\begin{proof} 
\bse
We observe first that \eqref{eq3a-12-nov-2013} is equivalent to saying 
\begin{equation}\label{eq1-03-dec-2023}
R_{\alpha\beta} - \frac{3}{2}\coeffi^2\nabla_{\alpha}\phi\nabla_{\beta}\phi - \frac{\coeffi}{2}\Ufrak(\phi)g_{\alpha\beta}
= 8\pi e^{-\coeffi\phi}\Big(\givenT_{\alpha\beta} - \frac{1}{2} \Tr_g(\givenT)g_{\alpha\beta}\Big)
\end{equation}
and, in view of the conformal transformation \eqref{eq4-12-nov-2023}, we have 
$$
\aligned
&\givenR_{\alpha\beta} - \kappa\givennabla_{\alpha}\givennabla_{\beta}\phi - \frac{\coeffi}{2}g_{\alpha\beta}\Box_{\giveng}\phi
- \coeffi^2\givennabla_{\alpha}\phi\givennabla_{\beta}\phi - \frac{\coeffi^2}{2}\giveng(\givennabla \phi,\givennabla \phi)\giveng_{\alpha\beta}
- \frac{\coeffi}{2}\Ufrak(\phi)\giveng_{\alpha\beta}
\\
& = 
8\pi e^{-\coeffi\phi}\big(\givenT_{\alpha\beta} - \frac{1}{2}\Tr_g(\givenT)g_{\alpha\beta}\big).
\endaligned
$$
Taking the trace of the above equation with respect to the metric $\giveng$, we obtain
\begin{equation}\label{eq5-12-nov-2023}
-8\pi\Tr_g(\givenT) = \givenR - 3\coeffi \, \Box_{\giveng}\phi - 3\coeffi^2\giveng(\givennabla\phi,\givennabla\phi) - 2\coeffi e^{\coeffi\phi} \Ufrak(\phi).
\end{equation}
On the other hand, \eqref{eq3b-12-nov-2013} together with \eqref{eq1-28-oct-2023} leads us to
\bel{equa-Tra} 
8\pi e^{-\coeffi\phi} \Tr_g(\givenT) = 2\coeffi e^{-\coeffi\phi}\Box_{\giveng}\phi + 3\coeffi^2 e^{-\coeffi\phi}\giveng(\givennabla\phi,\givennabla\phi) - \Ufrak'(\phi).
\ee
Comparing \eqref{eq5-12-nov-2023} with \eqref{equa-Tra}, we deduce the following purely algebraic relation between $\givenR$ and $\phi$: 
\begin{equation}\label{eq7-12-nov-2023}
e^{-\coeffi\phi}\givenR = -2\coeffi \Ufrak(\phi) + \Ufrak'(\phi).
\end{equation}
Next, by writing  $r = \hfrak(\phi)$, we have 
$$
e^{\coeffi\phi} = f'(r), 
\qquad   
\hfrak'(\phi) = \frac{dr}{d\phi} = \frac{\coeffi f'(r)}{f''(r)}
$$
and, consequently,
\be
\Ufrak'(\phi) = \frac{2f(r) - rf'(r)}{f''(r)}.
\ee
Finally, in view of the relation \eqref{eq7-12-nov-2023}, we deduce that the condition $\givenR = r$ implies $f'(\givenR) = e^{\coeffi\phi}$, as claimed. 
\ese
\end{proof}


\paragraph{Summary.}

We summarize our standpoint as follows. 

\begin{definition} The {\bf Einstein conformal system} of $f(R)$-modified gravity for self-gravitating matter consists of 

\begin{itemize} 

\item the field equations \eqref{eq3a-12-nov-2013} satisfied by the Einstein conformal metric $g$, 

\item the Klein-Gordon equation \eqref{eq3b-12-nov-2013} for the effective curvature field $\phi$, 

\item and the equations \eqref{eq3c-12-nov-2013} for the matter.

\end{itemize}
\end{definition}

Our definition is motivated by the fact (Proposition~\ref{prop1-12-nov-2023}) that the relation between $\phi$ to the spacetime scalar curvature $\givenR$ of the given metric 
$\giveng_{\alpha\beta} = e^{- \coeffi \phi} \gconf_{\alpha\beta}$, namely
\bel{equa:rhocondi} 
\phi = \ln \bigl( (f'(\givenR))^{1/\coeffi} \bigr), 
\ee
is a consequence of our conditions. All of the forthcoming notions and results in this paper will be stated in terms of the conformal metric $g$ by regarding the effective field $\phi$ as an independent unknown. Our results could be easily translated in terms of the given metric $\giveng$. 
 

\subsection{Further structure}

\paragraph{Evolution of a Klein-Gordon matter field.}
 
Specifically, in this paper we include the coupling with a scalar field $\psi: \Mcal \to \RR$ governed by the energy-momentum tensor 
\begin{equation}\label{equa:Talphabeta}
\givenT = \nabla\psi\otimes\nabla\psi - \frac{1}{2}(\giveng(\nabla\psi,\nabla\psi) + V(\psi))g. 
\end{equation}
In general, the potential $V = V(\psi)$ is a prescribed function depending on the nature of the matter under consideration. From \eqref{equa-evolution-matter} together with \eqref{equa:Talphabeta}, it then follows that the field $\psi$ satisfies the nonlinear Klein-Gordon equation (or wave equation if $V$ vanishes)
\bel{eq-KGG}
\givenBox \psi - V'(\psi) = 0 \quad \text{ in } (\Mcal,\giveng), 
\ee
which determines the evolution of the matter (under suitable initial data). For this model \eqref{equa:Talphabeta}, using the conformal metric $g$ we obtain  
\begin{equation}\label{equa-matter-conform}
\Box_g \psi - e^{-\coeffi\phi}V'(\psi) = \kappa \, g(\nabla \phi,\nabla\psi), 
\end{equation}
which, regarded as equation for the evolution of the matter field $\psi$, is a semi-linear Klein-Gordon equation.  


\paragraph{Remarks on the general relativity limit.}

While, in the present paper, we will not rigorously address the convergence problem $f(R) \to R$, it is interesting here to formally analyze this limit. In the modified gravity theory, Einstein equations are straightforwardly recovered by choosing the function $f(R)$ to be the linear function $R$. In the limit $\coeffi \to 0$ in \eqref{Equa--condi-coeffi}, we find $f(R) \to R$ and the field equations \eqref{equa--FieldEquations}-\eqref{equa--ModifiedCurv} formally converge to the Einstein equations \eqref{equa-1:einstein-massif}. Importantly, the limit $f(R) \to R$ is {\sl singular} in nature, since the classical gravity equations \eqref{equa-1:einstein-massif} involve up to {\sl second-order derivatives} of the metric $g$, while the modified gravity equations \eqref{equa--FieldEquations}-\eqref{equa--ModifiedCurv} contain {\sl fourth-order derivatives}. 

In terms of the Einstein conformal metric the scalar curvature field $\phi$ satisfies \eqref{equa--effective-phi-equation}, which is a Klein-Gordon equation and can be put in the form 
\be
\Ufrak'(\phi) + 8 \pi \, \Tr_\giveng(\givenT) =  3 \coeffi \, \Box_g \phi \to 0 \, \text{ in the limit } \kappa \to 0. 
\ee
In this formal limit, the Klein-Gordon field $\phi$ converges to the solution of an {\sl algebraic equation,} namely 
\be
\Ufrak'(\phi) + 8 \pi \, \Tr_\giveng(\givenT) = 0 \, \text{ when } \kappa = 0. 
\ee 
We may refer to this general relativity limit problem as an {\sl infinite mass problem} for the Klein-Gordon operator, and such an analysis requires an asymptotic analysis of the solutions to the differential operator $\Box_g \phi - \Ufrak'(\phi)/3 \coeffi$ (in the limit $\coeffi \to 0$). In the present paper we will derive estimates that are {\sl uniform} in $\kappa$, yet do not allow one to take this limit. Proving these convergence properties for solutions to the Cauchy problem in the Minkowski regime is beyond the scope of this paper and will be addressed in future work. 
 

\section{Global evolution in the theory of f(R) gravity}
\label{section=4}

\subsection{Geometric formulation of the initial value problem}

\paragraph{The f(R)-constraint equations.} 

From this point onwards, we exclusively focus on the Einstein conformal system of $f(R)$ gravity. Reformulating our findings in the Jordan metric would be a straightforward task, which we omit. As the evolution of the effective field $\phi$ is driven by the second-order evolution equation \eqref{equa--effective-phi-equation}, we need to specify, in addition to prescribing the intrinsic and extrinsic geometry and matter content, two initial data for $\phi$ in order to formulate the initial value problem. This means that we need to provide the spacetime scalar curvature and its Lie derivative in time. Furthermore, for definiteness, at this stage we consider the coupling with a scalar field $\psi$ governed by a potential $V=V(\psi)$. For further information on basic concepts, we refer the reader to~\cite{YCB-book}.

First of all, we recall the following Gauss-Godazzi equations, which hold for any hypersurface in spacetime: 
\be
\aligned
&R_{\gconf_0} + (\Tr_{g_0} (k_0))^2 - |k|^2_{g_0} = - 2 G_0^0,
\\
&\nablab_bK^b_a   - \nablab_a K_b^b = - N^{-1}R_{0a}, 
\endaligned 
\ee
written in the so-called Cauchy adapted frame.  Here,$e_a = \del_a$ with $\{x^a\}$ a local chart of the initial slice $\Mcal_0$, and 
\be
e_0 = \del_t - \beta^a\del_a,\quad \text{ such that } e_0\perp \Mcal_0, 
\ee
and $N$ is referred to as the lapse function, which, by definition, equals $g_{00} = -N^2$.

Then, we write $G_0^0 = G_{0\alpha}g^{\alpha0} = G_{00}g^{00} = -N^{-2}G_{00}$ and we recall \eqref{eq3a-12-nov-2013} in order to obtain
\be
G_0^0 = -8\pi e^{-\coeffi\phi}N^{-2} \givenT[\psi,\giveng]_{\alpha\beta} - N^{-2}T^{\eff}_{00}, 
\ee
which leads us to
\be
\aligned
G_0^0 & =  - 4\pi e^{-\coeffi\phi}\big(|N^{-1}\nabla_0\psi|^2 + g_0(\nabla \psi,\nabla\psi) + 2V(\psi)\big)
\\ 
&- \frac{3}{4}\coeffi^2 \big(|N^{-1}\nabla_0\phi|^2 + g_0(\nabla\phi,\nabla\phi) + \frac{2}{3\coeffi}\Ufrak(\phi)\big).
\endaligned
\ee
For the components $R_{0a}$ we observe that
\be
R_{\alpha\beta} = G_{\alpha\beta} - \frac{1}{2}Gg_{\alpha\beta}, 
\ee
which leads us to the relation $R_{0a} = G_{0a}$. Then, we have 
\be
R_{0a} = 8\pi e^{-\coeffi\phi} \nabla_0\psi\nabla_a\psi + \frac{3}{2}\coeffi^2\nabla_0\phi\nabla_a\phi. 
\ee
We arrive at the formulation of the {\bf constraint equations of $f(R)$-gravity}
\begin{equation}\label{eq1-14-nov-2023}
\aligned
R_{\gconf_0} + (\Tr_{g_0} (k_0))^2 - |k|^2_{g_0} & = 8\pi e^{-\coeffi\phi}\big(|\psi_1|^2 + g_0(\nabla\psi,\nabla\psi) + 2V(\psi)\big) 
\\
& \quad + \frac{3}{2}\coeffi^2\big(|\phi_1|^2 + g_0(\nabla\phi,\nabla\phi) + \frac{2}{3\coeffi}\Ufrak(\phi)\big),
\\
\nablab_bK^b_a   - \nablab_a K_b^b & =  -8\pi e^{-\coeffi\phi}\psi_1\nabla_a\psi - \frac{3}{2}\coeffi^2\phi_1\nabla_a\phi.
\endaligned
\end{equation}
For vacuum spacetimes, these equations become
\begin{equation}\label{eq2-14-nov-2023}
\aligned
R_{\gconf_0} + (\Tr_{g_0} (k_0))^2 - |k|^2_{g_0} 
& = \frac{3}{2}\coeffi^2\big(|\phi_1|^2 + g_0(\nabla\phi,\nabla\phi) + \frac{2}{3\coeffi}\Ufrak(\phi)\big),
\\
\nablab_bK^b_a   - \nablab_a K_b^b 
& =  - \frac{3}{2}\coeffi^2\phi_1\nabla_a\phi.
\endaligned
\end{equation}

\begin{definition}
\label{def36-22}
An {\bf initial data set for $f(R)$ gravity} consists of data $(\Mcal_0, \gconf_0, {\kd_0}, \phi_0, \phi_1, \psi_0, \psi_1)$ satisfying the following conditions. 

\begin{itemize}

\item[$\bullet$] $\Mcal_0$ is a $3-$dimensional manifold, endowed with a Riemannian metric $\gconf_0$ and a symmetric 2-covariant tensor field ${\kd_0}$.  

\item The scalar fields $\phi_0, \phi_1$ and $\psi_0, \psi_1$ prescribed on $\Mcal_0$ and associated with the effective  curvature field and the matter field, respectively.  

\item[$\bullet$] The  following {\rm Hamiltonian constraint} and {\rm momentum constraints} hold on $\Mcal_0$: 
\begin{equation}\label{consttransHan}
\aligned
R_{\gconf_0} + (\Tr_{g_0} (k_0))^2 - |k|^2_{g_0}
& = \mu^{\eff}_0 + \mu_0, 
\\
\nablab_b{k_0}^b_a   - \nablab_a {k_0}_b^b
& = J^{\eff}_a + J_a.
\endaligned
\end{equation}
Here, $R_{g_0}$ denotes the scalar curvature of $g_0$, while $\Tr_{g_0}$ and $\nablab$ denote the trace operator and connection operator associated with $g_0$, respectively. 

\item In \eqref{consttransHan}, the energy-momentum vectors $(\mu_0^f , J_0^f)$ and $(\mu_0^f , J_0^f)$ are defined on $\Mcal_0$ as 
$$
\aligned
\mu_0^{\eff} & := \frac{3}{2}\coeffi^2\big(|\phi_1|^2 + g_0(\nablab\phi,\nablab\phi) + \frac{2}{3\coeffi}\Ufrak(\phi)\big),
\\
J_a^{\eff} & := - \frac{3}{2}\coeffi^2\phi_1\nablab_a\phi, 
\endaligned
$$
and
$$
\aligned
\mu_0 & :=  8\pi e^{-\coeffi\phi}\big(|\psi_1|^2 + g_0(\nablab\psi,\nablab\psi) + 2V(\psi)\big) , 
\\
J_a & := -8\pi e^{-\coeffi\phi}\psi_1\nablab_a\psi.  
\endaligned
$$
\end{itemize}
\end{definition}


\paragraph{The $f(R)$-Cauchy developments.}

\begin{definition}
\label{def cauchy conf}
Given an initial data set $(\Mcal_0, \gconf_0, {\kd_0}, \phizerod, \phioned, \psi_0, \psi_1)$ as in Definition~\ref{def36-22}, the {\rm initial value problem associated with the conformal $f(R)$-gravity equations} consists of finding a Lorentzian manifold $(\Mcal, \gconf)$ (endowed with a time-orientation) and scalar fields $\phi, \psi$ defined in $\Mcal$ such that the following conditions hold. 
\begin{itemize}

\item The field equations of modified gravity \eqref{eq3-12-nov-2013} are satisfied.

\item There exists an embedding $i: (\Mcal_0, \gconf_0) \to (\Mcal, \gconf)$ with pull-back metric $i^\star \gconf = \gconf_0$ and second fundamental form $\kd_0$, hence one can embed the three-dimensional manifold $\Mcal_0$ as a hypersurface in the four-dimensional manifold $\Mcal$. 

\item The fields $\phi_0$ and $\psi_0$ coincide with the restrictions $\phi|_{\Mcal_0}$ and $\psi|_{\Mcal_0}$, respectively, while $\phi_1$ and $\psi_1$ coincide with the Lie derivatives $\mathcal{L}_{n} \phi|_{\Mcal_0}$ and $\mathcal{L}_{n} \psi|_{\Mcal_0}$, where $n$ denotes the future-oriented unit normal to $\Mcal_0$. 

\end{itemize}
\end{definition}

A solution satisfying the conditions in Definition~\ref{def cauchy conf} is referred to as a {\sl modified gravity Cauchy development} of the initial data set $(\Mcal_0, \gconf_0, {\kd_0}, \phi_0, \phi_1, \psi_0, \psi_1)$. The notion of {\sl maximally hyperbolic development} then follows by straightforwardly extending the definition in Choquet-Bruhat \cite{YCB-book}. In comparison with the classical gravity theory, the modified gravity theory has {\sl two extra degrees of freedom} specified from the two additional data $(\phi_0, \phi_1)$. Similarly as in classical gravity, these fields cannot be arbitrarily prescribed and suitable constraint equations must be assumed, as stated by \eqref{consttransHan}.

The following observations are in order. 

\begin{itemize} 

\item From Definition~\ref{def cauchy conf}, in the special case of vanishing data $\phi_0 = \phi_1= \psi_0= \psi_1 \equiv 0$ we recover the classical formulation of Einstein's vacuum equations.  

\item On the other hand, by taking a vanishing matter field $\psi\equiv 0$ but non-vanishing data $(\phi_0, \phi_1)$, the spacetimes in Definition~\ref{def cauchy conf} generally need not satisfy Einstein's vacuum equations. 

\end{itemize} 


\subsection{The reference metric} 

We follow the strategy presented in \cite{LR1, PLF-YM-main}, where the unknown metric is decomposed as 
\begin{equation}\label{eq2-23-dec-2023}
g_{\alpha\beta} = g_{\Mink,\alpha\beta} + h^{\star}_{\alpha\beta} + u_{\alpha\beta} = g^{\star}_{\alpha\beta} + u_{\alpha\beta},
\end{equation}
where $u_{\alpha\beta}$ is supposed to be small in an energy space (which will be called {\sl the perturbation}) while $h^{\star}$ is small in pointwise sense and will be taken as an Ansatz, and will be called {\sl the reference metric}. A typical choice taken in \cite{LR1} is
$$ 
h^{\star}_{\alpha\beta} = \chi(r/t)\chi(r)\frac{\eps}{r} \, g_{\Mink,\alpha\beta}
$$ 
where $\chi(s)$ is smooth cut-off function with constant value $1$ when $s\geq 3/4$ and $0$ when $s\leq 1/2$. This is essentially the Schwarzschild metric and the above formulation means that we are considering a finite (weighted) $L^2$ perturbation around the Schwarzschild spacetime.

For convenience of discussion, we recall some notation in our previous work \cite{PLF-YM-main}. We are working in (a subset of) $\RR^{1+3}$ with the standard Minkowski metric of signature $(-,+,+,+)$.  The Killing vector fields of the Minkowski space-time:
\bel{equa-ddd} 
\del_{\alpha} = \frac{\del}{\del x^{\alpha}},\quad L_a := t\del_a + x^a\del_t,\quad  
\Omega_{ab} = x_a\del_b - x_b\del_a.
\ee
A high-order differential operator composed by the above Killing vectors are called admissible. We denote by $\ord(Z)$ its order (in the sense of differential operator) and $\rank(Z)$ the number of boosts and rotations contained in it. Further, an operator $Z=\del^I L^J \Omega^K$ is called an {\bf ordered admissible operator.} 
To such an operator, we associate its {\bf order, degree,} and {\bf rank} by  
\begin{equation}
\ord(Z) = |I|+|J| + |K|, 
\quad 
\deg(Z) = |I|, 
\quad 
\rank(Z) = |J| + |K|,
\quad 
\text{ when } Z = \del^I L^J \Omega^K.
\end{equation}
An operator $\Gamma = \del^IL^J\Omega^KS^l$ is called an {\bf ordered conformal operator}, where $S = x^{\alpha}\del_{\alpha}$ is a {\sl conformal Killing vector} of the Minkowski spacetime. Its {\bf order}, {\bf degree}, and {\bf rank} are defined similarly:
$$
\ord(\Gamma) = |I|+|J| + |K| + l, 
\qquad 
\deg(\Gamma) = |I|, 
\qquad 
\rank(\Gamma) = |J| + |K| + l.
$$
If $u$ is a function defined in (an open subset of) $\RR^{1+3}$, we write  
$$
|u|_{p,k} = \sup_{\ord(Z)\leq p\atop \rank(Z)\leq k} |Zu|,\qquad |u|_N:=\sup_{\ord(Z)\leq N}|Zu|
$$
We also use $\del u, Lu$ or $\Omega u$ for any $\del_{\alpha}u$, $L_au$ or $\Omega_{ab}u$, and the notation $|\del u|,|Lu|_{p,k}, \ldots$ is then obvious.

In order to state minimal restrictions on our data, we may also localize the norms with respect to the \textsl{outgoing light cone} 
\begin{equation} \label{equa-lightcone}  
\qquad 
\Lscr:= \big\{ r = t-1 \big\} \subset \RR_+^{3+1}.
\end{equation} 
Its constant-$t$ slices are denoted by $\Lscr_t$.
In addition, a parameter $\ell \in (0,1/2]$ being fixed once for all, we introduce the {\sl near-light cone} domain
\begin{equation} \label{eq1-06-07-2021}
\Mscr^\near_{\ell}  := \Big\{  t \geq 2, \quad t-1 \leq r \leq \frac{t}{1-\ell} \Big\}, 
\end{equation} 
where, in agreement with our notation below, we restrict attention  to $t \geq 2$.

In the present work we continue to follow the notation in our companion paper \cite{PLF-YM-main}. The following notation makes sense for any metric $g^\star = g_\Mink + h^\star$ defined in $\RR^{1+3}$. We also introduce the {\sl reduced Ricci curvature} associated with $h^\star$ (which coincides with the Ricci curvature in the wave gauge of interest)  
\bel{equa-Riccistar} 
\aligned
\wR^\star_{\alpha\beta} 
& := R^\star_{\alpha\beta}
- \frac{1}{2} \big(\del_\alpha\Gamma^\star_\beta + \del_\beta\Gamma^\star_\alpha\big) 
- \frac{1}{2}
\Big( g^{\star\delta \delta'} \del_{\delta} g^\star_{\alpha \beta} \Gamma^\star_{\delta'} - \Gamma^\star_\alpha \Gamma^\star_\beta\Big)
\\
& =: - \gd^{\star \alpha'\beta'} \del_{\alpha'} \del_{\beta'} g^\star_{\alpha\beta}
+ \frac{1}{2} \Fbb_{\alpha\beta}(g^\star, g^\star;\del h^{\star}, \del h^{\star}), 
\endaligned
\end{equation}
in which $\Gamma^{\star \gamma} := g^{\star\alpha\beta} \Gamma_{\alpha\beta}^{\star\gamma}$ are the contracted Christoffel symbols in the global coordinate chart under consideration. 
 
 In particular, the following {\sl asymptotically Minkowski} conditions are assumed  
on $g^\star = \gMink + h^\star$:
\begin{equation}\label{equa-31-12-20}    
\aligned
&|h^{\star} |_{N+2} + \la r+t\ra|\del h^{\star} |_{N+1} + \la r+t\ra^2|\del\del h^{\star} |_{N} \lesssim \epss \la r+t\ra^{-\lambda},\quad &&\text{assymptotically flat}
\\
&\big|h_{00}^{\star}\big|_N + \big|h^{\star}_{0a}\big|_N + \big|(x^ax^b/r^2)h^{\star}_{ab} \big|_{N} 
\lesssim \epss\la r+t\ra^{-1+\theta}
&&\text{radial and tame decay,}
\\
&|\wR^{\star} |_{N} + \la r-t \ra|\del\wR^{\star} |_{N-1}
\lesssim
\begin{cases}
\epss^2 \, \la r+t\ra^{-2-2\lambda} \quad &\text{ in } \MME_{[s_0, + \infty)},
\\ 
\epss \la r+t\ra^{-2-\lambda} \quad &\text{ in } \MH_{[s_0, + \infty)},
\end{cases}
 && \text{almost Ricci flat,}
\endaligned
\end{equation}
with $3/4<\lambda<1, 0<\theta<1-\lambda$. Here, $s_0 \geq 2$ and we solve in the domain $ \M_{[s_0, + \infty)}$; we refer to Sectionl~\ref{label-spacetimefolmiation} for the presentation of the foliation that we use in our proof. 
This is covered by the {\sl Class A} in \cite{PLF-YM-main}. 
We also impose the light-bending condition (as we call it) 
\begin{equation}\label{eq1-23-dec-2023}
4\epss \leq \inf_{\Mscr^\near_{\ell}}  r \,  g^\star(\lbf, \lbf), 
\end{equation}
where $\lbf := \del_t - (x^a/r)\del_a$. This condition means that the reference metric has at its light cone contained in the standard cone for the Minkowski metric. Physically, this means that the photons are attracted by a positive mass towards the center.  

As a conclusion, we introduce the following terminology. 

\begin{definition}
A $(\lambda,\theta, \epss,N,\ell)-$ light-bending reference spacetime is a metric defined in $\RR_+^{1+3} = \{(t,x)| t\geq 2 \}$ satisfying \eqref{equa-31-12-20} and \eqref{eq1-23-dec-2023}.
\end{definition}


\paragraph{An example: merging the Minkowski and the Schwarzschild solutions.}

We consider the setup introduced for dealing with massive fields in \cite{PLF-YM-lambda1}. In wave coordinates the Schwarz\-schild metric, denoted by $g_{\Sch}$, reads 
$$ 
\aligned
g_{\Sch,00} & =   - \frac{r-m}{r+m},
\qquad g_{\Sch,0a} = 0,
\qquad 
g_{\Sch, ab} = \frac{r+m}{r-m} \omega_a \omega_b + \frac{(r+m)^2}{r^2}(\delta_{ab} - \omega_a \omega_b), 
\endaligned
$$
in which we have set $\omega_a := x_a/r$. We merge the Minkowski solution with the Schwarzschild solution, as follows. Let $\chi^\star(r)$ be (regular) cut-off function vanishing for all $r\leq 1/2$ and which is identically $1$ for all $r\geq 3/4$. Given a mass coefficient $m>0$, the reference metric of interest here is (by restricting attention to $t \geq 2$ for convenience in the discussion) 
\bel{equa-defineMS-new} 
g_\glue^\star = \gMink + \chi^\star (r) \, \chi^\star(r/(t-1)) (g_\Sch - \gMink), 
\qquad  
t \geq 2,
\ee 
which coincides with $\gMink$ in the cone $\big\{ r/(t-1)< 1/2 \big\}$ and with $g_\Sch$ in the exterior $\big\{ r/(t-1)\geq 3/4 \big\}$ (containing the light cone).  Observe that $g_{\Sch}$ is Ricci flat, then it is easy to check that \eqref{equa-31-12-20} holds with any $4/3<\lambda<1$ and $0<\theta<1-\lambda$. Thus $g^{\star}$ is chosen as a reference.


\subsection{Main result}

\paragraph{Admissible initial reference.}

As in \eqref{eq2-23-dec-2023}, the initial metric $g_0$, which is the restriction of $g$ on the initial slice $\Mcal_0$, can also be decomposed as 
\begin{equation}\label{equa-decomp-data-3} 
g_{0ab} = g^\star_{0ab} + u_{0ab} = \delta_{ab} + h^\star_{0ab} + u_{0ab}, 
\qquad 
k_{0 ab} = k^{\star}_{0ab} + l_{0ab},
\end{equation}
Inversely, if the pair $(g_0^{\star},k_0^{\star})$ defined on $\Mcal_0$ can be regarded as the restriction of a light-bending reference spacetime $g$ with parameters $(\lambda,\theta, \epss,N,\ell)$ on $\Mcal_0$ in the sens that
$$
\aligned
&g_0^{\star} = i_*(g) \quad &&\text{ is the first fundamental form of } i_*,
\\
&k_0^{\star}\quad &&\text{ is  the second fundamental form of } i_*,
\endaligned
$$
understood with respect to $g^\star$, 
with $i_*$ being the inclusion map $\Mcal_0 \to \RR^{1+3}$. Then $(g_0^{\star},k_0^{\star})$ is called a {\sl light-bending admissible initial reference} with parameter $(\lambda,\theta,\epss,N,\ell)$. 


\paragraph{Class of initial data sets of interest.}

In order to give a quantitative description 
of the perturbation, we introduce the following energy norms (with summation over all $|I|\leq N$): 
\begin{equation}\label{eq3-23-dec-2023}
\aligned
\NMenergy_{\expeta}^N(g_0,k_0) &:= \sum_I \|\la r\ra^{\expeta + |I|}\del_x^I \del_xu_0\|_{L^2(\RR^3)} + \|\la r\ra^{\expeta + |I|}\del_x^I l_0\|_{L^2(\RR^3)},
\\
\NCenergy_{\mu,\coeffi}^N(\phi_0,\phi_1) &:= 
\sum_I 
\coeffi\big\|\la r\ra^{\mu + N} \del^I \del \phi_0 \big\|_{L^2(\RR^3)}
+ \coeffi\big\|\la r\ra^{\mu+N} \del^I \phi_1 \big\|_{L^2(\RR^3)}
+ \coeffi^{1/2}\|\la r\ra^{\mu + N} \del^I \phi_0 \|_{L^2(\RR^3)},
\\
\NSenergy_{\mu}^N(\psi_0,\psi_1) &:= 
\sum_I 
\|\la r\ra^{\mu+N}\del^I\del\psi_0\|_{L^2(\RR^3)} + \|\la r\ra^{\mu+N} \del^I\psi_1\|_{L^2(\RR^3)} + \|\la r\ra^{\mu+N}\psi_0\|_{L^2(\RR^3)}. 
\endaligned
\end{equation}
The special form of $\NCenergy$ is determined by the energy associate to the evolution equation of $\phi$ (see \eqref{eq1-01-dec-2023}).
Here the coefficient $\coeffi$ in $\NCenergy$ is determined by the system \eqref{eq3-12-nov-2013}, more precisely, $f''(0)$. Then as in \cite{PLF-YM-main}, we introduce the admissible initial data set.

\begin{definition} \label{def1-03-01-2022}
An {\bf admissible initial data set} $\big( g_0,k_0,\phi_0,\phi_1,\psi_0,\psi_1\big)$ associated with the parameters $(\lambda,\theta, \epss, N, \ell, \expeta,\mu,\eps)$ 
consists of two symmetric two-tensors $g_0,k_0$ and two pairs of scalar fields $(\phi_0,\phi_1), (\psi_0,\psi_1)$ defined on $\RR^3$ and 
satisfying the following conditions. 

\begin{itemize}

\item  The f(R) constraint equations \eqref{consttransHan} are satisfied. 

\item There exists a decomposition into a sum 
\begin{equation} \label{eq1-06-01-2022}
g_{0ab} =  g^{\star}_{0ab} + u_{0ab},
\qquad\quad 
k_{0ab} = k^{\star}_{0ab} + l_{0ab}, 
\end{equation}
where  $(g_0^{\star},k_0^{\star})$  is a light-bending, $(\lambda,\theta, \epss,N,\ell)$-admissible initial reference.
 
\item Moreover, the perturbation $(u_0, l_0,\phi_0,\phi_1,\psi_0,\psi_1)$ has finite energy in the sense 
that
\begin{equation}
\NMenergy_{\expeta}^N(u_0, l_0) + \NCenergy_{\mu,\coeffi}^N(\phi_0,\phi_1) + \NSenergy_{\mu}^N(\psi_0,\psi_1)\leq \eps,
\end{equation}
\end{itemize} 
\end{definition}
Then similarly as we did in \cite{PLF-YM-main}, we state our main result for the $f(R)$ field equations, as follows. 
 
\begin{theorem}[Nonlinear stability of Minkowski space in modified gravity]
\label{theo:main2}
Consider the field equations of conformal $f(R)$ gravity \eqref{eq3-12-nov-2013} together with \eqref{equa:Talphabeta}, where $f=f(R)$ is a function satisfying the condition \eqref{Equa--condi-coeffi} for some $\coeffi \in (0,1)$. Fix some sufficiently large integer $N$ ($N=20$ being sufficient) and consider an admissible light-bending initial data $(g_0,k_0,\phi_0,\phi_1,\psi_0,\psi_1)$ with
parameters $(\lambda,\theta, \epss, N, \ell, \expeta,\mu,\eps)$ satisfying 
\begin{equation} \label{eq1-20-01-2022}
\expeta\in(1/2,1), 
\qquad 
\lambda \in (3/4,1), 
\qquad 
\mu\in(3/4,1),
\qquad 
\kappa\leq \mu\leq \lambda. 
\end{equation}  
Then there exists a small constant $c_0>0$ (determined by the Einstein system) such that for all 
\begin{equation}
\label{equa-all-conditions-exponents}
\aligned 
&  
N \, \theta \leq c_0 \,\min(\expeta-1/2, \mu-3/4), 
\qquad
\qquad
\eps< c_0 \epss <c_0, 
\endaligned
\end{equation}
the maximal globally hyperbolic Cauchy development of $(g_0,k_0,\phi_0,\phi_1,\psi_0,\psi_1)$ 
associated with the field equations of conformal $f(R)$ gravity \eqref{eq3-12-nov-2013} together with \eqref{equa:Talphabeta}
is future causally geodesically complete, and asymptotically approaches Minkowski spacetime in all (timelike, null, spacetime) directions. Moreover,  this solution satisfies the light-bending property for all times and 
(in a suitable sense) 
remains close to the reference spacetime $(\RR_+^{3+1}, g^\star)$.  
\end{theorem} 
 

\section{Spacetime foliation and energy inequality} 
\label{section=5}

\subsection{Spacetime foliation}
\label{label-spacetimefolmiation}

\paragraph{Time function.} 

We present a spacetime foliation which is based on prescribing a global time function (specified next)  and allows us to distinguish between asymptotically hyperboloidal and asymptotically Euclidean domains of the spacetime. 
(We recall that the comparison with the Friedrich's terminology was explained in the introduction.) 
More precisely, we define several domains along each hypersurface of constant time denoted by $s$. At each $s$ we introduce the {\it hyperboloidal and Euclidean radii} 
\be
\rhoH(s) := {1 \over 2} (s^2 -1), 
\qquad \rhoE(s) := {1 \over 2} (s^2 +1), 
\ee
respectively. We assume that the time function satisfies the following properties (see below for the construction): 
\bel{equa-prop-time1} 
\fTime(s,r) = \begin{cases} 
(s^2+r^2)^{1/2},   \quad &  r \leq \rhoH(s) \qquad \mbox{ (hyperboloidal domain),} 
\\
r + 1 = (s^2+1)/2, \quad &  r = \rhoH(s),
\\
\fTime^\E(s), &  r \geq \rhoE(s) \qquad \mbox{ (Euclidean domain),}
\end{cases}
\ee
in which $\fTime^\E=\fTime^\E(s)$ is independent of $r$ and, for universal constants $K_1, K_2>0$,
\be
K_1 \, s^2 \leq \fTime(s,r) \leq K_2 \, s^2, 
\ee
together with 
\be
\aligned
& 0\leq \del_r \fTime(s,r)<1, \quad 
&& \mbox{ (slices of constant $s$ are spacelike),}
\\ 
& 0 < \del_r\fTime(s,r)<1, \quad 
&& \mbox{ when $0<r \leq \rhoH(s)$},  
\\
& |\del_r \, \del_r \fTime(s,r)| \lesssim 1. 
\endaligned
\ee
We will provide an explicit time function that satisfies all of these properties.
 
 
\paragraph{Euclidean-hyperboloidal foliation.}

A one-parameter family of spacelike, asymptotically Euclidean hypersurfaces is defined as 
\be
\Mcal_s := \big\{ (t,x^a)\in \Mcal \, \big/ \, t = \fTime(s,r) \big\}. 
\ee
In the future of the initial surface $\{t=1\}$, namely $\{t\geq 1\} = \Mcal^{\init}\cup \bigcup_{s\geq s_0}\Mcal_s$, we distinguish the {\sl initial domain} (as we call it) $\Mcal^{\init} = \{(t,x)\,/\, 1\leq t\leq T(s_0,r)\}$ within which standard local-in-time existence arguments apply. Each slice $\Mcal_s = \MH_s \cup \Mtran_s \cup \Mext_s$  
is decomposed into three domains (with overlapping boundaries): 
\bel{equ-foliation-def} 
\aligned
\MH_s & :=   \big \{ t = \fTime(s,|x|) \, \big/ \, 
|x| \leq \rhoH(s)
\big\}, 
&& \mbox{asymptotically hyperboloidal,}
\\
\MM_s & :=   \big\{t = \fTime(s,|x| )  \, \big/ \, 
\rhoH(s)\leq  |x| \leq \rhoE(s)
\big\}, 
&&  \mbox{merging (or transition),}
\\
\Mext_s & :=   \big\{ t=\fTime(s)  \, \big/ \, 
\rhoE(s) \leq |x|
\big\}, 
&& \mbox{asymptotically Euclidean}, 
\endaligned
\ee
with also $\MME_s := \Mext_s \cup \MM_s$. Some additional notation is needed: 
\be
\aligned
\Mcal_{[s_0,s_1]} & :=   { \big\{ \fTime(s_0,r)\leq t\leq \fTime(s_1,r) \big\} = \bigcup_{s_0\leq s\leq s_1} \Mcal_s}, 
\qquad 
\Mcal_{[s_0, + \infty)} & := \bigcup_{s\geq s_0} \Mcal_s, 
\endaligned
\ee
and, similarly, we set $\Mcal^\H_{[s_0, s_1]}$, $\Mcal^\H_{[s_0, + \infty)}$, etc. By construction,  there exists a function $c=c(s)\in (0,1)$ such that the radial variable $r$ in each of the three domains satisfies 
\be
r = |x| \in  \,
\begin{cases}
\hskip1.cm [0, t-1],
& \MH_{[s_0, + \infty)},
\\ 
[t-1,t -  c(s)], 
\qquad & \MM_{[s_0, + \infty)}, 
\\
[t - c(s), +\infty),
& \Mext_{[s_0, + \infty)}. 
\end{cases}
\ee

The future-oriented normal to the spacelike hypersurfaces $\Mcal_s$ (with respect to the Euclidean metric in $\RR^4$) reads (with $a=1,2,3$) 
\begin{subequations}
\label{normal-Ms-equa1-2-june}
\be \label{normal-Ms}
n_s =  \frac{\big(1, - (x^a/r)\del_r \fTime\big)}{\sqrt{1+|\del_r \fTime|^2}} 
= \big( (1+\xi^2(s,r))r^2 +s^2 \big)^{-1/2} \Big((s^2 + r^2)^{1/2},-x^a\xi(s,r) \Big), 
\ee
while the surface element (with respect to the Euclidean metric) is 
\be \label{equa1-2-june}
d\sigma_s = (1+|\del_r \fTime|^2) \, dx = \big( s^2 +r^2(1+\xi(s,r)^2)\big)^{1/2} (s^2 + r^2)^{-1/2} \, dx
\ee
\end{subequations}
and, in particular,
\be \label{eq1-08-05-2020}
n_sd\sigma_s = (1, - (x^a/r)\del_r \fTime) \, dx = (1, - \del_a \fTime) \, dx = \Big(1,\frac{-\xi(s,r)x^a}{(s^2+r^2)^{1/2}}\Big) \, dx.
\ee


\paragraph{Construction of the time function.}

We now specify a construction of the time function. Consider any cut-off function $\chi: \RR \to [0,1]$ satisfying 
\be
\chi(x) = 
\begin{cases}
0,   & x \leq 0, 
\\
1,   & x> 1, 
\end{cases}
\ee
and (for simplicity in some of our arguments) $\chi^{(m)}(x)  > 0$ for all $x \in (0,1/2)$ and each $m= 0,1,2,3$. By definition, the {\it foliation coefficient} is the function 
\bel{equation-xi-def}
\xi(s,r) := 1-\chi(r- \rhoH(s)) 
= \begin{cases}
1, \quad & r <\rhoH(s), 
\\
0, \quad & r > \rhoE(s), 
\end{cases}
\ee
and will be applied to``select'' the hyperboloidal domain.  Specifically, we define our time function $t = \fTime(s,r)$ by solving the ordinary differential equation 
\be
\del_r \fTime(s,r) = \frac{r \, \xi(s,r)}{(s^2+r^2)^{1/2}}, 
\qquad 
\fTime(s,0) = s. 
\ee
It can be checked that this time function enjoys \eqref{equa-prop-time1}, as required. 

Moreover, the Jacobian matrix associated with the foliation  reads 
$$
\left(
\begin{array}{cc}
\del_s \fTime &\del_sx
\\
\del_x \fTime &\del_x x 
\end{array}
\right)
=
\left(
\begin{array}{cc}
\del_s \fTime &0
\\
(x^a/r) \, \del_r \fTime & I
\end{array}
\right)
$$
and the
Jacobian is $J(s,x) = \del_s\fTime(s,x)$, leading to the corresponding volume element $dtdx = J \, dsdx$. It can be checked that   
\be
J \leq
\begin{cases}
{s \over \fTime}= s \, (s^2 +r^2)^{-1/2},   
\\
\xi s \, (s^2 +r^2)^{-1/2} + (1- \xi) \, 2s,
\\
2s,  
\end{cases}
\, 
J 
\geq
\begin{cases}
{s \over \fTime} = s \, (s^2 +r^2)^{-1/2}, 
 \quad      & \MH_s,
\\
\xi \, s \, (s^2 + r^2)^{-1/2} + (1- \xi) 3s/5, 
 \quad      &   \MM_s, 
\\
3s/5 
 \quad      &  \Mext_s. 
\end{cases}
\ee


\subsection{Frames of vector fields of interest}

\paragraph{The semi-hyperboloidal frame.}

In our method of analysis, we combine estimates involving different frames, as follows. The {\it semi-hyperboloidal frame} denoted by 
\bel{equa-shf} 
\delH_0 = \del_t, 
\qquad
\delH_a = \delsH_a  = \frac{x^a}{t} \del_t + \del_a
\ee
was introduced first in~\cite{PLF-YM-book}. It is defined globally in $\Mcal_s$, relevant within the hyperboloidal domain and is relevant in order to exhibit the (quasi-)null form structure of the nonlinearities and establish decay properties in timelike and null directions. Some of our arguments also involve radial integration based on $\delsH_r = (x^a /r)\delsH_a$.


\paragraph{The semi-null frame.} 

Another important frame is the {\it semi-null frame} given by 
\bel{equa-snf} 
\delN_0  = \del_t, 
\qquad  
\delN_a = \delsN_a = {x^a \over r} \del_t + \del_a 
\ee
is defined everywhere in $\Mcal_s$ except on the center line $r=0$ and is the appropriate frame within the Euclidean-merging domain in order to exhibit the structure of the (null, quasi-null) nonlinearities of the field equations, and establish decay properties in spatial/null directions.
 

\paragraph{The Euclidean--hyperboloidal frame.}

Furthermore, the frame 
\be
\delEH_0 = \del_t, 
\quad
\delEH_a  = \delsEH_a 
=   \del_a + (x^a/r)\del_r \fTime \, \del_t  
\ee
is called the {\it Euclidean--hyperboloidal frame} and consists of tangent vectors to the slices $\Mcal_s$ and interpolates between $\delEH_a = \delH_a$ in $\MH_s$, and $\delEH_a = \del_a$ in $\Mext_s$. Some of our arguments are also based on radial integration based on $\delsEH_r = (x^a/r)\delsEH_a$. 

 Changes of frame formulas between these frames are used throughout, such as
$\delN_\alpha = \PhiN{}_\alpha^\beta \, \del_\beta$ and $\del_\alpha = \PsiN_{\alpha}^\beta \delN_\beta$ with 
\be
\big(\PhiN_{\alpha}^\beta \big)
= 
\left(
\begin{array}{cccc}
1 &0 &0 &0
\\
x^1/r& 1 &0 &0
\\
x^2/r&0  &1 &0
\\
x^3/r& 0 &0 &1
\end{array}
\right), 
\qquad 
\big(\PsiN_{\alpha}^\beta \big)
= 
\left(
\begin{array}{cccc}
1 &0 &0 &0
\\
-x^1/r& 1 &0 &0
\\
-x^2/r&0  &1 &0
\\
-x^3/r& 0 &0 &1
\end{array}
\right).
\ee 


\subsection{Basic energy functional}

\paragraph{Energy weight coefficients.}

The fundamental energy functional (stated in \eqref{equa-eergy} and \eqref{equa-eergy-equiv}, below)
 involved another geometric weight, denoted by $\zeta = \zeta(t,x)$ and defined by 
\bel{equa-zeta}
\zeta(s,r)^2 = 1 - (\del_r \fTime(s,r))^2. 
\ee
This weight coincides with  $s/t = s /(s^2+r^2)^{1/2}$ in the hyperboloidal domain, while it reduces to $1$ in the Euclidean domain. In fact, it provides us with an interpolation between the energy density induced on hyperboloids and the one induced on Euclidean slices. Various estimates on $\zeta$ can be checked, for instance $\frac{|r-t \, |+1}{r} \lesssim \zeta^2 \leq\zeta\leq 1$ valid within the Euclidean-merging domain, as well as
$K_1 \, \zeta^2 s\leq J\leq  K_2 \, \zeta^2 s$ in the merging domain (for some universal constants $K_1,K_2$). 

Furthermore, we introduce a weight which measures the distance to the light cone, and is defined from the prescription of a smooth and non-decreasing function $\aleph$ satisfying 
\be
\aleph(y) = 
\begin{cases}
0, \quad 
& y \leq -2, 
\\
y+2, \quad &  y \geq {-1}, 
\end{cases}
\ee 
and specifically we set $\crochet := 1 + \aleph({r-t})$, which we refer to as the {\sl energy weight.} By definition, $\aleph'$ is non-negative.


\paragraph{Energy identity (flat case).}

We multiply the wave-Klein-Gordon equation  (with $c \geq 0$) 
\be
\Box u - c^2 u = F
\ee
by $- 2 \, \crochet^{2 \expeta}  \del_t u$ with $\Box  = \Box_{\gMink}  = - \del_t \, \del_t + \sum_{a=1,2,3} \del_a\del_a$. We treat simultaneously the wave and Klein-Gordon operators by assuming here that $c \geq 0$. We find the divergence identity 
$$
\aligned
 \del_t \big( V^0_{\expeta,c}[u] \big )+ \del_a \big( V^a_{\expeta,c}[u] \big)
& = 
2\expeta \crochet^{-1}  \aleph'({ r-t})(-1, x^a/r) \cdot V_{\expeta,c}[u] 
- 2 \, \crochet^{2 \expeta} \del_t u \, F, 
\\
V_{\expeta,c}[u]  
& = - \crochet^{2 \expeta} \big(-|\del_t u|^2 - \sum_a|\del_au|^2 - c^2u^2, 2 \del_t u\del_au\big). 
\endaligned
$$
We define our energy functional on each Euclidean--hyperboloidal slice $\Mcal_s$ as
\bel{equa-eergy}
\aligned
\Eenergy_{\expeta,c}(s,u) 
& = \int_{\Mcal_s}V_{\expeta,c}[u] \cdot n_sd\sigma_s 
\\
&
= \int_{\Mcal_s} 
\Big(|\del_tu|^2 + \sum_a|\del_au|^2 
+ 2 \del_a \fTime(s,r) 
\del_t u\del_a u + c^2u^2 \Big) \, 
\crochet^{2\expeta} \, dx 
\endaligned
\ee
or, equivalently,
\bel{equa-eergy-equiv}
\aligned
\Eenergy_{\expeta,c}(s,u) 
& =  \int_{\Mcal_s} \Big(\zeta^2|\del_t u|^2 + \sum_a |\delsEH_au|^2 + c^2u^2 \Big) \,
\crochet^{2\expeta} \, dx, 
\\
& = \int_{\Mcal_s}
\Big(\zeta^2\sum_a|\del_au|^2 
+ (\del_r \fTime(s,r))^2 
\sum_{a<b} |\Omega_{ab}u|^2 
+ |\delsEH_r u|^2 + c^2 \, u^2 \Big) \, \crochet^{2\expeta} \, dx. 
\endaligned
\ee
This functional involves the energy coefficient $\zeta$,  which is non-trivial in the merging and hyperboloidal domains, and depends upon our choice of coefficient $\xi = \xi(s,r)$. 


\section{Field equations in conformal wave gauge} 
\label{section=6}
 
\subsection{Field equations in conformal wave gauge} 

\paragraph{Gauge freedom.}
 
{\sl From now on, we work exclusively with the conformal metric $\gconf$.} Since the field equations under consideration are geometric in nature, it is essential to fix the degrees of gauge freedom before tackling any stability issue from the perspective of the initial value problem for partial differential equations. As already pointed out, our analysis relies on a single global coordinate chart $(x^\alpha)=(t, x^a)$ and, more specifically, on a choice of coordinate functions $x^\alpha$ satisfying the  homogeneous linear wave equation in the unknown metric. This means that we impose that the functions $x^\alpha: \Mcal \to \RR$ satisfy the so-called {\sl (conformal) wave gauge conditions} ($\alpha=0,1,2,3$)
\bel{equa:wcoo}
\Box_\gconf x^\alpha = 0
\ee
which, with $\Gamma_\gamma = \gconf_{\lambda\gamma}\gconf^{\alpha\beta}\Gamma_{\alpha\beta}^{\lambda}$, is equivalent to 
\begin{equation}\label{eq1-27-oct-2023}
\Gamma_{\gamma} \equiv 0. 
\end{equation}
In this gauge, we obtain a nonlinear system of second-order partial differential equations with second-order constraints. For the Einstein theory, the unknowns are the metric coefficients $g_{\alpha\beta}$ in the chosen coordinates, together with the scalar field $\psi$. It is well-known that the constraints are preserved during time evolution \cite{YCB-book}. On the other hand, for $f(R)$-theory, the unknowns include the conformal metric coefficients $\gconf_{\alpha\beta}$, the scalar curvature field $\phi$, and the matter field $\psi$. In \cite{PLF-YM-memoir}, it was established that the associated constraints are also preserved during time evolution. We first present the equations in a schematic form, while the algebraic structure of this system of nonlinear and coupled equations will be analyzed in greater detail later on.


\paragraph{Ricci curvature.}

\bse
Let us recall here that the Ricci curvature can be decomposed as 
\be
2 \, R_{\alpha\beta}
= - g^{\mu\nu} \del_\mu \del_\nu g_{\alpha\beta} 
+  \Fbb_{\alpha\beta}(g,g;\del g, \del g) 
+ \big(\del_\alpha\Gamma_\beta + \del_\beta\Gamma_\alpha\big) 
+  W_{\alpha\beta},
\ee
in which the spurious second-order terms
\be
W_{\alpha\beta} 
:= g^{\delta \delta'} \del_{\delta} g_{\alpha \beta} \Gamma_{\delta'} - \Gamma_\alpha \Gamma_\beta
\ee
are known to vanish identically when the wave gauge condition is imposed. The first-order derivatives 
\be
\Fbb_{\alpha\beta}(g,g;\del g, \del g) 
:= \Pbb_{\alpha\beta}(g,g;\del g, \del g)  + \Qbb_{\alpha\beta}(g,g;\del g, \del g) 
\ee
and
\be
\Pbb_{\alpha\beta} (g,g;\del g, \del g) 
:= - \frac{1}{2} g^{\mu\mu'} g^{\nu\nu'} \del_\alpha g_{\mu\nu} \del_\beta g_{\mu'\nu'} + \frac{1}{4} g^{\mu\mu'} g^{\nu\nu'} \del_\alpha g_{\mu\mu'} \del_\beta g_{\nu\nu'}
\ee
and
\be
\aligned
\Qbb_{\alpha\beta}(g,g;\del g, \del g) 
& := 
g^{\mu\mu'} g^{\nu\nu'} \del_\mu g_{\alpha\nu} \del_{\mu'} g_{\beta\nu'}
- g^{\mu\mu'} g^{\nu\nu'} \big(\del_\mu g_{\alpha\nu'} \del_\nu g_{\beta\mu'} - \del_\mu g_{\beta\mu'} \del_\nu g_{\alpha\nu'} \big)
\\
& \quad + g^{\mu\mu'} g^{\nu\nu'} \big(\del_\alpha g_{\mu\nu} \del_{\nu'} g_{\mu'\beta} - \del_\alpha g_{\mu'\beta} \del_{\nu'} g_{\mu\nu} \big)
\\
& \quad 
+ \frac{1}{2} g^{\mu\mu'} g^{\nu\nu'} \big(\del_\alpha g_{\mu\beta} \del_{\mu'} g_{\nu\nu'} - \del_\alpha g_{\nu\nu'} \del_{\mu'} g_{\mu\beta} \big)
\\
& \quad + g^{\mu\mu'} g^{\nu\nu'} \big(\del_\beta g_{\mu\nu} \del_{\nu'} g_{\mu'\alpha} - \del_\beta g_{\mu'\alpha} \del_{\nu'} g_{\mu\nu} \big)
\\
& \quad + \frac{1}{2} g^{\mu\mu'} g^{\nu\nu'} \big(\del_\beta g_{\mu\alpha} \del_{\mu'} g_{\nu\nu'} - \del_\beta g_{\nu\nu'} \del_{\mu'} g_{\mu\alpha} \big). 
\endaligned
\ee
These two quadratic forms are called {\bf quasi-null} and {\bf null} nonlinearities and act on the gradient of $g$. 
\ese

Furthermore, for the clarity in the presentation we recall that 
\be
\aligned
2 \, ^{(w)}R_{\alpha\beta}  
& =
2 \, R_{\alpha\beta} 
- \big(\del_\alpha\Gamma_\beta + \del_\beta\Gamma_\alpha\big) 
- W_{\alpha\beta}
\endaligned
\ee
represents the Ricci curvature after suitable reduction taken the wave gauge condition into account.  
Thanks to the wave gauge conditions in the form of \eqref{eq1-27-oct-2023} together with the following notation for the  {\sl reduced wave operator} 
\begin{equation}
\BoxChapeaud = \BoxChapeau_{\gconf} := \gconf^{\mu\nu} \del_{\mu} \del_{\nu},
\end{equation}
we see that the Ricci curvature reads 
\begin{equation}\label{eq2-27-oct-2023}
2 \, R_{\alpha\beta} =2 \, ^{(w)}R_{\alpha\beta} 
= - \BoxChapeau_{\gconf} g_{\alpha\beta} 
+  \Fbb_{\alpha\beta}(g,g;\del g, \del g).
\end{equation}


\paragraph{Nonlinear second-order system.}

Taking \eqref{eq2-27-oct-2023}, we obtain the f(R) field equation 
\begin{equation}
\BoxChapeaud g_{\alpha\beta} = \Fbb_{\alpha\beta}(g,g;\del g,\del g)
 - 3\kappa^2\nabla_{\alpha}\phi\nabla_{\beta}\phi - \coeffi\Ufrak(\phi)\gconf_{\alpha\beta} - 2 e^{-\coeffi \phi}\big(\givenCurv_{\alpha\beta} - \frac{1}{2}\Tr_{\gconf}(\givenCurv) \big)
\end{equation}
together with the evolution of the effective curvature field 
\begin{equation}\label{eq1-01-nov-2023}
3\kappa\BoxChapeaud  \phi - \Ufrak'(\phi) = e^{-\coeffi\phi}\Tr_{\gconf}(\givenCurv).
\end{equation}
Now if we take the (physical) matter field $T_{\alpha\beta}$ as defined in \eqref{equa:Talphabeta} and consider the Jordan coupling $\givenCurv_{\alpha\beta} = 8\pi\givenT_{\alpha\beta}$, we obtain the main system of interest.

\begin{proposition}[$f(R)$-gravity equations in conformal wave gauge] 
\label{propo41}
In the conformal wave gauge \eqref{eq1-27-oct-2023}, the reduced field equations of modified gravity 
\eqref{equa--FieldEquations} for a matter field $\phi$ satisfying \eqref{equa:Talphabeta} 
take the following form of a nonlinear system of coupled wave and Klein-Gordon equations for the conformal metric components $\gconf_{\alpha\beta}$, the scalar curvature field $\phi$, and the matter field $\psi$:
\begin{equation}\label{MainPDE}
\aligned
\BoxChapeaud_\gconf \gconf_{\alpha\beta}
& =\, \Fbb[g]_{\alpha\beta} + \Gbb_{\coeffi}[\phi,g]_{\alpha\beta} + \Mbb_{\coeffi}[\phi,\psi,g]_{\alpha\beta},
\\
3 \coeffi\, \BoxChapeaud_\gconf \phi - \Ufrak'(\phi) 
& = -8\pi e^{-\coeffi\phi} \gconf(\del\psi,\del\psi) - 32\pi e^{-2\coeffi\phi}V(\psi),
\\
\BoxChapeaud_\gconf \psi - e^{-\coeffi\phi}V'(\psi) 
& =  \coeffi \, \gconf(\del\psi, \del \phi),
\endaligned
\end{equation}
in which
\begin{equation}\label{equa:ABs}
\aligned
\Gbb_{\coeffi}[\phi,g]_{\alpha\beta}
& := - 3\kappa^2\del_{\alpha}\phi\del_{\beta}\phi - \coeffi\Ufrak(\phi)\gconf_{\alpha\beta} ,
\\
\Mbb_{\coeffi}[\phi,\psi,g]_{\alpha\beta}
& := - 16 \pi \,e^{-\coeffi\phi} \big(\del_\alpha \psi \del_\beta \psi + V(\psi) e^{-\coeffi\phi}  \, \gconf_{\alpha \beta} \big). 
\endaligned
\end{equation}
\end{proposition}


\paragraph{Bounds on the potential functions.}

Throughout the discussion, we impose the following conditions on the potentials $V$ and $\Ufrak$: for  any scalar fields $\phi, \psi$ satisfying the uniform bounds 
\be
|\phi|_{[p/2]}\lesssim1, \quad |\coeffi^{1/2}\psi|_{[p/2]}\lesssim 1,
\ee
one requires that 
\begin{subequations}\label{eq6-03-dec-2023}
\begin{equation}\label{eq3-03-dec-2023}
|V(\psi)|_{p,k}\lesssim \sum_{p_1+p_2=p\atop k_1+k_2=k}|\psi|_{p_1,k_1} |\psi|_{p_2,k_2}
\end{equation}
and
\begin{equation}\label{eq4-03-dec-2023}
|\Ufrak(\phi)|_{p,k}\lesssim \sum_{p_1+p_2=p\atop k_1+k_2=k}|\phi|_{p_1,k_1}|\phi|_{p_2,k_2}.
\end{equation}
\end{subequations}
The  condition \eqref{eq3-03-dec-2023} is guaranteed by
\be\label{eq3-24-dec-2023}
V(0) = 0, 
\qquad 
V'(0) \simeq 1, 
\qquad 
\sup_\psi |V^{(n)}(\psi)|\lesssim 1,\quad n=2,3, \ldots, 
\ee
while for the second condition we need
\begin{equation}\label{eq5-03-dec-2023}
\Ufrak(0) = 0, 
\qquad
\Ufrak'(0) \simeq 1, 
\qquad 
\sup_\phi 
|\Ufrak^{(n)}(\psi)|\lesssim 1 ,\quad n=2, 3, \ldots,  
\end{equation}
 the supremum being taken over a small neighborhood of the origin. 
Such bounds are easily checked from our assumption \eqref{Equa--condi-coeffi}on the modified gravity function $f$.
 

\subsection{Structure of the f(R) field equation} 

\paragraph{Decomposition of the metric.}

As we have mentioned in Section~\ref{section=4}, the metric is decomposed as the sum of a reference metric plus a perturbation, namely 
$$
g_{\alpha\beta} = \eta_{\alpha\beta} + h^{\star}_{\alpha\beta} + u_{\alpha\beta} = g^{\star}_{\alpha\beta} + u_{\alpha\beta}, 
$$
in which $g^{\star}$ is chosen here to be \eqref{equa-defineMS-new}. Following our general strategy in~\cite{PLF-YM-main}, the equation satisfied by the perturbation $u$ is written as
\begin{equation}
\BoxChapeaud_\gconf u_{\alpha\beta} = \Fbb^{\star}[u]
 + \Gbb_{\coeffi}[\phi,g]_{\alpha\beta} + \Mbb_{\coeffi}[\phi,\psi,g]_{\alpha\beta}
  + \Ibb^{\star}_{\alpha\beta}[u] - u^{\mu\nu}\del_{\mu}\del_{\nu}h^{\star}_{\alpha\beta} + 2^{(w)}R^{\star}_{\alpha\beta}, 
\end{equation}
with
\bse
\begin{equation}
2^{(w)}R_{\alpha\beta}=  -\Boxt_g u_{\alpha\beta} 
- u^{\mu\nu} \del_{\mu} \del_{\nu} g^\star_{\alpha\beta}  
+ 2 \, \Rwave_{\alpha\beta}
+ \Fbb^{\star}_{\alpha\beta}[u] + \Ibb^{\star}_{\alpha\beta}[u], 
\ee
\be
2^{(w)}R_{\alpha\beta} - 3\coeffi^2\nabla_{\alpha}\phi\nabla_{\beta}\phi - \Ufrak(\phi)g_{\alpha\beta}
= 16\pi e^{-\coeffi\phi}\Big(\givenT_{\alpha\beta} - \frac{1}{2}\Tr_g(\givenT)g_{\alpha\beta}\Big). 
\ee
Here, the nonlinearity $\Fbb^{\star}[u]$ can further be decomposed as
\be
\Fbb^{\star}[u] =  \Pbb^\star_{\alpha\beta}[u] + \Qbb^{\star}_{\alpha\beta}[u] =
\Pbb_{\alpha\beta} (g^{\star},g^{\star};\del u, \del u) + \Qbb_{\alpha\beta}(g^{\star},g^{\star};\del u, \del u),
\ee
and
\be
\Ibb^{\star}_{\alpha\beta}[u] = \Lbb^\star_{\alpha\beta}[u]   + \Bbb^\star_{\alpha\beta}[u] + \Cbb^\star_{\alpha\beta}[u], 
\ee
where linear, bilinear, cubic (and higher order) interactions terms are defined, respectively, as  
\begin{equation}\label{eq4-06-01-2020}
\aligned
\Lbb^\star_{\alpha\beta}[u]
& := \Fbb_{\alpha\beta}(g^\star, g^\star;\del g^\star, \del u) + \Fbb_{\alpha\beta}(g^\star, g^\star;\del u, \del g^\star) 
\\
& \quad 
+ \Fbb_{\alpha\beta}(u,g^\star;\del g^\star, \del g^\star) + \Fbb_{\alpha\beta}(g^\star,u;\del g^\star, \del g^\star),
\\
\Bbb^\star_{\alpha\beta} [u]
& := \Fbb_{\alpha\beta}(u,g^\star;\del u, \del g^\star) 
+ \Fbb_{\alpha\beta}(u,g^\star;\del g^\star, \del u)
+ \Fbb_{\alpha\beta}(g^\star,u;\del g^\star, \del u) 
\\
& \quad 
+ \Fbb_{\alpha\beta}(g^\star,u;\del u, \del g^\star) 
+ \Fbb_{\alpha\beta}(u,u;\del g^\star, \del g^\star),
\\
\Cbb^\star_{\alpha\beta}  [u]
& :=  \Fbb_{\alpha\beta}(u,g^\star;\del u, \del u) + \Fbb_{\alpha\beta}(g^\star,u;\del u, \del u) + \Fbb(u,u;\del g^\star, \del u) 
\\
& \quad + \Fbb_{\alpha\beta}(u,u;\del u, \del g^\star)
+ \Fbb_{\alpha\beta}(u,u;\del u, \del u).
\endaligned
\end{equation}
These latter terms are much easier to analyze in comparison with the quadratic terms $\Fbb^{\star}[u]$ above.
\ese
%


\paragraph{PDE formulation of the Cauchy problem.}

As we have explained in Section~\ref{section=4}, the unknown metric is decomposed as the sum of a reference plus a perturbation. We now formulate the evolution problem introduced in Definition~\ref{def cauchy conf}  
as a Cauchy problem associated with a system of PDEs with unknown $(u_{\alpha\beta},\phi,\psi)$:
\begin{equation}\label{eq1-24-dec-2023}
\aligned
\BoxChapeaud_\gconf u_{\alpha\beta} &= \Fbb^{\star}[u]
+ \Gbb_{\coeffi}[\phi,g]_{\alpha\beta} + \Mbb_{\coeffi}[\phi,\psi,g]_{\alpha\beta}
+ \Ibb^{\star}_{\alpha\beta}[u] - u^{\mu\nu}\del_{\mu}\del_{\nu}h^{\star}_{\alpha\beta} + 2^{(w)}R^{\star}_{\alpha\beta},
\\
3 \coeffi\, \BoxChapeaud_\gconf \phi - \Ufrak'(\phi) 
& = -8\pi e^{-\coeffi\phi} \gconf(\del\psi,\del\psi) - 32\pi e^{-2\coeffi\phi}V(\psi),
\\
\BoxChapeaud_\gconf \psi - e^{-\coeffi\phi}V'(\psi) 
& =  \coeffi \, \gconf(\del\psi, \del \phi),
\endaligned
\end{equation}
where $\wR^{\star}$ was introduced in \eqref{equa-Riccistar}, 
with the following initial data
\begin{equation}\label{eq2-24-dec-2023}
u_{\alpha\beta}(t_0,x),\quad \del_t u_{\alpha\beta}(t_0,x),\qquad \phi(t_0,x),\quad\del_t\phi(t_0,x),\qquad \psi(t_0,x),\quad\del_t\psi(t_0,x)
\end{equation}
being fixed.

Let us now state the pointwise estimates enjoyed by the nonlinearities. Also for the coupling source terms arising in the right-hand side of the Klein-Gordon equations satisfied by $\phi$ and $\psi$, the estimates are stated in Lemma~\ref{lem1-08-dec-2023}, below. 

\begin{lemma}[Null quadratic terms] 
\label{Lemme1-03-dec-2023}
In the Euclidean-merging domain $\MME = \{r\geq t-1\}$, null forms are controlled by good derivatives and a contribution depending upon the reference metric:
\be
|\Qbb^{\star}[u]|_{p,k} := \max_{\alpha,\beta}|\Qbb_{\alpha\beta}^\star[u] |_{p, k}
\lesssim  \sum_{p_1+p_2 = p\atop k_1+k_2=k} |\del u|_{p_1, k_1} |\delsN u|_{p_2, k_2} 
+ | h^\star |_p\sum_{p_1+p_2=p\atop k_1+k_2=k} |\del u|_{p_1, k_1} |\del u|_{p_2, k_2}.
\ee
In the hyperboloidal domain $\MH = \{r<t-1\}$,
\begin{equation}
|\Qbb_{\alpha\beta}^\star[u] |_p
\lesssim  \sum_{p_1+p_2 = p} |\del u|_{p_1} \big( |\delsH u|_{p_2} + (s/t)^2 |\del u|_{p_2} \big)
+ | h^\star |_p\sum_{p_1+p_2=p} |\del u|_{p_1} |\del u|_{p_2}.
\end{equation}
\end{lemma}

\begin{lemma}[Quasi-null interaction terms] 
\label{lem1-02-dec-2023}
\bse
In the Euclidean-merging domain $\MME$, under the smallness condition $| h^\star |_p + |u|_{[p/2]} \ll 1$, the quasi-null terms satisfy   
\be
\aligned
|\slashed{\Pbb}^{\star\N}[u] |_p
&
\lesssim \sum_{p_1+p_2=p} |\del u|_{p_1} |\delts u|_{p_2} + \sum_{p_1+p_2+p_3=p} | h^\star |_{p_3} |\del u|_{p_1} |\del u|_{p_2},
\\
|\Pbb_{00}^{\star \Ncal}[u] |_{p,k}
& \lesssim   \hskip-.3cm   \sum_{p_1+p_2=p\atop k_1+k_2=k}|\del \uts|_{p_1,k_1} |\del \uts|_{p_2,k_2} 
+ \hskip-.3cm  \sum_{p_1+p_2=p}\Big(|\delts u|_{p_1} |\del u|_{p_2}
+ \SbbME_{p_1}[u] \, |\del u|_{p_2}  \Big) 
\\
& \quad + \hskip-.3cm  \sum_{p_1+p_2+p_3=p}\hskip-.3cm   | h^\star |_{p_3} |\del u|_{p_1} |\del u|_{p_2}, 
\endaligned
\ee
where
\be
\aligned
\SbbME_p[u]
&:=  
r^{-1} |u|_p   +|\delsN h^\star |_p + r^{-1} | h^\star |_p  
+ \sum_{p_1+p_2=p}|\del  u|_{p_1} |u|_{p_2} 
\\
& \quad
+ \sum_{p_1+p_2=p}\big(| h^\star |_{p_1} |\del u|_{p_2} + |u|_{p_1} |\del h^\star |_{p_2} + |h^{\star}|_{p_1}|\del h^{\star}|_{p_2}\big).
\endaligned
\ee
\ese
\bse
In the hyperboloidal domain $\Mcal^\Hcal$ and under the smallness condition $| h^\star |_p + |u|_{[p/2]} \ll 1$,  the quasi-null terms satisfy 
\be
\aligned
|\slashed{\Pbb}^{\star\H} [u] |_p
& \lesssim \sum_{p_1+p_2=p} |\del u|_{p_1} \big(  |\delsH u|_{p_2} + (s/t)^2 |\del u|_{p_2} \big) 
+ \sum_{p_1+p_2+p_3=p} | h^\star |_{p_3} |\del u|_{p_1} |\del u|_{p_2},
\\
|\Pbb_{00}^{\star\H} [u] |_{p,k}
& \lesssim     \hskip-.3cm  \sum_{p_1+p_2=p}   \hskip-.3cm  \Big(|\del \usH |_{p_1} |\del \usH |_{p_2} 
+ \big( |\delsH u|_{p_1} + (s/t)^2 |\del u|_{p_1} \big)  |\del u|_{p_2}
\Big) 
\\
& \quad 
+   \hskip-.3cm 
\sum_{p_1+p_2=p}  \hskip-.3cm  | \Sbb_p^\H[u] |_{p_1} |\del u|_{p_2}  
+  \hskip-.3cm  \sum_{p_1+p_2+p_3=p}   \hskip-.3cm  | h^\star |_{p_3} |\del u|_{p_1} |\del u|_{p_2}, 
\endaligned
\ee
where
\be
\aligned
\Sbb_p^\H[u] 
& := t^{-1} |u|_{p_1} + \big(|\del h^\star |_{p_1} + t^{-1} | h^\star |_{p_1} \big)
+ \sum_{p_1+p_2=p_1} \big(|\del  u|_{p_1} |u|_{p_2} + |u|_{p_1} |u|_{p_2} \big)
\\
& \quad +   \sum_{p_1+p_2=p_1} \big(| h^\star |_{p_1} |\del u|_{p_2} + |u|_{p_1} |\del h^\star |_{p_2} + | h^\star |_{p_1} |\del h^\star |_{p_2} \big)
\\
& \quad
+  \sum_{p_1+p_2=p_1} 
\big(| h^\star |_{p_1} | u|_{p_2} + |u|_{p_1} | h^\star |_{p_2} + | h^\star |_{p_1} | h^\star |_{p_2} \big). 
\endaligned
\ee
\ese
\end{lemma} 

Finally, for the terms $\Mbb_{\coeffi}[\phi,\psi,g]$ and $\Gbb_{\coeffi}[\phi,g]$, we have the following estimates.

\begin{lemma}\label{lem1-08-dec-2023}
In $\{s\geq s_0\}$, provided that 
\begin{equation}
|\del\psi|_{[p/2]} + |\psi|_{[p/2]} + |\kappa\phi|_{[p/2]} + |h^{\star}|_{[p/2]}\lesssim 1,
\end{equation}
one has
\begin{equation}\label{eq2-03-dec-2023}
\aligned
|\Mbb_{\kappa}[\phi,\psi,g]|_{p,k}\lesssim &
\sum_{p_1+p_2=p\atop k_1+k_2=k}|\del\psi|_{p_1,k_1}|\del\psi|_{p_2,k_2} + |\psi|_{p_1,k_1}|\psi|_{p_2,k_2} 
\\
& \quad + |\coeffi\del\phi|_{p}|\big(|\del\psi|_{[p/2]}^2+|\psi|_{[p/2]}^2\big) + |\psi|_{[p/2]}^2  \sum_{\ord(Z)=1}|Zu|_{p-1},
\endaligned
\end{equation}
\begin{equation}\label{eq1-09-dec-2023}
|\Gbb_{\coeffi}[\phi,g]|_{p,k}\lesssim \sum_{p_1+p_2=p\atop k_1+k_2=k}\big(|\coeffi\del\phi|_{p_1,k_1}|\coeffi\del\phi|_{p_2,k_2} + \coeffi|\phi|_{p_1,k_1}|\phi|_{p_2,k_2} \big)
+ \coeffi|\phi|_{[p/2]}^2 \sum_{\ord(Z)=1}|Zu|_{p-1}.
\end{equation}
\end{lemma}

\begin{proof}
Both estimates are checked by direct but tedious calculations. Importantly, we need to pay attention to the nonlinear terms $V(\psi)$ and $\Ufrak(\phi)$ and, to this end, we apply our condition \eqref{eq6-03-dec-2023}.
\end{proof}


\subsection{Energy estimates for the f(R) field equations} 

The energy identity in curved spacetime associated with the wave equation  
\begin{equation}\label{eq1-26-dec-2023}
g^{\alpha\beta} \del_\alpha \del_\beta u  = F
\end{equation}
is expressed by decomposing the curved metric $g$ as 
$$
g^{\alpha\beta} =  \Minsk^{\alpha\beta} + H^{\alpha\beta}. 
$$
\begin{subequations}
After defining the energy-flux vector (with $a =1,2,3$) 
\begin{equation}  \label{eq1-02-nov-2023}
V_{g, \expeta}[u]
=
-\crochet^{2 \expeta} \Big(
g^{00} |\del_t u|^2 - g^{ab} \del_au \, \del_bu, \ 2 g^{a\beta} \del_t u  \del_\beta u \Big), 
\end{equation}
which depends upon $g$ as well as the weight $\crochet^\expeta$, we easily find the energy identity
\begin{equation}\label{eq2-26-dec-2023}
\aligned 
\dive  V_{g, \expeta}[u] + \Omega_{g, \expeta}[u]
= G_{g, \expeta}[u] - 2 \crochet^{2 \expeta} \, \del_tu \, F, 
\endaligned
\end{equation}
in which 
\begin{equation}\label{eq8-27-11-2022-M}
\aligned
\Omega_{g, \expeta}[u] 
& = 
-2\expeta \, \crochet^{-1} \aleph'({ r-t}) (-1,x^a/r) \cdot V_{g,\expeta}[u]
\\
& = 2 \, \expeta \, \crochet^{2\expeta-1} \aleph'({ r-t}) \, \Big(g^{\N ab}\delsN_au\delsN_bu - H^{\N00} |\del_tu|^2\Big), 
\\ 
G_{g, \expeta}[u] 
& =  -  \del_tH^{00} | \crochet^\expeta \del_t u|^2 + \del_tH^{ab} \crochet^{2 \expeta} \del_au\del_b u 
- 2 \crochet^{2 \expeta}  \del_aH^{a\beta} \del_t u\del_\beta  u. 
\endaligned
\end{equation}
\end{subequations}
Based on the notation above, we introduce
$$
\Eenergy_{g,\expeta}(s,u) := \int_{\Mcal_s}V_{g,\expeta}[u]\cdot n_s \, d\sigma_s
= \int_{\Mcal_s}V_{g,\expeta}[u]\cdot \Big(1,\frac{-\xi(s,r)x^a}{(s^2+r^2)^{1/2}}\Big) dx.
$$
and its restrictions in the hyperbolic domain and the merging-Euclidean domain:
$$
\Eenergy_{g,\expeta}^{\Hcal}(s,u) = \int_{\ME_s}V_{g,\expeta}[u]\cdot n_s \, d\sigma_s,
\quad
\Eenergy_{g,\expeta}^{\ME}(s,u) = \int_{\MME_s}V_{g,\expeta}[u]\cdot n_s \, d\sigma_s.
$$
Recalling the definition of the weight $\crochet$, the energy $\Eenergy_{g,\expeta}^{\Hcal}(s,u)$ defined in the hyperbolic region $\MH_{[s_0,s_1]}$ is equivalent to 
$$
\int_{\MH_s}V_g[u]\cdot n_s d\sigma_s,\quad\text{with } V_g[u] = 
- \Big(
g^{00} |\del_t u|^2 - g^{ab} \del_au \, \del_bu, \ 2 g^{a\beta} \del_t u  \del_\beta u \Big).
$$
This is in fact independent of $\expeta$. We thus denote by $\Eenergy_g^{\Hcal}(s,u)$ the energy in the hyperboloidal region.

In turn we arrive at a weighted energy estimate associated with the Euclidean-hyperboloidal foliation by integrating \eqref{eq2-26-dec-2023} in the region $\Mcal_{[s_0,s_1]}$:  
\begin{equation} \label{eq1-27-11-2022-M} 
\aligned
&\Eenergy_{g,\expeta}(s_1,u) - \Eenergy_{g,\expeta}(s_0,u)
+ 2\expeta \int_{s_0}^{s_1}\int_{\Mcal_s} \aleph'({ r-t}) \crochet^{2\expeta-1}\   g^{\N ab} \delsN_au\delsN_bu 
 \ Jdxds
\\
&= \int_{s_0}^{s_1}\int_{\Mcal_s} \Big(G_{g, \expeta}[u] + 2\expeta\aleph'({ r-t})\crochet^{2\expeta-1}  \HN^{00}  |\del_t u|^2   \Big)\ Jdx ds
- 2\int_{s_0}^{s_1}\int_{\Mcal_s}  \del_t u F  \, \crochet^{2 \expeta} \, Jdxds. 
\endaligned
\end{equation}

On the other hand,  we have also a weighted energy estimate in the Euclidean-merging domain, stated as follows. 
Any solution $u: \MME_{[s_0,s_1]} \mapsto \RR$ to the wave equation \eqref{eq1-26-dec-2023} with  
right-hand side $F: \MME_{[s_0,s_1]} \mapsto \RR$ satisfies  
\begin{equation}
\aligned
&\Eenergy_{g,\expeta}^{\ME}(s_1,u) - \Eenergy_{g,\expeta}^{\ME}(s_0,u) + \Eenergy^{\Lcal}_{g}(s_1,u;s_0)  
+  2 \expeta \int_{s_0}^{s_1}\int_{\MME_s} g^{\N ab} \delsN_au\delsN_bu  \crochet^{2\expeta-1} \aleph'({ r-t}) \, Jdxds
\\
&=
\int_{s_0}^{s_1}\int_{\MME_s} \Big(G_{g, \expeta}[u] + 2\expeta \aleph'({ r-t})\crochet^{2\sigma-1} \HN^{00}  |\del_t u|^2   \Big) \, Jdxds 
-2 \int_{s_0}^{s_1}\int_{\MME_s} \crochet^{2 \expeta} \del_t u F \, Jdxds, 
\endaligned
\end{equation}
in which the latter integral is bounded by 
$
\int_{s_0}^{s_1} \, (\Eenergy_{\sigma,c}^{\ME}(s,u))^{1/2} \, \big\|J \, \zeta^{-1} \crochet^{\sigma} f\big\|_{L^2(\MME_s)} \, ds
$ and the conical boundary term
\begin{equation} \label{equa-Ebfgc}
\Eenergy_{g}^{\Lcal}(s_1,u;s_0)  
= 2^{2\expeta}\int_{\{r=t-1\}\atop\rhoH(s_0)\leq r\leq \rhoH(s_1)}- \HN^{00} |\del_t u|^2 + g^{\N ab} \delsN_au\delsN_bu \, dx,
\end{equation} 
where the numerical coefficient $2^{2\sigma} = \crochet^{2\expeta}\big|_{\{r=t-1\}}$.
Furthermore, we also have a weighted energy estimate in the asymptotically hyperboloidal domain, stated as follows. 
Any solution $u: \MH_{[s_0,s_1]} \mapsto \RR$ to the wave equation \eqref{eq1-26-dec-2023}
with right-hand side $F: \MH_{[s_0,s_1]} \mapsto \RR$ satisfies 
\begin{equation}\label{equa-light-energy}
\aligned
& \Eenergy^\H_{g}(s_1,u) -\Eenergy_{g}^{\H}(s_0,u) - \Eenergy^{\Lcal}_{g}(s_1,u;s_0) 
+ 2\expeta \int_{s_0}^{s_1}\int_{\MH_s} \aleph'({ r-t}) \crochet^{2\expeta-1}\   g^{\N ab} \delsN_au\delsN_bu 
\ (s/t)dxds
\\
& =  \int_{s_0}^{s_1}\int_{\MH_s} \Big(G_{g, \expeta}[u] + 2\expeta \aleph'({ r-t})\crochet^{2\sigma-1} \HN^{00}  |\del_t u|^2   \Big)(s/t)\,dxds 
-2 \int_{s_0}^{s_1}\int_{\MH_s} \crochet^{2 \expeta} \del_t u F \, (s/t)dxds.
\endaligned
\end{equation}
Here, $\aleph'(r-t)$ and $\crochet$ are uniformly bounded in $\MH_{[s_0,s_1]}$, and the estimate in fact independent of $\expeta$.

For convenience in the discussion, we introduce the energy for the metric components 
\begin{equation}
\EMenergy_{g,\expeta}(s,u) =: 
\int_{\Mcal_s}V_{g,\expeta}[u]  \cdot n_s\,d\sigma_s, 
\end{equation}
together with 
\begin{equation}
\EMenergy_{g,\expeta}^{\ME}(s,u) := \int_{\MME_s}V_{g,\expeta}[u]  \cdot n_s\,d\sigma_s,
\quad
\EMenergy_g^{\Hcal}(s,u) := \int_{\MH_s}V_g[u]\cdot n_s d\sigma_s.
\end{equation}
For the effective curvature, we introduce
\begin{equation}\label{eq1-01-dec-2023}
\ECenergy_{g,\expeta,\coeffi}(s,\phi) :=  \int_{\Mcal_s} 3V_{g,\expeta}[\coeffi\phi]\cdot n_s + \coeffi \crochet^{2\expeta} \Ufrak'(0)\phi^2\, d\sigma_s, 
\end{equation}
while $\ECenergy_{g,\coeffi}^{\Hcal}(s,\phi)$ and $\ECenergy_{g,\expeta,\coeffi}^{\ME}(s,\phi)$ denote its restriction in the hyperboloidal and merging-Euclidean regions. Finally, we introduce
\begin{equation}
\ESenergy_{g,\expeta}(s,\psi) = \int_{\Mcal_s} V_{g,\expeta}[\psi]\cdot n_s + \crochet^{2\expeta}V'(0)\psi^2\, d\sigma_s
\end{equation}
for the scalar field, as well as $\ESenergy_g^{\Hcal}(s,\psi), \ESenergy_{g,\expeta}^{\EM}(s,\psi)$ for its restrictions. We also have following quantities for the conical boundary terms:
\begin{equation}
\ECenergy_{g,\coeffi}^{\Lcal}(s_1,\phi;s_0) := 3\Eenergy_{g}^{\Lcal}(s_1,\coeffi \phi;s_0) 
+ 2^{2\sigma}\coeffi\int_{\rhoH(s_0)\leq r\leq \rhoH(s_1)}\Ufrak'(0)\phi^2\big|_{\{r=t-1\}} dx 
\end{equation}
and
\begin{equation}
\ESenergy_{g}^{\Lcal}(s_1,\psi;s_0) := \Eenergy_{g}^{\Lcal}(s_1,\psi;s_0) + 2^{2\expeta}\int_{\rhoH(s_0)\leq r\leq \rhoH(s_1)} V'(0)\psi^2\big|_{\{r=t-1\}}dx.
\end{equation}

When $g = \eta$ the flat Minkowski metric, the above curved energies reduce to the following flat ones:
\begin{equation}
\aligned
\EMenergy_{\expeta}(s,u) 
&=: 
 \int_{\Mcal_s} \Big(\zeta^2|\del_t u|^2 + \sum_a |\delsEH_au|^2  \Big) \,
\crochet^{2\expeta} \, dx, 
\\
&
= \int_{\Mcal_s}
\Big(\zeta^2 \sum_a|\del_au|^2 
+ (\del_r \fTime(s,r))^2 
\sum_{a<b} |\Omega_{ab}u|^2 
+ |\delsEH_ru|^2 \Big) \, \crochet^{2\expeta} \, dx
\endaligned
\end{equation}
Furthermore, we set
\begin{equation}
\aligned
\ECenergy_{\expeta,\coeffi}(s,\phi) 
&: =  \int_{\Mcal_s} \Big(3\zeta^2|\coeffi\del_t \phi|^2 + 3\sum_a |\coeffi\delsEH_a\phi|^2 + \coeffi\Ufrak'(0)\phi^2 \Big) \,
\crochet^{2\expeta} \, dx, 
\\
&
= \int_{\Mcal_s} \Big(3\zeta^2 \sum_a|\coeffi \del_au|^2 
+ 3(\del_r \fTime(s,r))^2 
\sum_{a<b} |\coeffi\Omega_{ab}u|^2 
+ 3|\coeffi \delsEH_r u|^2 + \coeffi\Ufrak'(0)\phi^2\Big) \, \crochet^{2\expeta} \, dx
\endaligned
\end{equation}
for the effective curvature, and
\begin{equation}
\aligned
\ESenergy_{\expeta}(s,\psi) 
&: =  \int_{\Mcal_s} \Big(\zeta^2|\del_t \psi|^2 + \sum_a |\delsEH_a\psi|^2 + V'(0)\psi^2\Big) \,
\crochet^{2\expeta} \, dx, 
\\
&
= \int_{\Mcal_s}
\Big(\zeta^2\sum_a|\del_au|^2 
+ (\del_r \fTime(s,r))^2 
\sum_{a<b} |\Omega_{ab}u|^2 
+ |\delsEH_r u|^2 + V'(0)\psi^2 \Big) \, \crochet^{2\expeta} \, dx. 
\endaligned
\end{equation}
Then parallel to \eqref{eq1-27-11-2022-M}, we also have the following energy estimates.
\begin{proposition}
Consider sufficiently regular solutions $u,\phi$ and $\psi$ to the following wave/Klein-Gordon equations within in the domain $\Mcal_{[s_0,s_1]}$:
\begin{subequations}
\begin{equation}
\BoxChapeaud_\gconf u = S_m,
\end{equation}
\begin{equation}
3 \coeffi\, \BoxChapeaud_\gconf \phi - \Ufrak'(0)\phi = S_c,
\end{equation}
\begin{equation}
\BoxChapeaud_\gconf \psi - V'(0)\psi = S_s. 
\end{equation}
Then the following energy identities hold for the merging-Euclidean region:
\end{subequations}
\begin{subequations}
\begin{equation}
\aligned
&\EMenergy_{g,\expeta}^{\ME}(s_1,u) - \EMenergy_{g,\expeta}^{\ME}(s_0,u) + \EMenergy^{\Lcal}_{g}(s_1,u;s_0)  
\\
& \qquad +  2 \expeta \int_{s_0}^{s_1}\int_{\MME_s} g^{\N ab} \delsN_au\delsN_bu  \crochet^{2\expeta-1} \aleph'({ r-t}) \, Jdxds
\\
& = 
\int_{s_0}^{s_1}\int_{\MME_s} \Big(G_{g, \expeta}[u] + 2\expeta \aleph'({ r-t})\crochet^{2\sigma-1} \HN^{00}  |\del_t u|^2   \Big) \, Jdxds 
-2 \int_{s_0}^{s_1}\int_{\MME_s} \crochet^{2 \expeta} \del_t u S_m \, Jdxds, 
\endaligned
\end{equation}
\begin{equation}
\aligned
&\ECenergy_{g,\expeta,\coeffi}^{\ME}(s_1,\phi) - \ECenergy_{g,\expeta,\coeffi}^{\ME}(s_0,\phi) + \ECenergy_{g,\coeffi}^{\Lcal}(s_1,\phi;s_0) 
\\
& \qquad + 6\sigma\int_{\MME_s} \big( \coeffi^2g^{\N ab} \delsN_a\phi\delsN_b\phi + \frac{1}{3}\coeffi\Ufrak'(0)\phi^2
\big)  \aleph'({ r-t}) \crochet^{2\expeta-1}\,Jdxds
\\
& =  3\int_{s_0}^{s_1}\int_{\MME_s} \Big(G_{g, \expeta}[\coeffi \phi] + 2 \expeta\aleph'({ r-t})\crochet^{2\expeta-1} \HN^{00}  |\coeffi\del_t \phi|^2   \Big)\ Jdxds 
- 2\int_{s_0}^{s_1}\int_{\MME_s}  \coeffi \crochet^{2 \expeta}\del_t \phi S_c \, Jdxds, 
\endaligned
\end{equation}
\begin{equation}
\aligned
&\ESenergy_{g,\expeta}^{\ME}(s_1,\psi) - \ESenergy_{g,\expeta}^{\ME}(s_0,\psi) + \ESenergy_{g,\expeta}^{\Lcal}(s_1,\phi;s_0)
\\
& \qquad + 2\sigma\int_{\Mcal_s} \big(g^{\N ab} \delsN_a\psi\delsN_b\psi + V'(0)\psi^2\big)  \aleph'({ r-t}) \crochet^{2\expeta-1}\ Jdxds
\\
&= \int_{\Mcal_s} \Big(G_{g, \expeta}[\psi] + 2\expeta\aleph'({ r-t})\crochet^{2\expeta-1}  \HN^{00}  |\del_t\psi|^2  \Big)\ Jdx ds
- 2\int_{\Mcal_s}  \crochet^{2 \expeta}\del_t\psi S_s  \, Jdxds.
\endaligned
\end{equation}
\end{subequations}
Meanwhile, the following energy identities hold for the hyperbolic region:
\begin{subequations}
\begin{equation}
\aligned
&\EMenergy_{g}^{\H}(s_1,u) - \EMenergy_{g}^{\H}(s_0,u) - \EMenergy^{\Lcal}_{g}(s_1,u;s_0)  
\\
& \qquad +  2 \expeta \int_{s_0}^{s_1}\int_{\MH_s}\aleph'({ r-t})\crochet^{2\expeta-1}  g^{\N ab} \delsN_au\delsN_bu   \, (s/t)dxds
\\
& = 
\int_{s_0}^{s_1}\int_{\MH_s} \Big(G_{g, \expeta}[u] + 2\expeta \aleph'({ r-t})\crochet^{2\sigma-1} \HN^{00}  |\del_t u|^2   \Big) \, (s/t)dxds 
-2 \int_{s_0}^{s_1}\int_{\MH_s} \crochet^{2 \expeta} \del_t u S_m \, (s/t)dxds, 
\endaligned
\end{equation}
\begin{equation}
\aligned
&\ECenergy_{g,\coeffi}^{\H}(s_1,\phi) - \ECenergy_{g,\coeffi}^{\H}(s_0,\phi) - \ECenergy_{g,\coeffi}^{\Lcal}(s_1,\phi;s_0) 
\\
& \qquad + 6\sigma\int_{\MH_s} \big( \coeffi^2g^{\N ab} \delsN_a\phi\delsN_b\phi + \frac{1}{3}\coeffi\Ufrak'(0)\phi^2
\big)  \aleph'({ r-t}) \crochet^{2\expeta-1}\,(s/t)dxds
\\
& =  3\int_{s_0}^{s_1}\int_{\MH_s} \Big(G_{g, \expeta}[\coeffi \phi] + 2\expeta\aleph'({ r-t})\crochet^{2\expeta-1}  \HN^{00}  |\coeffi\del_t \phi|^2   \Big)\ Jdx 
- 2\int_{s_0}^{s_1}\int_{\MH_s}  \coeffi \crochet^{2 \expeta}\del_t \phi S_c \,(s/t)dx, 
\endaligned
\end{equation}
\begin{equation}
\aligned
&\ESenergy_{g}^{\H}(s_1,\psi) - \ESenergy_{g}^{\H}(s_0,\psi) - \ESenergy_{g}^{\Lcal}(s_1,\phi;s_0)
\\
& \qquad + 2\sigma\int_{\MH_s} \big(g^{\N ab} \delsN_a\psi\delsN_b\psi + V'(0)\psi^2\big)  \aleph'({ r-t}) \crochet^{2\expeta-1}\ (s/t)dx
\\
&= \int_{\MH_s} \Big(G_{g, \expeta}[\psi] + 2\expeta\crochet^{2\expeta-1} \aleph'({ r-t}) \HN^{00}  |\del_t\psi|^2  \Big)\ (s/t)dx 
- 2\int_{\MH_s}  \crochet^{2 \expeta}\del_t\psi S_s  \,(s/t)dx.
\endaligned
\end{equation}
\end{subequations}
\end{proposition}

\begin{proof}
These are similar to the standard energy estimate for wave/Klein-Gordon equations with curved background metric. However, for $\ECenergy_{g,\expeta,\coeffi}(s,\phi)$, we apply the multiplier $-2\crochet^{2\expeta}\kappa\del_t\phi$. 
\end{proof}

\bse
Furthermore, for convenience in the presentation we also introduce the following notation: 
\begin{equation}
\aligned
\FMenergy_{g,\expeta}(s,u) &:= \big(\EMenergy_{g,\expeta}(s,u)\big)^{1/2},
\\
\FMenergy_{g,\expeta}^{\ME}(s,u) &:= \big(\EMenergy_{g,\expeta}^{\ME}(s,u)\big)^{1/2},
\quad
\\
\FMenergy_g^{\Hcal}(s,u) &:= \big(\EMenergy_g^{\Hcal}(s,u)\big)^{1/2}, 
\endaligned
\end{equation}
and 
\begin{equation}
\aligned
\FCenergy_{g,\expeta,\coeffi}(s,\psi) &:= \big(\ECenergy_{g,\expeta,\coeffi}(s,\phi)\big)^{1/2},
\\
\FCenergy_{g,\expeta,\coeffi}^{\ME}(s,\psi) &:=  \big(\ECenergy_{g,\expeta,\coeffi}^{\ME}(s,\phi)\big)^{1/2},\quad
\\
\FCenergy_{g,\coeffi}^{\Hcal}(s,\psi) &:=  \big(\ECenergy_{g,\coeffi}^{\Hcal}(s,\phi)\big)^{1/2}, 
\endaligned
\end{equation}
as well as 
\begin{equation}
\aligned
\FSenergy_{g,\expeta}(s,\psi) &:= \big(\ESenergy_{g,\expeta}(s,\psi)\big)^{1/2},
\\
\FSenergy_{g,\expeta}^{\ME}(s,\psi) &:= \big(\ESenergy_{g,\expeta}^{\ME}(s,\psi)\big)^{1/2},
\quad
\\ 
\FSenergy_g^{\Hcal}(s,\psi) &:=  \big(\ESenergy_g^{\Hcal}(s,\psi)\big)^{1/2}.
\endaligned
\end{equation}
When $g = \eta$ is the Minkowski metric, the subscript $g$ will be omitted.
\ese
Finally, for high-order norms we use a similar notation and, for instance, without specifying the relevant subscripts 
 \begin{equation} \label{equa-defEF} 
\Eenergy^N(s,u): = \sum_{\ord(Z)\leq N} \Eenergy(s,Zu),\quad 
\qquad
\Fenergy^N(s,u) :=  \big(\Eenergy^N(s,u)\big)^{1/2},
\end{equation}
where the summation over all admissible operators (cf.~the discussion near \eqref{equa-ddd}). 

 
\section{Main statement in conformal wave gauge}
\label{section=7}

\subsection{Construction of PDE initial data}

In Section~\ref{section=4}, we have introduced the initial data set in a geometric context. According to \eqref{eq1-24-dec-2023}, we need to determine the quantities listed in \eqref{eq2-24-dec-2023} form the geometric data $\big( g_0,k_0,\phi_0,\phi_1,\psi_0,\psi_1\big)$. As we have done in \cite[Section~10.3]{PLF-YM-main}, this is a standard procedure and the essential work is an application of the wave gauge conditions on the initial slice $\{t=t_0\}$. In the present article we are exactly in the same situation so we omit the calculation therein. 

Then we need to bound the energy quantities such as $\FMenergy_{\expeta}^N(s_0,u)$. This can be done by applying the same calculation in the proof of \cite[Proposition~10.4, 10.5]{PLF-YM-main}. We will only explain how to estimate the following quantities by \eqref{eq3-23-dec-2023}.
$$
\|\coeffi \la r\ra^{\mu}Z\del \phi(t_0,\cdot)\|_{L^2(\RR^3)},\quad
\|\coeffi^{1/2} \la r\ra^{\mu}Z\phi(t_0,\cdot)\|_{L^2(\RR^3)}.
$$
This will lead us to the estimate on $\FCenergy_{\mu}^N(s_0,\phi)$. In fact due to our choice of on the lapse and shifts functions in \cite[Section~10]{PLF-YM-main},
$$
\phi(t_0,x) = \phi_0(x),\quad \del_t\phi(t_0,x) = M^{\star}\phi_1(x) 
$$
where $M^{\star}$ is the lapse of the reference metric, and, by assumption \eqref{equa-31-12-20}, is a uniformly bounded quantity (independent of $\coeffi$). For example when $g^{\star}$ is taken to be \eqref{equa-defineMS-new} , then 
$$
M^{\star} =  g^{\star}_{00} + 1 \simeq \frac{2m}{r+m},\quad m>0.
$$
Thus we apply the proof in \cite[Appendix~F]{PLF-YM-main} and obtain the bounds on initial energy quantities associate to $\phi$ as well as to $u$ and $\psi$
\begin{equation}\label{eq4-23-dec-2023}
\FMenergy_{\expeta}^N(s_0,u) + \FCenergy_{\mu,\coeffi}^N(s_0,\phi) + \FSenergy_{\mu}^N(s_0,\psi) \lesssim \eps.
\end{equation}


\subsection{Main statements in wave gauge} 
\label{section-555}

Once the initial energy quantities are bounded by the norms of geometric initial data, we are ready to state our main result in a PDEs setting.

\begin{theorem}[Nonlinear stability. Formulation in conformal wave gauge]  
\label{theo-main-result} 
Let $g^{\star}$ be a $(\lambda,\theta, \epss,N,\ell)-$ light-bending reference spacetime defined in $\RR^{1+3}_+$ with regularity order ${ N\geq 20}$ and 
for a  sufficiently small $\epss$. The parameters $\lambda,\expeta,\mu$ satisfying \eqref{eq1-20-01-2022} and \eqref{equa-all-conditions-exponents}. Then there exists a small constant $\eps_0 > 0$ such that for all initial data satisfying  
\begin{equation} \label{eqergy-boundMas-initial} 
\aligned
\FMenergy_{\expeta}(s_0, Z u_{\alpha\beta} ) + \FCenergy_{\coeffi,\mu}(s_0, Z \phi) + \FSenergy_{\mu}(s_0,Z \psi) 
& \lesssim    
\eps \leq \eps_0, 
\quad &&   
\ord(Z) \leq N,
\endaligned
\end{equation}  
the Cauchy problem \eqref{eq1-24-dec-2023} admits a global-in-time solution and the following estimates hold for all $s \geq s_0$
\begin{equation} \label{eqergy-boundMas} 
\aligned
\FMenergy_{\expeta}(s, Z u_{\alpha\beta} ) 
+ {s^{-1/2} \FCenergy_{\mu,\coeffi}(s, Z \phi)}
+ {s^{-1/2} \FSenergy_{\mu}}(s, Z\psi)
& \lesssim    
\eps \, s^{\delta},
\quad &&&& \ord(Z) \leq N-5,
\\
\FMenergy_{\expeta}(s, Z u_{\alpha\beta} ) 
+ {\FCenergy_{\mu,\coeffi}(s, Z \phi)}
+ { \FSenergy_{\mu}}(s, Z\psi)
& \lesssim        
\eps \, s^\delta,
\quad &&&& \ord(Z)  \leq N-9. 
\endaligned
\end{equation}  
Furthermore, the solution metric $g_{\alpha\beta} = g^{\star}_{\alpha\beta} + u_{\alpha\beta}$ still enjoys the light-bending property.
\end{theorem} 

\begin{remark}
The proposed range for the parameters $\lambda,\expeta,\mu,\theta$ is not optimal but suffices to encompass a broad class of spacetimes with Schwarzschild-type asymptotic behavior. The conditions will be weakened in in future work.
\end{remark}


\section{Analysis in the Euclidean-merging domain}
\label{section-Euclidean}

\subsection{New nonlinear terms}

Comparing \eqref{eq1-24-dec-2023} with the standard Einstein-Klein-Gordon system written in harmonic gauge:
\begin{equation}\label{eq1-02-dec-2023}
\aligned
\BoxChapeaud_\gconf \gconf_{\alpha\beta}
& =\, \Fbb[g]_{\alpha\beta}  + \Mbb[\psi,g]_{\alpha\beta},
\\ 
\BoxChapeaud_\gconf \psi - V'(\psi) 
& = 0
\endaligned
\end{equation}
with 
$$
 \Mbb[\psi,g]_{\alpha\beta} = - 16 \pi \, \big(\del_\alpha \psi \del_\beta \psi + V(\psi)   \, \gconf_{\alpha \beta} \big)
$$
We observe that there are some newly appearing nonlinear terms, which concerning the interaction $\Gbb_{\alpha\beta}$ of the effective curvature field $\phi$ on the metric component $g$, and the modification on the mass-metric coupling term $\Mbb_{\alpha\beta}$. These terms are quadratic with respect to the Klein-Gordon components ($\phi,\psi$), thus their contributions are integrable when we do energy estimates. However, when we do pointwise estimates on $u$ via fundamental solution, in order to obtain $t^{-1}$ decay rate, we need to bound these terms in a more precise way. To cope with this, we need the decay estimate based on integration along rays on two Klein-Gordon components, which is also similar to what we have done for the Einstein-Klein-Gordon system. Importantly, for the effective curvature field we need to be precise about the dependency of the estimates upon the coefficient $\coeffi$. However, thanks to the structure of $\Gbb_{\kappa}, \Mbb_{\kappa}$ and the right-hand side of the equation of $\psi$ where $\coeffi$ is involved in with the right power at each place, it is possible to establish estimates which are uniform with respect to $\coeffi$.   

In the remaining of this section, we sketch of the proof by presenting the estimates associated with comparatively newly terms. 


\subsection{Bootstrap assumptions and direct estimates}

We make the following bootstrap assumptions on the time interval $[s_0,s_1]$:
\begin{subequations}\label{eqs1-14-01-2021}  
\begin{equation}\label{eq1-14-01-2021}
\FMenergy_{\expeta}^{\ME,N}(s,u) + s^{-1} \, \FCenergy_{\mu,\coeffi}^{\ME,N}(s, \phi) + s^{-1} \, \FSenergy_{\mu}^{\ME,N}(s, \psi) \leq(\epss+C_1\eps) \, s^{\delta},
\end{equation}
\begin{equation}\label{eq2-14-01-2021}
\FMenergy_{\expeta}^{\ME,N-5}(s,u) +\FCenergy_{\mu}^{\ME,N-5}(s, \phi) +\FSenergy_{\mu}^{\ME,N-5}(s, \psi) \leq(\epss+C_1\eps) \, s^{\delta}, 
\end{equation}
while, in the course of our analysis, we will also control the spacetime integral\footnote{which
we do not need to include in the set of bootstrap assumptions}
\begin{equation}\label{equa-new-spacetime-bound} 
{\mathscr G}_\kappa^{\ME,p,k}(s_0, s,u)
:=
\sum_{\ord(Z)\leq p\atop \rank(Z)\leq k}\int_{s_0}^{s}
\int_{\MME_\tau}\crochet^{2\kappa-1} |\delsN Z u|^2 J \, dxd\tau
\leq (\epss+C_1\eps)^2 s^{2\delta}.
\end{equation}
\end{subequations} 

Furthermore, near the light cone, we assume the light-bending condition
\begin{equation}\label{eq3-27-05-2020} 
\inf_{\Mscr^\near_{\ell, [s_0,s_1]}} (- \hN{}^{00} )\geq 0.
\end{equation}

The purpose is to establish the following {\bf improved} estimates:
\begin{subequations}\label{eqs'-ext}
\begin{equation}\label{eq5'-03-05-2020}
\Fenergy_{\kappa}^{\ME,N}(s,u) + s^{-1} \, \Fenergy_{\mu,c}^{\ME,N}(s, \phi) \leq\frac{1}{2}(\epss+C_1\eps) \, s^{\delta},
\end{equation}
\begin{equation}\label{eq6'-03-05-2020}
\Fenergy_{\kappa}^{\ME,N-5}(s,u) +\Fenergy_{\mu,c}^{\ME,N-5}(s, \phi) \leq \frac{1}{2}(\epss+C_1\eps) \, s^{\delta},
\end{equation}
\end{subequations}
\begin{equation}\label{eq3'-27-05-2020}
\inf_{\Mscr^\near_{\ell, [s_0,s_1]}} (- r \, \hN{}^{00} ) 
\geq  {1\over 2} \epss.
\end{equation}

The bootstrap assumptions lead to the following estimates on the $L^2$ norms:
\begin{subequations}\label{eqs3-08-dec-2023}
\begin{equation}\label{eq7a-03-05-2020}
\| \crochet^\kappa \zeta \, | \del u|_N \|_{L^2(\MME_s)}
+
\| \crochet^\kappa |\delts u|_N\|_{L^2(\MME_s)} 
\lesssim (\epss+C_1\eps) \, s^{\delta}, 
\end{equation}
\begin{equation}\label{eq1-12-05-2020}
\| \crochet^{-1+\kappa} |u|_{p,k} \|_{L^2(\MME_s)} 
\lesssim \delta^{-1} \, \Fenergy_{\kappa}^{\ME,p,k}(s,u) + \Fenergy_{\kappa}^{0}(s,u)
\lesssim \delta^{-1}  (\epss+C_1\eps)  \, s^{\delta},
\end{equation}
\begin{equation}\label{eq7b-03-05-2020}
\aligned
& \| \crochet^\mu \zeta \, |\coeffi\del \phi|_p\|_{L^2(\MME_s)} + 
\| \crochet^\mu |\coeffi\delts \phi|_p\|_{L^2(\MME_s)} +
\| \crochet^\mu |\Ufrak'(0)\coeffi^{1/2}\phi|_p\|_{L^2(\MME_s)}
\\
&
\lesssim
(\epss+C_1\eps) \, 
\begin{cases} 
s^{1+\delta}, \quad & p=N,
\\
s^{\delta}, & p=N-5,
\end{cases}
\endaligned
\end{equation}
\begin{equation}\label{eq2-08-dec-2023}
\aligned
& \| \crochet^\mu \zeta \, |\del \psi|_p\|_{L^2(\MME_s)} + 
\| \crochet^\mu |\delts \psi|_p\|_{L^2(\MME_s)} +
\| \crochet^\mu |V'(0)\psi|_p\|_{L^2(\MME_s)}
\\
&
\lesssim
(\epss+C_1\eps) \, 
\begin{cases} 
s^{1+\delta}, \quad & p=N,
\\
s^{\delta}, & p=N-5,
\end{cases}
\endaligned
\end{equation}
\end{subequations}
We apply the global Sobolev inequalities Proposition~\ref{Sobolev-ext}, and obtain
\begin{subequations}\label{eqs1-08-dec-2023}
\begin{equation}
r\crochet^{\expeta}|\del u|_{N-3} + r^{1+\expeta}|\delsN u|_{N-3}\lesssim (\epss+C_1\eps)s^{\delta},
\end{equation}
\begin{equation}
r \, \crochet^{\expeta-1} \,  |u|_{N-2} \lesssim \delta^{-1} \, (\epss+C_1\eps)  \, s^{\delta}.
\end{equation}
\end{subequations}
This leads us to
\begin{equation}\label{eq6-09-dec-2023}
r\crochet^{\expeta}|\del h|_{N-3} + \delta r\crochet^{-1+\expeta}|h|_{N-2}\lesssim (\epss+C_1\eps)s^{\delta}.
\end{equation}

Throughout we assume the Class A conditions in \eqref{equa-31-12-20}. 
We find 
\begin{subequations}
\begin{equation}
r\crochet^{\mu}|\coeffi \del \phi|_{p-3} + r^{1+\mu}|\coeffi \delsN \phi|_{p-3} +  r\crochet^{\mu}\Ufrak'(0)|\coeffi^{1/2} \phi|_{p-2}
\lesssim 
\begin{cases}
\aligned
&(\epss+C_1\eps)s^{1+ \delta},\quad &&p=N,
\\
&(\epss+C_1\eps)s^{\delta},\quad &&p=N-5,
\endaligned
\end{cases}
\end{equation}
\begin{equation}
r\crochet^{\mu}|\del \psi|_{p-3} + r^{1+\mu}|\delsN \psi|_{p-3} +  r\crochet^{\mu}V'(0)|\psi|_{p-2}
 \lesssim 
\begin{cases}
\aligned
&(\epss+C_1\eps)s^{1+ \delta},\quad &&p=N,
\\
&(\epss+C_1\eps)s^{\delta},\quad &&p=N-5.
\endaligned
\end{cases}
\end{equation}
\end{subequations}  
Then we substitute the above bounds into the equations
$$
3\coeffi\Box \phi - \Ufrak'(0)\phi = -3\coeffi h^{\mu\nu}\del_{\mu}\del_{\nu}\phi ,\quad
\Box \psi - V'(0)\psi = -h^{\mu\nu}\del_{\mu}\del_{\nu}\psi
$$
and recall Proposition~\ref{lem 1 d-KG-e}, in order to obtain 
\begin{subequations}\label{eqs2-08-dec-2023}
\begin{equation}\label{eq2-09-dec-2023}
r \, \crochet^\mu \, \Ufrak'(0)|\phi|_{p-4} 
\lesssim 
(\epss+C_1\eps) r^{-1}\crochet 
\begin{cases}
s^{1+2\delta},  \quad  
& p=N,
\\
s^{2\delta}, \quad 
&  p=N-5, 
\end{cases}
\end{equation}
together with 
\begin{equation}
r \, \crochet^\mu \, V'(0) |\psi|_{p-4} 
\lesssim 
(\epss+C_1\eps)r^{-1}\crochet
\begin{cases}
s^{1+2\delta},  \quad  
& p=N,
\\
s^{2\delta}, \quad 
&  p=N-5. 
\end{cases}
\end{equation}
\end{subequations}
The pointwise estimates \eqref{eqs1-08-dec-2023} and \eqref{eqs2-08-dec-2023} are refered to as the {\sl direct pointwise bounds}. They are sufficient for deriving $L^2$ estimates for the non-critical terms, including $\Gbb_{\coeffi}[\phi,g]_{\alpha\beta}$ and $\Mbb_{\coeffi}[\phi,\psi,g]$ which we treat below.


\subsection{$L^2$ and pointwise estimates on $\Gbb_{\coeffi}[\phi,g]_{\alpha\beta}$ and $\Mbb_{\coeffi}[\phi,\psi,g]$}

We will establish the following estimates:
\begin{equation}\label{eq4-09-dec-2023}
\|J\zeta^{-1}\crochet^{\expeta}|\Gbb_{\coeffi}[\phi,g]|_{N}\|_{L^2(\MME_s)} + \|J\zeta^{-1}\crochet^{\expeta}|\Mbb_{\coeffi}[\phi,\psi,g]|_{N}\|_{L^2(\MME_s)}\lesssim (\epss + C_1\eps)^2s^{-1-\delta}.
\end{equation}
\begin{equation}\label{eq3-27-01-2021}
(r \, \crochet^\kappa)^2
\big(|\Gbb_{\coeffi}[\phi,g]|_p + |\Mbb_{\coeffi}[\phi,\psi,g]|_p\big)\lesssim
(\epss+C_1\eps)^2
r^{-3}\crochet^{1-2\mu}s^{1+3\delta}. 
\end{equation}

\begin{proof} These bounds are sufficient to make the mechanism in \cite{PLF-YM-main} work. For the proof, we only need to substitute \eqref{eqs3-08-dec-2023},  \eqref{eqs1-08-dec-2023} and \eqref{eqs2-08-dec-2023} into the estimates in Lemma~\ref{lem1-08-dec-2023}. This is in fact the same calculation that we have done in \cite{PLF-YM-main} on the Klein-Gordon quadratic terms coupled in Einstein equation. We need to be careful about the power of the coefficient $\coeffi$ in each term. Fortunately, in $\Gbb_{\coeffi}[\phi,g]$, each factor $\del \phi$ is multiplied by a $\coeffi$ and each factor $\phi$ is multiplied by a $\coeffi^{1/2}$ (see in \eqref{eq1-09-dec-2023}), which is compatible with their expression in the energy $\FCenergy$. For the term $\Mbb_{\coeffi}[\phi,\psi,g]$, the situation is similar. The factor $\del\phi$ is multiplied by a $\coeffi$ (see \eqref{eq2-03-dec-2023}). The desired estimates are thus established exactly as in \cite{PLF-YM-main}.
\end{proof}


\subsection{Sharp decay and energy estimates for the metric tensor}  

 Once we have established \eqref{eq4-09-dec-2023} and \eqref{eq3-27-01-2021}, we can adopt the setup described in \eqref{equa-31-12-20}, namely {\bf Class A} of \cite{PLF-YM-main}.
 The contributions of Klein-Gordon components on the metric components enjoy the same $L^2$ and pointwise estimates. We thus arrive directly on the following {\sl sharp decay estimates} established in \cite{PLF-YM-main}. 

\begin{proposition} Assume the conditions {\bf Class A}  in \eqref{equa-31-12-20}. 
1.  
The metric and the perturbation satisfy
\begin{equation}\label{eq7-26-03-2021}
\aligned
|\del u|_{N-4,k} + |\del H|_{N-4,k}
& \lesssim  (\ell^{-\delta/2} + \delta^{-2})(\epss + C_1\eps) r^{-1+k\theta}\crochet^{-1/2-\delta/2} 
\, \text{ in } \Mnear_{\ell, [s_0,s_1]},
\, 0\leq k\leq N-4, 
\\
|\del\del u|_{N-5,k} + |\del\del H |_{N-5,k}& \lesssim  (\ell^{-\delta} + \delta^{-2})(\epss + C_1\eps) r^{-1+k\theta}\crochet^{-1-\delta}
\,
  \text{ in } \Mnear_{\ell, [s_0,s_1]}, 
\, 0\leq k\leq N-5.
\endaligned
\end{equation}
2. Furthermore, in the whole exterior domain, the lapse, orthogonal, and radial components of the metric enjoy the near-Schwarz\-schild decay 
\begin{equation}\label{eq3-29-03-2021}
| h^{00}, h^\rr, h^{0a} |_k\lesssim (\epss + C_1\eps)r^{-1+(k+1) \theta}
\quad \text{ in } \MME_{[s_0,s_1]}, 
\qquad 0\leq k\leq N-5.
\end{equation}
\end{proposition}

Once the above sharp decay estimates are established, we repeat the argument in  \cite{PLF-YM-main} for the energy estimate on $u$. The only difference is that we have tow extra source terms $\Gbb_{\coeffi}[\phi,g], \Mbb_{\coeffi}[\phi,\psi,g]$. However, their $L^2$ bounds \eqref{eq4-09-dec-2023} and \eqref{eq3-27-01-2021} are integrable. Thus we conclude exactly as in \cite{PLF-YM-main}.  The energy bound that we obtain can be written into the following form:
\begin{equation}\label{eq5-17-dec-2023}
\FMenergy_{g,\expeta}^{\ME, N-5}(s,u)\lesssim (\epss+C_1\eps)^2 s^{C_2(\epss+C_1\eps)+N\theta},
\end{equation}
where $C_2$ is a constant determined by the reference metric, $\delta$ and $N$.

\begin{remark}
For $|\del u|, |\del\del u|$ and $|\del H|, |\del\del H|$, when one is far from the light cone, the direct pointwise bounds are already sufficient because $\crochet \simeq r$.
\end{remark}


\subsection{Energy estimates for Klein-Gordon components}

The energy estimates on the scalar field $\psi$ is slightly different: we have a source term in the right-hand side on which we should give sufficient estimates. Based on \eqref{eqs2-08-dec-2023}, we easily get the following $L^2$ estimate:
\begin{equation}\label{eq7-09-dec-2023}
\|J\zeta^{-1}\crochet^{\mu}\coeffi |g(\del\phi,\del\psi)|_N\|_{L^2(\MME_s)}\lesssim (\epss+C_1\eps)^2s^{-1-\delta}, 
\end{equation}
which is integrable in $s$. However, we need to modify the decomposition of commutators.
 Let us recall the following result established in \cite{PLF-YM-main}.

\begin{lemma}[Decomposition of quasi-linear terms for Klein-Gordon fields]
\label{lem1-14-03-2021}
For every  $C^2$ function defined in $\MME_{[s_0,s_1]}$, with the notation
$H^\rr : =(x^ax^b/r^2)H^{ab}$ one has 
\begin{equation}\label{eq3-09-dec-2023}
\aligned
H^{\alpha\beta}\del_{\alpha}\del_{\beta}\phi 
& =  \frac{H^{00} + H^{\rr}}{1+H^{\rr}} \del_t\del_t\phi + \frac{2H^{a0}}{1+H^{\rr}}\del_a\del_t\phi 
+ \frac{H^\rr}{1+H^{\rr}} \BoxChapeaud_\gconf \phi
+ \frac{H^{ab} - H^{\rr}g_{\Mink}^{ab}}{r(1+H^{\rr})} \, \mathbb{D}_{ab}[\phi], 
\endaligned
\end{equation}
where 
$$
|\mathbb{D}_{ab}[u]|\lesssim |\del u|_{1,1}.
$$
\end{lemma}

Then we establish the following decomposition of commutators.

\begin{proposition}[Estimate on quasi-linear commutator of Klein-Gordon equation]
\label{prop1-28-03-2021}
Provided
$$
\BoxChapeaud_\gconf \phi - c^2\phi  = f,\quad 
$$
and
\begin{equation}\label{eq5-09-dec-2023}
|H|\ll 1,\quad \sum_{\circledast \in \{ \rr, 00, 0a \}} |H^\circledast|_{[p/2]}\ll 1, 
\end{equation}
then
the following estimate holds for all $Z$ with $\ord(Z)=p$ and $\rank(Z)=k$: 
\begin{equation} \label{eq6-14-03-2021}
\aligned
|[Z,H^{\alpha\beta}\del_{\alpha}\del_{\beta}]\phi| 
& \lesssim  W_{p,k}^\hard [\phi] + W_{p,k}^\easy [\phi] + W_{p,k}^{\sour}[f],  
\\
W_{p,k}^\hard [\phi]
&:=  \sum_{* \in \{ \rr, 00, 0a \}} |H^\circledast||\del\del \phi|_{p-1,k-1} 
\\
& \quad 
+ \sum_{k_1+p_2=p\atop k_1+k_2=k}  \sum_{\circledast \in \{ \rr, 00, 0a \}} 
| \LOmega  H^\circledast|_{k_1-1}\big(|\del\del \phi|_{p_2,k_2} +|\phi|_{p_2,k_2}\big)  
\\
& \quad + \sum_{p_1+p_2=p\atop k_1+k_2=k} |\del H|_{p_1-1,k_1}\big(|\del\del\phi|_{p_2,k_2} + |\phi|_{p_2,k_2}\big), 
\\
W_{p,k}^\easy [\phi]
& :=  r^{-1} |H||\del \phi|_p 
+ r^{-1}\!\!\!\!\sum_{0\leq p_1\leq p-1}\!\!\!|H|_{p_1+1} |\del \phi|_{p-p_1},
\\
W_{p,k}^{\sour}[f] &:=  \sum_{k_1+p_2=p\atop k_1+k_2=k}  
| \LOmega  H^{rr}|_{k_1-1}|f|_{p_2,k_2} 
+ \sum_{p_1+p_2=p\atop k_1+k_2=k} |\del H|_{p_1-1,k_1}|f|_{p_2,k_2},
\endaligned
\end{equation} 
\end{proposition}

\begin{proof}[Sketch of proof] Observe that $[Z,\BoxChapeaud_\gconf ]\phi = [Z,H^{\alpha\beta}\del_{\alpha}\del_{\beta}]\phi$.
Then we only need to commute each term presented in the RHS of \eqref{eq3-09-dec-2023}. We consider especially the third term:
$$
\aligned 
&  \frac{H^{rr}}{1+H^{rr}}[Z,\BoxChapeaud_\gconf]\phi + \sum_{Z_1\odot Z_2=Z\atop |Z_1|\geq 1}Z_1\Big(\frac{H^{rr}}{1+H^{rr}}\Big)\,Z_2\BoxChapeaud_\gconf\phi
\\
& = \frac{H^{rr}}{1+H^{rr}}[Z,\BoxChapeaud_\gconf]\phi + \sum_{Z_1\odot Z_2=Z\atop |Z_1|\geq 1}Z_1\Big(\frac{H^{rr}}{1+H^{rr}}\Big)\,Z_2f
+ c^2\sum_{Z_1\odot Z_2=Z\atop |Z_1|\geq 1}Z_1\Big(\frac{H^{rr}}{1+H^{rr}}\Big)\,Z_2\phi. 
\endaligned
$$
By recalling \eqref{eq5-09-dec-2023}, the first term can be absorbed in the left-hand side. The second term can be bounded by $W_{p,k}^{\hard}$ while the last one is bounded by $W_{p,k}^{\sour}$.
\end{proof}

Regarding the above result, we also need to give sufficient $L^2$ estimates on $W_{p,k}^{\sour}$ (while the remaining terms are already treated in \cite{PLF-YM-main}). For the scalar field, we have $f = \coeffi g(\del\phi,\del\psi)$. Then due to \eqref{eqs2-08-dec-2023}, \eqref{eq6-09-dec-2023} and \eqref{eq7-09-dec-2023}, the following estimate is easily checked:
\begin{equation}\label{eq8-09-dec-2023}
\|J\zeta^{-1}\crochet^{\mu} \, W_N^{\sour}[f]\|_{L^2(\MME_s)}
\lesssim (\epss+C_1\eps)^2s^{-1-\delta}.
\end{equation}
Thus the new terms appearing in the energy estimate for $\psi$ are all integrable. Then the bootstrap argument can be closed as in \cite{PLF-YM-main} for the scalar field.

For the effective curvature field $\phi$, the situation essentially the same. However, we need to consider two main points. First, as for $\psi$, we also have a source term in right-hand side. But \eqref{eqs2-08-dec-2023} still gives integrable $L^2$ bounds, because this term is also a KG$\times$KG quadratic term. Second, when we compute the commutator $[Z,3\coeffi \BoxChapeaud_\gconf]$,
we need to take care of the factor $\coeffi$. In fact \eqref{eq6-14-03-2021} can be written as
\begin{equation}
|[Z,\coeffi\BoxChapeaud_\gconf]\phi|\lesssim W_{p,k}^{\hard}[\coeffi \phi] + W_{p,k}^{\easy}[\coeffi \phi] + W_{p,k}^{\sour}[\coeffi f], 
\end{equation}
where 
$$
f:= 3\coeffi \BoxChapeaud_\gconf \phi - \Ufrak'(0)\phi =  -8\pi e^{-\coeffi\phi} \gconf(\del\psi,\del\psi) - 32\pi e^{-2\coeffi\phi}V(\psi).
$$
Observe that in $W_{p,k}^{\hard}[\coeffi \phi]$ and $W_{p,k}^{\easy}[\coeffi \phi]$, each factor $\del \phi$ is multiplied by $\coeffi$ and each factor $\phi$ is multiplied at least by $\coeffi^{1/2}$, which is compatible with the energy bounds \eqref{eq7b-03-05-2020}. Furthermore, the pointwise bounds \eqref{eq2-09-dec-2023} are uniform with respect to $\coeffi$. Thus we repeat again the argument in \cite{PLF-YM-main} 
 for a scalar field 
and close the bootstrap argument.  The energy bound that we obtain is the same to \cite[Section 18]{PLF-YM-main}, and can be written as
\begin{equation}\label{eq6-17-dec-2023}
\FCenergy_{g,\coeffi,\mu}^{\ME, N-5}(s,u) + \FSenergy_{g,\mu}^{\ME,N-5}(s,u) \lesssim (\epss+C_1\eps)^2 s^{C_2(\epss+C_1\eps)+N\theta},
\end{equation}
where $C_2$ is a constant determined by $\delta,N$ and the reference metric.


\subsection{Estimate for the boundary energy}

The light cone energy $\Ebf^{\Lcal}_{g}(s,u; s_0) $ was introduced in \eqref{equa-light-energy}. When we perform an energy estimate in the hyperboloidal domain (cf.~next section),  
we should be careful with
 the treatment of the boundary terms along $\Lscr_{[s_0,s_1]} :=  \big\{ r=t-1, \, {s_0^2-1\over 2} \leq t \leq {s_0^2+1 \over 2}\big\}$, which brings in the influence from the Euclidean-merging domain. Fortunately, these boundary terms are already bounded when doing energy estimates in the Euclidean-merging domain. Let us point out the following estimate deduced from the energy estimate in the Euclidean-merging domain:
$$
\Eenergy^{\Lcal}_g(s,Z u;s_0) \lesssim \Eenergy^{\ME}_g(s,Zu) + \Eenergy^{\ME}_g(s_0,Zu)
$$
where $\Eenergy$ represents $\EMenergy, \ECenergy$ or $\ESenergy$. Then we obtain
$$
\Eenergy^{\Lcal,N-5}_g(s,u;s_0) \lesssim (\epss+C_1\eps)^2s^{C_2(\epss+C_1\eps) + N\theta}.
$$
By choosing $\epss+C_1\eps$ sufficient small, the above inequality can be relaxed into a simpler form:
\begin{equation}\label{eq8-17-dec-2023}
\Eenergy^{\Lcal,N-5}_g(s,u;s_0) \lesssim (\epss+C_1\eps)^2s^{C_2\theta},
\end{equation}
where $C_2$ is a constant determined by $\delta, N$ and the reference metric. This estimate 
will be useful in the hyperboloidal domain when we derive energy estimate.


\section{Analysis in the hyperboloidal domain} 
\label{section-hyperb} 

\subsection{Objective}

Once the solution is constructed outside of the light cone $\{r > t-1\}$, we are in a position to analyze the interior region. The essential work was accomplished in our previous work \cite{PLF-YM-two} when the initial data coincides with the Schwarzschild metric outside of a disc on the initial slice. In the present case as well as in \cite{PLF-YM-main},  
we need to take care of the influence from the outside region.
These effects are viewed as several additional contributions:
\begin{itemize}
\item boundary terms when we apply the Hardy-Poincar\'e type inequalities,

\item boundary contribution when we apply an integration argument along characteristics associated with wave equations,

\item boundary contribution when we apply an integration argument along rays associated with Klein-Gordon equations,

\item exterior contribution when we apply pointwise estimates based on Kirchhoff's formula, and 

\item boundary terms when we derive energy estimates.
\end{itemize} 
In the same time, we also have to treat the new quadratic terms $\Gbb_{\coeffi}[\phi,g]$, $\Mbb_{\coeffi}[\phi,\psi,g]$ coupled in the Einstein equation, and the source terms coupled in the right-hand side of the Klein-Gordon equations satisfied by $\phi$ and $\psi$
 

\subsection{Bootstrap assumptions and direct estimates}

We take the following bootstrap assumptions:
\begin{subequations}\label{eqs1-16-dec-2023}
\begin{equation}\label{eq3-16-dec-2023}
\FMenergy^{\H,N-5}(s,u) + s^{-1/2} \, \FCenergy_{\coeffi}^{\H,N-5}(s, \phi) 
+ s^{-1/2} \, \FSenergy^{\H,N-5}(s, \psi) \leq(\epss+C_1\eps) \, s^{\delta},
\end{equation}
\begin{equation}\label{eq4-16-dec-2023}
\FMenergy^{\H,N-9}(s,u) +\FCenergy_{\coeffi}^{\H,N-9}(s, \phi) +\FSenergy^{\H,N-9}(s, \psi) \leq(\epss+C_1\eps) \, s^{\delta}. 
\end{equation} 
\end{subequations}
We are going to establish the following improved bounds on the same time interval:
\begin{subequations}\label{eqs2-17-dec-2023}
\begin{equation}\label{eq2a-17-dec-2023}
\FMenergy^{\H,N-5}(s,u) + s^{-1/2} \, \FCenergy_{\coeffi}^{\H,N-5}(s, \phi) 
+ s^{-1/2} \, \FSenergy^{\H,N-5}(s, \psi) \leq \frac{1}{2}(\epss+C_1\eps) \, s^{\delta},
\end{equation}
\begin{equation}\label{eq2b-17-dec-2023}
\FMenergy^{\H,N-9}(s,u) +\FCenergy_{\coeffi}^{\H,N-9}(s, \phi) +\FSenergy^{\H,N-9}(s, \psi) \leq \frac{1}{2}(\epss+C_1\eps) \, s^{\delta}. 
\end{equation}
\end{subequations}
From the definition of the energy density, the following $L^2$ estimates are direct form the above assumptions:
\begin{subequations}
\begin{equation}
\|(s/t)|\del u|_{N-5}\|_{L^2(\MH_s)} + \||\delsH u|_{N-5}\|_{L^2(\MH_s)}\lesssim (\epss+C_1\eps)s^{\delta},
\end{equation}
\begin{equation}
\|(s/t)\coeffi|\del \phi|_{p}\|_{L^2(\MH_s)} + \|\coeffi|\delsH \phi|_p\|_{L^2(\MH_s)} 
+\|\coeffi^{1/2}|\phi|_p\|_{L^2(\MH_s)} \lesssim
\begin{cases}
(\epss+C_1\eps)s^{1/2+\delta}, \,  \, p=N-5,
\\
(\epss+C_1\eps)s^{\delta},\, \, p=N-9.
\end{cases}
\end{equation}
\begin{equation}
\|(s/t)|\del \psi|_{p}\|_{L^2(\MH_s)} + \||\delsH \psi|_p\|_{L^2(\MH_s)} 
+\||\psi|_p\|_{L^2(\MH_s)} \lesssim
\begin{cases}
(\epss+C_1\eps)s^{1/2+\delta}, \, \,  p=N-5,
\\
(\epss+C_1\eps)s^{\delta},\, \, p=N-9.
\end{cases}
\end{equation}

\end{subequations}
Also, by the global Sobolev inequalities, we obtain the direct pointwise bounds:
\begin{equation}
(s/t)|\del u|_{N-7} + |\delsH u|_{N-7} \lesssim (\epss+C_1\eps)t^{-3/2} s^{\delta}
\end{equation}
\begin{equation}
(s/t)|\coeffi \del \phi|_{p-2} + |\coeffi\delsH \phi|_{p-2} + |\coeffi^{1/2}\phi|_{p-2}\lesssim
\begin{cases}
(\epss+C_1\eps)t^{-3/2}s^{1/2+\delta},\quad & p=N-5,
\\
(\epss+C_1\eps)t^{-3/2}s^{\delta},\quad & p=N-9.
\end{cases}
\end{equation}
\begin{equation}
(s/t)|\del \psi|_{p-2} + |\delsH \psi|_{p-2} + |\psi|_{p-2}\lesssim 
\begin{cases}
(\epss+C_1\eps)t^{-3/2}s^{1/2+\delta},\quad & p=N-5,
\\
(\epss+C_1\eps)t^{-3/2}s^{\delta},\quad & p=N-9.
\end{cases}
\end{equation}

The wave gauge condition will bring us additional estimates. In fact, parallel to \cite[Lemma~11.3]{PLF-YM-main} (see also \cite[Chapter 4.7]{PLF-YM-two}), we have
\begin{equation}
\aligned
|\del u^{\H00}|_{p,k}& \lesssim  |\delsH u|_{p,k} + (s/t)^2|\del u|_{p,k} + t^{-1}|u|_{p,k} + |\del h^{\star}|_{p,k} + t^{-1}|h^{\star}|_{p,k}
\\
& \quad + \sum_{p_1+p_2=p}(|u|_{p_1}|\del u|_{p_2}+|u|_{p_1}|\del h^{\star}|_{p_2} + |h^{\star}|_{p_1}|\del u|_{p_2} + |h^{\star}|_{p_1}|\del h^{\star}|_{p_2}).
\endaligned
\end{equation}
Observe that $r^{-1}|Lu|_{N-6}\lesssim |\delsH u|_{N-6} + t^{-1}|\del u|_{N-6}$ and recall the bounds \eqref{equa-31-12-20}  on the reference $h^{\star}$, we obtain
\begin{subequations}\label{eq7-16-dec-2023}
\begin{equation}\label{eq7a-16-dec-2023}
\||\del Lu^{\H00}|_{N-6}\|_{L^2(\MH_s)}\lesssim (\epss+C_1\eps)s^{2\delta}.
\end{equation}
\begin{equation}\label{eq7b-16-dec-2023}
\|(s/t)^{\delta}|\del Lu^{\H00}|_{p,k}\|_{L^2(\MH_s)}\lesssim (\epss+C_1\eps) + \FMenergy^{\H,p,k}(s,u)\lesssim (\epss+C_1\eps)s^{\delta},\quad p\leq N-6.
\end{equation}
\end{subequations}


\subsection{$L^2$ and pointwise estimates on on $\Gbb_{\coeffi}[\phi,g]_{\alpha\beta}$ and $\Mbb_{\coeffi}[\phi,\psi,g]$}

The above estimates allow us to establish $L^2$ estimates on $\Gbb_{\coeffi}$ and $\Mbb_{\coeffi}$ (which are sufficient for our purpose although it involves a non-integrable bound when $p=N-5$): 
\begin{equation}\label{eq5-16-dec-2023}
\big\|\Gbb_{\coeffi}[\phi,g]|_p\big\|_{L^2(\MH_s)} + \big\||\Mbb_{\coeffi}[\phi,\psi,g]|_p\big\|_{L^2(\MH_s)}
\lesssim
\begin{cases}
(\epss+C_1\eps)^2s^{-1+2\delta},\quad &p=N-5,
\\
(\epss+C_1\eps)^2 s^{-3/2+2\delta},\quad & p=N-9.
\end{cases}
\end{equation}
Here we observe that for the lower-order case, the above $L^2$ bounds are integrable and, via a calculation based on the direct pointwise bounds, 
\begin{equation}
|\Gbb_{\coeffi}[\phi,g]|_{N-9} + |\Mbb_{\coeffi}[\phi,\psi,g]|_{N-9}\lesssim (\epss+C_1\eps)^2 t^{-2}s^{-1/2+2\delta}.
\end{equation}
Here as in the Euclidean-merging domain, we apply Lemma~\ref{lem1-08-dec-2023}. These bounds agree with those ones in \cite[Chapter 6]{PLF-YM-two} on $QS_{\phi}$, which represents the Klein-Gordon coupling in Einstein equation. Thus the argument therein remain valid until the $L^{\infty}$ estimates on metric components based on characteristics, i.e., \cite[Chapter 9]{PLF-YM-two}, where the boundary contributions will influence the sharp pointwise bounds. This will be discussed in the next section.


\subsection{Pointwise estimate based on characteristics}\label{subsec1-17-dec-2023}

\paragraph{Integration along characteristics.}

In order to obtain sharp pointwise estimate, we need an ODE argument along characteristics of the wave equation. We need to integrate from the light cone $\Lscr$.  
In \cite{PLF-YM-two} the metric coincides with the Schwarzschild metric and this contribution can be explicitly calculated. In the present case, it will be bounded by the argument in the Euclidean-merging domain. To be more precise, let us recall the following result formulated in \cite[Chapter 3.2]{PLF-YM-two}:
\begin{proposition}
Let $u$ be a $C^2$ function defined in $\MH_{[s_0,s_1]}$ and $g^{\alpha\beta} = \gMink^{\alpha\beta} + H^{\alpha\beta}$ be a $C^1$ metric defined in $\MH_{[s_0,s_1]}$. Provided that
$$
|H^{\H 00}|\lesssim \eps_s(t-r)/t 
$$
with $\eps_s>0$ sufficiently small, the following estimate holds:
\begin{equation}\label{eq1-16-dec-2023}
t| (\del_t-\del_r ) u(t,x)|\lesssim \sup_{\MH_{s_0}\cup \Lscr_{[s_0,s]}}\big(|(\del_t-\del_r)(ru)|\big) + |u(t,x)| + \int_{s_0}^t\tau F|_{\varphi(\tau;t,x)}d\tau + \int_{s_0}^t M_s[u,H]|_{\varphi(\tau;t,x)}d\tau
\end{equation}
where we use the slice corresponding to $s= \sqrt{2t-1}$ and 
$$
|M_s[u,H]|_{p,k} \lesssim  \sum_{|Z|\leq 1}|\delsH Z u|_{p,k} + \sum_{p_1+p_2=p, k_1+k_2=k\atop |Z|\leq 1}|H|_{p_1,k_1}|\del Zu|_{p_2,k_2}
$$
and 
$\varphi_{\tau;t,x}$ is the integral curve, the initial condition $\varphi(t;t,x) = (t,x)$,  of the vector field
$$
\del_t + \frac{(t+r)^2+t^2H^{\H00}}{(t+r)^2 -t^2H^{\H00}}\del_r. 
$$
\end{proposition}
In the present case, we need a modification on the first term in the right-hand side of \eqref{eq1-16-dec-2023}. In the proof, the above estimate is applied on the ``good components'' of the metric $\usH = u^{\H}_{0a}$ or $\usH = u^{\H}_{ab}$. If we rely on the estimates established on $\us^{\N}$ in the Euclidean-merging domain (cf \cite[15.4]{PLF-YM-main}), we easily find:
\begin{equation}
|(\del_t-\del_r)(r\usH)|_{N-4,k}\lesssim (\ell^{-\delta/2}+\delta^{-2})(\epss+C_1\eps)t^{k\theta},\quad \text{on }\Lscr_{[s_0,s_1]}.
\end{equation} 
This contribution is not as good as in \cite{PLF-YM-two}, however this will not harm the argument therein. The key structure is that, when $k=0$, the above estimates is sharp in the sense that it leads us to $|\del \del^J \usH|\lesssim (\epss+C_1\eps)t^{-1}$. When $k\geq 1$, in the estimate of quadratic nonlinear terms, only the energy of lower order will be concerned. Let us explain this with the typical (and most critical) quasi-null quadratic terms
$
|\delsH u|_{p_1,k_1} |\delsH u|_{p_2,k_2}
$
which is the first one in the right-hand side of the second estimate in Lemma~\ref{lem1-02-dec-2023}. Without loss of generality, suppose that $p_1+ k_1\leq p_2+k_2$. Then
\begin{equation}
\aligned
\||\delsH u|_{p_1,k_1} |\delsH u|_{p_2,k_2}\|_{L^2(\MH_s)}
& \lesssim  
(\ell^{-\delta/2}+\delta^{-2})(\epss+C_1\eps)s^{-1}\FMenergy^{\H,p,k}(s,u) 
\\
& \quad +
(\ell^{-\delta/2}+\delta^{-2})(\epss+C_1\eps)
\sum_{k_1=1}^{[N/2]-1} s^{-1+2k_1\theta}\FMenergy^{\H,p,k-k_1}(s,u).
\endaligned
\end{equation}
Observe that, when $k=0$, the second term disappears. Regardless of other source terms, we can apply Gronwall's inequality and an induction argument on $k$ to establish refined energy bounds. 


\subsection{Pointwise estimate on $h^{\N00}$}

In \cite[Chapter 11 and Chapter 13]{PLF-YM-two}, we have applied a $L^{\infty}-L^{\infty}$ estimate based on the Kirchhoff's formula. This bound is crucial because when we do $L^2$ estimate the commutator $[Z,h^{\alpha\beta}\del_{\alpha}\del_{\beta}]\phi$, recalling Proposition~\ref{prop1-12-02-2020-interior} established in \cite[Proposition~19.3]{PLF-YM-main},  the worst term is
$$
T^{\textbf{hier}}[\phi] = h^{\N00}|\del \del\phi|_{p-1,k-1} + \sum_{k_1+p_2=p\atop k_1+k_2=k}|Lh^{\N00}|_{k_1-1}|\del\del \phi|_{p_2,k_2}
$$
Recalling that 
$$
h^{\H00} = u^{\H00} + h^{\star\H00}, 
$$
We need to give sufficient decay on $h^{\H00}, Lh^{\H00}, h^{\star\H00}$ and $L^2$ bounds on $Lu^{\H00}$. The $L^2$ estimate will be treated in the next section. We concentrate on the pointwise one.

Observe that in the right-hand side each factor contain at most $k-1$ Lorentzian boosts, we can thus perform an induction on $k$. This structure is called {\sl the hierarchy property} in \cite{PLF-YM-main} (see in detail Proposition~\ref{prop1-12-02-2020-interior}). Due to the hierarchy property, we only need 
\begin{equation}\label{eq6-16-dec-2023}
|h^{\star\H00}|_{N-5}+|h^{\N00}|_{N-12}\lesssim (\epss+C_1\eps)t^{-1+C\theta},\quad \theta\ll \delta,
\end{equation}
since when we estimate the $L^2$ norm of $T^{\textbf{hier}}$ for $\FCenergy_{\coeffi}^{\H,p,k}(s,\phi)$, there appears $\FCenergy_{\coeffi}^{\H,p,k-1}(s,\phi)$ only. The bound on $h^{\star\H00}$ is checked by 
the assumptions on the reference \eqref {equa-31-12-20}. 
 While the bounds on $h^{\H00}$ was established via the Kirchhoff's formula in \cite{PLF-YM-two,PLF-YM-main}. In the present case, there is one more contribution from the Euclidean-merging domain. Fortunately, if we check precisely the proof of \cite[Proposition 16.1]{PLF-YM-main}, the Euclidean-merging contribution is bounded in the same way, i.e., bounded by $(\epss+C_1\eps)r^{-1+(k+1)\theta}$. Combined with other contributions form the initial data and source terms in the hyperboloidal domain (which is checked in \cite[Chapter 11]{PLF-YM-two}), we conclude by \eqref{eq6-16-dec-2023}.


\subsection{$L^2$ estimates on the component $u^{\H00}$}

In order to obtain a sufficiently good $L^2$ estimate, we will rely on the following observation in which it is important to take the boundary contribution into account, namely in the condition \eqref{eq1-14-dec-2023}, below, involving the constant term $K$ introduce in \eqref{eq1-14-dec-2023}).

\begin{proposition}\label{prop1-16-dec-2023}
Consider sufficient regular functions $u$ defined in $\MH_{[s_0,s_1]}$ with $u|_{\Lcal_s}$ and satisfies the following pointwise bound on the light cone $\{r=t-1\}$:
\begin{equation}\label{eq1-14-dec-2023}
\Big|u\Big(\frac{s^2+1}{2}, x\Big)\Big|\lesssim K \, s^{-2+C_2\theta}.
\end{equation}
Then one has 
\begin{equation}\label{eq2-16-dec-2023}
\aligned
\|(s/t)^{-1/4}s^{-1}u\|_{L^2(\MH_s)}& \lesssim  K + \|(s_0/t)^{-1/4}s_0^{-1}u\|_{L^2(\MH_{s_0})} + \|t^{-1}u\|_{L^2(\MH_s)}
\\
& \quad + \int_{s_0}^s\tau^{-1}\Big(\|(\tau/t)^{3/4}\del u\|_{L^2(\MH_s)} + \|t^{-1}u\|_{L^2(\MH_s)}\Big)d\tau, 
\endaligned
\end{equation}
in which $K>0$ is some constant and $C_2\theta + \delta \leq 1/4$.
\end{proposition}

When we apply the above result on $Lu^{\H00}$, we observe that on $\{r=t-1\}$, from the argument on Euclidean-merging domain (cf \cite[Proposition~14.1]{PLF-YM-main}),
$$
|Lu^{\H00}|_{N-4}\lesssim (\epss+C_1\eps) t^{-1+\theta}
\lesssim (\epss+C_1\eps) s^{-2+2\theta}.
$$
Then, thanks to \eqref{eq2-16-dec-2023} we find
$$
\aligned
& \|(s/t)^{-1/4} s^{-1}|Lu^{\H00}|_{k-1}\|_{L^2(\MH_s)}
\\
& \lesssim  (\epss+C_1\eps) + C_0\eps + \|t^{-1}|Lu|_{k-1}\|_{L^2(\MH_s)} 
\\
& \quad + \int_{s_0}^s\tau^{-1}\Big(\|(\tau/t)^{3/4}|\del Lu^{\H00}|_{k-1}\|_{L^2(\MH_s)} + \|t^{-1}|Lu|_{k-1}\|_{L^2(\MH_s)}\Big)d\tau
\\
& \lesssim   (\epss+C_1\eps) + C_0\eps + \||\delsH u|_{k-1}+t^{-1}|\del u|_{k-1}\|_{L^2(\MH_s)} 
\\
& \quad +\int_{s_0}^s \tau^{-1}\Big(\|(\tau/t)^{3/4}|\del Lu^{\H00}|_{k-1}\|_{L^2(\MH_s)} + \|t^{-1}|Lu|_{k-1}\|_{L^2(\MH_s)}\Big)d\tau.
\endaligned
$$  
This leads us to, regarding \eqref{eq7b-16-dec-2023}, 
\begin{equation}\label{eq2-17-dec-2023}
\aligned
&
\|s^{-1}(s/t)^{-1/4}|Lu^{\H00}|_{k-1}\|_{L^2(\MH_s)}
\\
& \lesssim (\epss+C_1\eps) + \FMenergy^{\H,k-1}(s,u) +  \int_{s_0}^s\tau^{-1}
\delta^{-1}\big((\epss+C_1\eps)+\FMenergy^{\H,k}(\tau,u)\big)d\tau.
\endaligned
\end{equation}
This $L^2$ bound, though a bit weaker than the one in~\cite[Proposition~7.1]{PLF-YM-two}, is still satisfactory for our purpose. The key is that whenever the power of the weight $(s/t)$ is a negative constant, we can always ``extract'' an additional decay in $(s/t)$ when doing quadratic estimates. See an example in the proof of \cite[Lemma~7.7]{PLF-YM-two}.


\begin{proof}[Sketch of proof for Proposition~\ref{prop1-16-dec-2023}]
We adapt the proof presented in \cite[Chapter 3.6]{PLF-YM-two}, and consider the vector field
$$
W := \Big(0, -(s/t)^{-2\delta}\frac{t}{s^2}\frac{x^a}{(1+r^2)}(\chi(r/t)u)^2\Big)
$$
where $\chi$ is smooth cut-off function, $\chi(r) \equiv 0$ for $r\leq 1/3$ and $\chi(r)\equiv 1$ for $r\geq 2/3$. $\delta>0$ and will be taken as $\delta=1/4$ later in the proof.  
Then we apply Stokes formula on $W$ in the region $\MH_{[s_0,s]}$ and obtain 
$$
\aligned
&\Big\|(s/t)^{-\delta}s^{-1}\frac{r}{\sqrt{1+r^2}}\chi(r/t)u\Big\|_{L^2(\MH_s)}^2 - \Big\|(s/t)^{-\delta}s^{-1}\frac{r}{\sqrt{1+r^2}}\chi(r/t)u\Big\|_{L^2(\MH_s)}
+ \int_{\Lscr_{[s_0,s_1]}} W\cdot \vec{n} d\sigma_{\Lscr}
\\
&=\int_{s_0}^s T_1+T_2+T_3+T_4\, ds 
\endaligned
$$ 
The third term in the left-hand side is the conical boundary term
$$
\int_{\Lscr_{[s_0,s_1]}} W\cdot \vec{n} d\sigma_{\Lscr} = -\int_{\rhoH(s_0)\leq s \leq \rhoH(s)}(s/t)^{-2\delta}\frac{t}{r}\frac{r^2}{1+r^2}\frac{u^2}{s^2} dx.
$$ 
Following the estimates on $T_i$ presented in \cite{PLF-YM-two}, we have 
\begin{equation}
\aligned
&\Big\|(s/t)^{-\delta}\frac{r\chi(r/t)u}{s\sqrt{1+r^2}}\Big\|_{L^2(\MH_s)}
\frac{d}{ds}\Big\|(s/t)^{-\delta}\frac{r\chi(r/t)u}{s\sqrt{1+r^2}}\Big\|_{L^2(\MH_s)}
+ \frac{d}{ds}\int_{\Lscr_{[s_0,s]}}W\cdot \vec{n}d\sigma_{\Lscr}
\\
&\lesssim s^{-1}\Big\|(s/t)^{-\delta}\frac{r\chi(r/t)u}{s\sqrt{1+r^2}}\Big\|_{L^2(\MH_s)}
\Big(\|(s/t)^{1-\delta}\del u\|_{L^2(\MH_s)} + \|r^{-1}u\|_{L^2(\MH_s)}\Big).
\endaligned
\end{equation}
Observe that on $\Lscr, s^2\simeq t, r\simeq t$ and thanks to \eqref{eq1-14-dec-2023} 
$$
0\leq -\frac{d}{ds}\int_{\Lscr_{[s_0,s]}}W\cdot \vec{n}d\sigma_{\Lscr} = \int_{r = \rhoH(s)} (s/t)^{-2\delta}\frac{t}{r}\frac{r^2}{1+r^2}\frac{u^2}{s^2}\,  \rhoH(s)^2d\sigma_{\mathbb{S}^2} \lesssim K^2s^{-3/2+2C_2\theta} 
$$
for some constant $K>0$. Here we fix $\delta = 1/4$.
This together with the ODE argument presented in Lemma~\ref{lemma-83} (below)  
leads us to (with $g = C_1 s^{-3/2 + 2C_2\theta}$ which is integrable)
$$
\aligned
\Big\|(s/t)^{-\delta}\frac{r\chi(r/t)u}{s\sqrt{1+r^2}}\Big\|_{L^2(\MH_s)}
& \lesssim  
K + \Big\|(s_0/t)^{-\delta}\frac{r\chi(r/t)u}{s_0\sqrt{1+r^2}}\Big\|_{L^2(\MH_{s_0})}
\\
& \quad + \int_{s_0}^s\tau^{-1}\Big(\|(\tau/t)^{1-\delta}\del u\|_{L^2(\MH_s)} + \|t^{-1}u\|_{L^2(\MH_s)}\Big)d\tau.
\endaligned
$$
Finally, we note that
$$
\|s^{-1}u\|_{L^2(\MH_s)}\lesssim \Big\|(s/t)^{-\delta}\frac{r\chi(r/t)u}{s\sqrt{1+r^2}}\Big\|_{L^2(\MH_s)} + \|t^{-1} u \|_{L^2(\MH_s)} 
$$
and arrive at the desired result.
\end{proof}

\begin{lemma}
\label{lemma-83}
Let $u$ be a non-negative $C^1$ function defined on $[s_0,s_1]$, while $f,g$ are continuous functions defined on $[s_0,s_1]$. Provided
\begin{equation}
u(s) u'(s) \leq u(s)f(s) + C_1g(s), 
\end{equation} 
one has 
$$
u(s)\leq C_1 + u(s_0) + \int_{s_0}^s f(\tau) + g(\tau)\,d\tau.
$$
\end{lemma}


\subsection{Uniform estimates for Klein-Gordon components}

\paragraph{Aim.}

In order to establish the sharp decay estimates in \cite[Chapter 11, 13]{PLF-YM-two}, there is still one ingredient to be discussed, namely the $L^{\infty}$ estimates on Klein-Gordon components via integration along segments. In the present case, there are two new effects. First of all, a boundary contribution arises on the light cone $\Lscr_{[s_0,s_1]}$ and, secondly, 
contributions due to the mass coefficient $\coeffi$ should be carefully handled in order to 
arrive at suitably uniform bounds.
Our strategy is to review
 the main proof given in \cite{PLF-YM-two} and take care of the above tow points. As we will see, these new effects will not destroy the desired estimates.


\paragraph{An ODE argument.}

\begin{lemma}\label{lem1-03-dec-2023}
Suppose that $\phi$ is a $C^2$ function defined on $[s_0,s_1]$ and satisfying the  ODE
\begin{equation}\label{eq7-03-dec-2023}
3\coeffi(1+h(s)) \phi''(s) + M\phi(s) = f(s), 
\end{equation}
where $h\in C^1([s_0,s_1]$ and $|h(s)|\leq 1/3$ while $f$ is continuous. Here, $\coeffi, M>0$ are constants. Suppose furthermore that
\begin{equation}
\int_{s_0}^{\infty}|h'(s)|ds\lesssim 1.
\end{equation}
Then the following estimate holds for $s\in(s_0,s_1)$:
\begin{equation}
|\coeffi\phi'(s)| + |M^{1/2}\coeffi^{1/2}\phi(s)| \lesssim (M\coeffi)^{1/2}|\phi(s_0)| + \coeffi|\phi'(s_0)| + \int_{s_0}^s|f(\tau)|d\tau.
\end{equation}
\end{lemma}

\begin{proof}
Multiplying \eqref{eq7-03-dec-2023} by $\coeffi\phi'$, we obtain
\begin{equation}\label{eq8-03-dec-2023}
\frac{1}{2}\frac{d}{ds}\Big(3(1+h(s))|\coeffi \phi'|^2 + \coeffi M |\phi|^2\big) - \frac{3}{2}h'(s)|\coeffi\phi'|^2 = \coeffi\phi(s)f(s).
\end{equation}
Let us introduce 
$$
\Phi (s) := \big(3(1+h(s))|\coeffi \phi'|^2 + \coeffi M|\phi|^2)^{1/2}.
$$
Then \eqref{eq8-03-dec-2023} leads us to 
$$
\Phi(s)\Phi'(s) \lesssim (\coeffi|h'(s)\phi'(s)|+ |f(s)|)|\coeffi \phi'|  \lesssim (|h'(s)|\Phi(s)+ |f(s)|)\Phi(s), 
$$
that is,  
$
\Phi'(s)\lesssim |h'(s)|\Phi(s) + |f(s)|.
$
Then we apply Gronwall's inequality on the interval $[s_0,s_1]$, and obtain the desired result.
\end{proof}


\paragraph{Pointwise estimates for Klein-Gordon equations.}

The method below was introduced in Klainerman~\cite{Klainerman85}  
and the derivation below was proposed in LeFloch and Ma~\cite{PLF-YM-one}. 
Here, we present yet another version of the argument which 
takes into account the contribution from the boundary (namely the light cone) and the influence of the small coefficients. 

We focus here on the hyperboloidal domain and rely on the decomposition 
\be
g^{\alpha\beta}\del_{\alpha}\del_\beta  = s^{-3/2}(t/s)^2\gu^{00}\Lcal^2(s^{3/2}\phi) - \RR_g[\phi],
\ee
where $g^{\alpha\beta} = \Minsk^{\alpha\beta} + H^{\alpha\beta}$ and 
the following field (which is nothing but the unit normal to the hyperboloids for the Minkowski metric) 
\be
\Lcal := (s/t)\del_t + (x^a/s)\delus_a = (t/s)\del_t + (x^a/s)\del_a,
\ee
while 
\be
\aligned
\RR_g[\phi] & =  -(3/4)s^{-2}\phi - 3s^{-1} (x^a/s)\delus_a\phi - (x^ax^b/s^2)\delus_a\delus_b\phi - \sum_a\delus_a\delus_a\phi
\\
& \quad -(t/s)^2\Hu^{00}\Big((3/4)s^{-2}\phi + \big(3+(r/t)^2\big)t^{-1}\del_t\phi + 3s^{-1}(x^a/s)\delus_a\phi\Big)
\\
& \quad -(t/s)^2\Hu^{00}\Big((2x^a/t)\del_t\delus_a\phi + (x^ax^b/s^2)\delus_a\delus_b\phi\Big)
\\
& \quad +\Hu^{a0}\delus_a\del_t\phi + \Hu^{0a}\del_t\delus_a\phi + \Hu^{ab}\delus_a\delus_b\phi 
+ H^{\alpha\beta}\del_{\alpha}\big(\PsiH_\beta ^{\beta'}\big)\delu_{\beta'}\phi.
\endaligned
\ee
For a Klein-Gordon equation 
$
3\coeffi\BoxChapeaud_\gconf  \phi - \Ufrak'(0)\phi = f,
$
thanks to $\gu^{00} = -(s/t)^2 + \Hu^{00}$, under the assumption
$
(s/t)^2|\Hu^{00}|\leq 1/3,
$
the above identity leads us to
\be
3\coeffi(1 - (t/s)^2\H^{\Hcal00}) \Lcal^2(s^{3/2}\phi) + \Ufrak'(0) (s^{3/2}\phi) = -s^{3/2}( 3\coeffi \RR_g + f).
\ee
For any function $\phi$ defined in $\MH_{[s_0,s_1]}$ and at each point $(t,x)\in \MH_{[s_0,s_1]}$, we use the notation $\Phi_{t,x}(\lambda) := \lambda^{3/2}\phi(\lambda t/s,\lambda x/s)$. Since $\Lcal(s^{3/2}\phi)|_{(\lambda t/s,\lambda x/s)} = \Phi_{t,x}'(\lambda)$, we find 
\bel{eq2-22-11-2022}
3\coeffi\big(1-\Hb_{t,x}\big)\Phi_{t,x}''(\lambda) + \Ufrak'(0)\Phi_{t,x}(\lambda) = -\lambda^{3/2}(f+  3\coeffi \RR_g[\phi])|_{(\lambda t/s,\lambda x/s)},
\ee
where  $\Hb_{t,x} = (t/s)^2\Hu^{00}|_{(\lambda t/s,\lambda x/s) }$. 
 
By Lemma~\ref{lem1-03-dec-2023}, we then arrive at the following result. 

\begin{proposition}[Sharp decay of Klein-Gordon solutions in the hyperboloidal domain]
\label{prop1-23-11-2022-M}
Suppose that for all $(t,x)\in \MH_{[s_0,s_1]}$ and for all $\lambda_0\leq \lambda\leq s_1$, one has 
\bel{eq4-23-11-2022-M}
|\Hb_{t,x}|\leq 1/3,\qquad \int_{\lambda_0}^{s_1}|\Hb_{t,x}'(\lambda)|d\lambda\lesssim 1. 
\ee
Then 
for any $\expeta\in \RR$, 
any solution $\phi$ to the Klein-Gordon equation $3\coeffi\Boxt_g \phi - \Ufrak'(0)\phi = f$ satisfies 
\bel{eq5-23-11-2022-M}
\aligned
& (s/t)^{\sigma}s^{3/2} \, \Big(\coeffi^{1/2}|\phi(t,x)|+ (s/t)\, \coeffi|\del\phi(t,x)| \Big)
\\
& \lesssim (s/t)^{\sigma}s^{1/2}\coeffi|\phi|_1(t,x) 
+ \sup_{\MH_{s_0}\cup \Lscr_{[s_0,s]}}(s/t)^{\sigma}\Big( t^{1/4}(t^{1/2}\coeffi^{1/2}|\phi| + \coeffi|\del\phi| + t \, \coeffi|\delsH\phi|)\Big)
\\
& \quad + (s/t)^{\sigma}\int_{\lambda_0}^s\lambda^{3/2}
\big( |f| + 3\coeffi |\RR_g[\phi] | \big)\big|_{(\lambda t/s,\lambda x/s)} d\lambda,
\endaligned
\ee
in which 
\be
\lambda_0 = 
\begin{cases}
s_0,\quad &0\leq r/t\leq \frac{s_0^2-1}{s_0^2+1},
\\
\sqrt{\frac{t+r}{t-r}},&  \frac{s_0^2-1}{s_0^2+1}\leq r/t<1
\end{cases}
\ee
and 
\bel{eq11-23-11-2022-M}
\aligned
|\RR_g[\phi]|_{p,k}
& \lesssim s^{-2}|\phi|_{p+2} + (t/s)^2\sum_{p_1+p_2=p}|\Hu^{00}|_{p_1} \big(
s^{-2}|\phi|_{p_2+2} + t^{-1}|\del\phi|_{p_2+1} \big)
\\
& \quad + \sum_{p_1+p_2=p}|H|_{p_1} \big(
t^{-1}|\del\phi|_{p_2+1} + t^{-2}|\phi|_{p_2+2} \big).
\endaligned
\ee
\end{proposition}

\begin{proof} Observe that the integral curve of $\Lcal$ is $\gamma_{t,x} = \{(\lambda t/s,\lambda x/s)\}$. When
$0\leq r/t\leq \frac{s_0^2-1}{s_0^2+1}$, the segment
$
\big\{( \lambda t/s,\lambda x/s)|s_0 \leq \lambda\leq s \big \}
$
is contained in $\MH_{[s_0,s]}$ and $\gamma_{t,x}$ meets $\del \MH_{[s_0,s]}$ at $(s_0/s)(t,x)\in\MH_{s_0}$. When $\frac{s_0^2-1}{s_0^2+1}\leq r/t<1$, the segment
$$
\Big\{(\lambda t/s,\lambda x/s)| \sqrt{\frac{t+r}{t-r}}\leq \lambda\leq s \Big\}
$$
is contained in $\MH_{[s_0,s_1]}$ and meets $\del\MH_{[s_0,s_1]}$ at the point $\frac{1}{t-r}(t,r)\in \MH_{[s_0,s_1]}\cap \MME_{[s_0,s_1]}$. Here we emphasize  that 
$
\sqrt{\frac{t+r}{t-r}}\simeq t/s.
$
Observe that $(s/t)$ is constant along a given $\gamma_{t,x}$. We denote by $\Phi_{t,x,\sigma} := (s/t)^{\sigma}\Phi_{t,x}$ and obtain
$$
3\coeffi\big(1-\Hb_{t,x}\big)\Phi_{t,x,\sigma}''(\lambda) + \Ufrak'(0)\Phi_{t,x,\sigma}(\lambda) = -\lambda^{3/2}(s/t)^{\sigma}(f+  3\coeffi \RR_g[\phi])|_{(\lambda t/s,\lambda x/s)},
$$
Then by Lemma~\ref{lem1-03-dec-2023}, we have 
$$
\aligned
\big|\coeffi\Phi_{t,x,\sigma}'(s)\big| + \big|(\Ufrak'(0)\coeffi)^{1/2}\Phi_{t,x,\sigma}(s)\big| 
& \lesssim  \big|\coeffi\Phi_{t,x,\sigma}'(s_0)\big| + \big|(\Ufrak'(0)\coeffi)^{1/2}\Phi_{t,x,\sigma}(s_0)\big|
\\
& \quad + (s/t)^{\sigma}\int_{s_0}^s\lambda^{3/2}\big(|f| + 3\coeffi |\RR_g[\phi]|\big)_{\lambda t/s,\lambda x/s}d\lambda.
\endaligned
$$
On the other hand, when $\lambda_0\geq s_0\geq 2$ and $\coeffi \leq 1$ we have 
$$
\lambda^{1/2} \coeffi |\phi(\lambda t/s,\lambda x/s)| 
+ 
\coeffi|\Phi_{t,x}'(\lambda)| + \coeffi^{1/2}|\Phi_{t,x}(\lambda)|
\gtrsim
\lambda^{3/2}\big(\coeffi^{1/2}|\phi(\lambda t/s,\lambda x/s)| + \coeffi|\Lcal\phi(\lambda t/s,\lambda x/s)|\big), 
$$
in which $\Lcal = (t/s)\del_t + (x^a/s)\del_a = (t/s)\del_t + (x^a/s)\delsH_a$. We also observe that
$$
\Lcal = (s/t)\del_t + s^{-1}(x^a/t)L_a,\quad \del_a = t^{-1}L_a - (x^a/t)\del_t,
$$
thus
$$
\aligned
& 
(s/t)^{\expeta}\lambda^{3/2}\big(\coeffi^{1/2}|\phi(\lambda t/s,\lambda x/s)| + (s/t) \,\coeffi |\del \phi(\lambda t/s,\lambda x/s)|\big) 
\\
& \lesssim  (s/t)^{\expeta}\lambda^{1/2}\coeffi |\phi|_1(\lambda t/s,\lambda x/s)
+ \coeffi \, |\Phi_{t,x,\expeta}'(\lambda)| + \coeffi^{1/2} \, |\Phi_{t,x,\expeta}(\lambda)|. 
\endaligned
$$
We also observe that 
$$
\coeffi \, |\Phi_{t,x,\expeta}'(\lambda)| + \coeffi^{1/2} |\Phi_{t,x,\expeta}(\lambda)|
\lesssim 
(s/t)^{\expeta}\lambda^{3/2} \, \big(
\coeffi^{1/2}|\phi|_1+(s/t) \, \coeffi|\del\phi| + (t/s)\coeffi|\delsH\phi|\big) \Big|_{(\lambda t/s,\lambda x/s)}.
$$
This gives the desired result.
\end{proof}

In our case, Proposition~\ref{prop1-23-11-2022-M} will be applied on both $\phi$ and $\psi$ and when applied on $\psi$, we take $\coeffi = 1/3$ in \eqref{eq5-23-11-2022-M}. The contribution from the Euclidean-merging domain appears in the second term, 
 for which we need to take into account
  the sup-norm of $(s/t)^{\sigma}\Big( t^{1/4}(t^{1/2}\coeffi^{1/2}|\phi| + \coeffi|\del\phi|+ t \,\coeffi |\delsH\phi|)\Big)$. In fact form \eqref{eqs2-08-dec-2023} we obtain
\begin{equation}
\aligned
& t^{1/4}(t^{1/2}\coeffi^{1/2}|\phi|_{N-9} + \coeffi|\del\phi|_{N-9}+ t \,\coeffi |\delsH\phi|_{N-9})\lesssim (\epss+C_1\eps)t^{-5/4+\delta}\\
& \lesssim (\epss+C_1\eps)(s/t)^{-5/2+2\delta}
\quad \text{on } \Lscr_{[s_0,s_1]},
\endaligned
\end{equation}
where we used $s\simeq t^{1/2}, r\simeq t$ and $\crochet\simeq 1$ on $\Lscr_{[s_0,s_1]}$. This uniform bound will not harm the sharp decay estimate on $\phi$. The same argument works also for $\psi$.

On the other hand, if we apply Proposition~\ref{prop1-23-11-2022-M}, we need to bound the right-hand side of the Klein-Gordon equations satisfied by $\phi$ and $\psi$, and the corresponding commutators
$$
3\coeffi[Z, h^{\mu\nu}\del_{\mu}\del_{\nu}]\phi,\quad [Z, h^{\mu\nu}\del_{\mu}\del_{\nu}]\psi
$$ 
which will be regarded as $f$ presented in \eqref{eq5-23-11-2022-M}. Then we will rely on the following key structure. 
\\
$\bullet$ First of all, if we count the power of $\coeffi$, we will find that in the right-hand side of the Klein-Gordon equation of $\phi$ and $\psi$, each factor $\del\phi$ is multiplied with $\coeffi$ and each factor $\phi$ is multiplied by $\coeffi^{1/2}$. This is true for all terms presented in the right-hand side of \eqref{eq5-23-11-2022-M}. Here we observe especially that $\RR_g[\phi]$ is in fact {\sl linear} with respect to $\phi$.
\\
$\bullet$ Second, if we regard the left-hand side of \eqref{eq4-23-11-2022-M}, the conclusion obtained form the above proposition also obeys the correct scaling on $\coeffi$: it is $\coeffi^{1/2}|\phi|_{p,k}$ and $\coeffi|\del \phi|_{p,k}$ that we estimate. Thus consistency allow us to perform exactly the same argument presented in \cite{PLF-YM-two} to obtain sharp decay for $u,\phi$ and $\psi$ in the following form:
\begin{equation}\label{eq1-17-dec-2023}
\aligned
|u^{\H00}|_{N-9,k}+|h^{\H00}|_{N-9,k}& \lesssim  (\epss+C_1\eps)t^{-1+(1+k)\theta},
\\
\coeffi^{1/2}|\phi|_{N-9} + (s/t) \coeffi|\del \phi|_{N-9}& \lesssim  \theta^{-1}(\epss+C_1\eps)(s/t)^{2-3\delta}s^{-3/2+k\theta},
\\
|\psi|_{N-9} + (s/t)|\psi|_{N-9}\lesssim &\theta^{-1}(\epss+C_1\eps)(s/t)^{2-3\delta}s^{-3/2+k\theta}.
\endaligned
\end{equation}
This  is still a bit weaker than those established in \cite[Chapter 11,13]{PLF-YM-two}. This is due to the loss $t^{\theta}$ caused by $h^{\star}$. However, these are still sufficient to close the bootstrap argument.


\subsection{Improved energy estimate and conclusion}

Equipped with the sharp decay estimates \eqref{eq1-17-dec-2023}, we are ready to establish the improve the energy bounds \eqref{eqs2-17-dec-2023}. We focus on the most critical terms and new terms, which are:
\\
$\bullet$ In wave equation: 
$$
\aligned
&|\delsH u|_{p_1,k_1}|\delsH u|_{p_2,k_2},\quad &&\text{quasi-null terms},
\\
&|h^{\H00}||\del \del u|_{p-1,k-1} + \sum_{k_1+p_2=p\atop k_1+k_2=k}|Lh^{\H00}|_{k_1-1}|\del\del u|_{p_2,k_2},&&\text{critical commutator},
\\
&\Gbb_{\coeffi}[\phi,g],&&\text{contribution of effective curvature field},
\\
&\Mbb_{\coeffi}[\phi,\psi,g], &&\text{contribution of scalar field}.
\endaligned
$$  
\\
$\bullet$ In Klein-Gordon equation for $\phi$:
$$
\aligned
&3\coeffi |h^{\H00}||\del\del \phi|_{p-1,k-1} + 3\coeffi \sum_{k_1+p_2=p\atop k_1+k_2=k}|Lh^{\H00}|_{k_1-1}|\del\del \phi|_{p_2,k_2},\quad &&\text{critical commutator},
\\
&-8\pi e^{-\coeffi\phi} \gconf(\del\psi,\del\psi) - 32\pi e^{-2\coeffi\phi}V(\psi), &&\text{new source terms}
\endaligned
$$
\\
$\bullet$ In Klein-Gordon equation for $\psi$:
$$
\aligned
&|h^{\H00}||\del \del u|_{p-1,k-1} + \sum_{k_1+p_2=p\atop k_1+k_2=k}|Lh^{\H00}|_{k_1-1}|\del\del u|_{p_2,k_2},&&\text{critical commutator},
\\
& \coeffi g^{\mu\nu}\del_{\mu}\phi\del_{\nu}\psi, &&\text{new source term}.
\endaligned
$$

Among the above terms, the quasi-null terms coupled in wave equation has been explained in Section~\ref{subsec1-17-dec-2023}. 

The new source terms are relatively easy: they are all essentially KG$\times$KG quadratic terms. We only need to substitute the sharp decay estimates \eqref{eq1-17-dec-2023} into the estimates of Lemma~\ref{lem1-08-dec-2023}. We obtain
\begin{equation}\label{eq3-17-dec-2023}
\||T|_{p,k}\|_{L^2(\MH_s)}
\lesssim 
\begin{cases}
\theta^{-1}(\epss+C_1\eps)^2s^{-3/2+k\theta},\quad \hskip2.cm p\leq N-9
\\
\theta^{-1}(\epss+C_1\eps)\sum_{0\leq k_1\leq k}s^{-3/2+k_1\theta}\big(\FCenergy^{\H,p,k-k_1}(s,\phi) + \FSenergy^{\H,p,k-k_1}(s,\psi)\big),
\\
 \quad \hskip5.cm N-8\leq p\leq N-5. 
\end{cases}
\end{equation}

The two critical commutators acting on Klein-Gordon components can be treated similarly. We only write that of $\phi$ in detail.
$$
\aligned
&
\|\coeffi |h^{\H00}||\del\del \phi|_{p-1,k-1}\|_{L^2(\MH_s)}\lesssim (\epss+C_1\eps) s^{-1+2\theta} \|\coeffi (s/t)|\del \phi|_{p,k-1}\|_{L^2(\MH_s)}
\\
& \lesssim (\epss+C_1\eps)s^{-1+2\theta}\FCenergy_{\coeffi}^{\H,p,k-1}(s,\phi),
\endaligned
$$
and 
$$
\aligned
&\|\coeffi|Lh^{\H00}|_{k_1-1}|\del\del \phi|_{p_2,k_2}\|_{L^2(\MH_s)}
\\
& \lesssim  
\begin{cases}
\|\coeffi |Lh^{\star}|_{k_1-1}|\del \phi|_{p,k-1}\|_{L^2(\MH_s)} + \|\coeffi|Lu^{\H00}|_{k_1-1}|\del \phi|_{p_2,k_2}\|_{L^2(\MH_s)},\quad &k_1-1\geq N-9,
\\
(\epss+C_1\eps)s^{-1+2(k_1+1)\theta}\|\coeffi(s/t)|\del\phi|_{p,k-k_1}\|_{L^2(\MH_s)} &k_1-1\leq N-10,
\end{cases}
\\
& \lesssim 
\begin{cases}
{
\aligned
(\epss+C_1\eps)&s^{-1+2\theta}\|\coeffi (s/t)|\del\phi|_{p,k-1}\|_{L^2(\MH_s)} 
\\
& \quad + \theta^{-1}(\epss+C_1\eps)s^{-3/2 + k_2\theta}\|(s/t)^{1-3\delta}|Lu^{\H00}|_{k_1-1}\|_{L^2(\MH_s)}
\endaligned
}
\quad &k_1-1\geq N-9
\\
(\epss+C_1\eps)s^{-1+2(k+1)\theta}\FCenergy^{\H,p,k-k_1}(s,\phi),\quad &k_1-1\leq N-10
\end{cases}
\\
& \lesssim  
(\epss+C_1\eps)\sum_{1\leq k_1\leq k}s^{-1+2(k_1+1)\theta}\FCenergy^{\H,p,k-k_1}(s,\phi) 
\\
& \quad + \theta^{-1} (\epss+C_1\eps)\sum_{N-8\leq k_1\leq k}s^{-1/2+k_1\theta}
\Big(\FMenergy^{\H,p,k-k_1}(s,u) + \int_{s_0}^s\tau^{-1}\FCenergy^{\H,p,k-k_1}(\tau,u)d\tau\Big)
\\
& \quad + \theta^{-1}(\epss+C_1\eps)^2\ln(s/s_0)\sum_{8\leq k_1\leq k}s^{-1/2+k_1\theta}.
\endaligned
$$
Here we observe that the last two terms do not exist when $k\leq 9$. 

The critical commutator acting on $u$ is bounded exactly as in \cite[Lemma~10.2]{PLF-YM-two}. The key is the following estimate on the Hessian of a wave solution: 
\begin{equation}\label{eq4-17-dec-2023}
\la t-r\ra|\del\del u|_{p-1,k-1}\lesssim |\del u|_{p,k} + t|\Box u|_{p,k}
\end{equation}
In the curved background, there will be several quadratic correction terms. The extra $\la t-r\ra \simeq t/s^2$ decay on the Hessian form combined with the first estimate in \eqref{eq1-17-dec-2023} and the $L^2$ estimate \eqref{eq2-17-dec-2023} (particularly, the weight $(s/t)^{-1/4}$) brings us an integrable decay on $s$.

We are finally in a position to apply the argument in \cite[Proposition~12.1 and Lemma~12.4]{PLF-YM-two} and \cite[Proposition~10.1 and Proposition~14.1]{PLF-YM-two}. Also recall that by \eqref{eq8-17-dec-2023}, the boundary terms appearing in energy estimate only contribute a slow increasing bound, viz., $s^{C_1\theta}$, which is much smaller that $s^{\delta}$. Thus it will not prevent us from establish a $s^{\delta}$-increasing energy bound.


\paragraph*{Acknowledgments}

The first author (PLF) was supported by the research project {\sl Einstein-PPF} entitled {\it ``Einstein constraints: past, present, and future''} funded by the Agence Nationale de la Recherche (ANR), as well as by the MSCA Staff Exchange grant Project 101131233: {\sl Einstein-Waves} entitled {\it ``Einstein gravity and nonlinear waves: physical models, numerical simulations, and data analysis'',} funded by the European Research Council (ERC).


\small 

\appendix

\section{Conformal identities} 
\label{appendix-A}

\paragraph{Conformal transformation.}

We recall several calculations from~\cite{PLF-YM-memoir}. 
Let $(\Mcal,\giveng)$ be a pseudo-Riemannian manifold. Given any (regular) function $\phi$ defined on $\Mcal$, we can introduce the conformal metric  
\bel{equa-conformal-trans} 
g_{\alpha\beta} := e^{\coeffi \phi}\giveng_{\alpha\beta},\quad g^{\alpha\beta} 
= e^{- \coeffi \phi}\giveng^{\alpha\beta}.
\ee
The Christoffel symbols read 
\be
\Gamma_{\alpha\beta}^{\gamma} 
= \Gammaj_{\alpha\beta}^{\gamma} 
+ {\coeffi \over 2} \Big(
\delta_{\alpha}^{\gamma}\del_\beta \phi + \delta_\beta ^{\gamma}\del_{\alpha}\phi - g_{\alpha\beta}g^{\gamma\delta}\del_{\delta}\phi.
\Big),
\ee
This leads us to
\begin{equation}\label{eq1-28-oct-2023}
\Box_g w = e^{-\coeffi\phi}\Box_{\giveng} w + \coeffi e^{-\coeffi\phi}\giveng(\nabla\phi,\nabla w)
\end{equation}
and the Ricci curvature satisfies 
\begin{equation}\label{eq4-12-nov-2023}
R_{\alpha\beta} 
= \widetilde R_{\alpha\beta} 
- \coeffi \, \givennabla_{\alpha}\givennabla_\beta \phi + {\coeffi^2 \over 2} \,  \givennabla_{\alpha}\phi\givennabla_\beta \phi
 - \Big( 
 {\coeffi \over 2} \Box_{\giveng}\phi + {\coeffi^2 \over 2} \, \giveng(\givennabla\phi,\givennabla\phi)
 \Big) \, \giveng_{\alpha\beta}.
\end{equation}


\paragraph{Contracted Bianchi identity.}

\begin{lemma} 
\label{appendix-propositionA1}
In view of \eqref{equa-conformal-trans} one has 
\begin{equation}\label{eq2-26-mars-2023}
\nabla^{\alpha}\givenCurv_{\alpha\beta} 
=  \coeffi \, \nabla^{\mu}\phi\,\givenCurv_{\mu\beta}
- {\coeffi \over 2} \Tr_g(\givenCurv)\nabla_\beta \phi.
\end{equation}
\end{lemma}
 
\begin{proof} For any symmetric two-tensor $T = T_{\alpha\beta}dx^{\alpha}\otimes dx^{\beta}$, we have 
$$
\nabla_{\gamma}T_{\alpha\beta} = \del_{\gamma}T_{\alpha\beta} 
- \Gamma_{\gamma\alpha}^{\delta}T_{\delta\beta} - \Gamma_{\gamma\beta}^{\delta}T_{\alpha\delta} 
$$
therefore, in view of \eqref{equa-conformal-trans}, 
$$
\aligned
\nabla_{\gamma}T_{\alpha\beta} 
= \givennabla_{\gamma} T_{\alpha\beta} 
+ {\coeffi \over 2} \Big(
- T_{\alpha\gamma}\del_\beta \phi 
- T_{\gamma\beta}\del_{\alpha}\phi - 2T_{\alpha\beta}\del_{\gamma}\phi
 + \big(g_{\alpha\gamma}T_{\delta\beta} + g_{\gamma\beta}T_{\delta\alpha}\big)g^{\delta\mu}\del_{\mu}\phi
\Big).
\endaligned
$$
We then obtain
$$
\aligned
\nabla^{\alpha}T_{\alpha\beta} 
= e^{- \coeffi \phi}\givennabla^{\alpha}T_{\alpha\beta} 
+ \coeffi \, g^{\mu\nu}T_{\mu\beta}\del_{\nu}\phi 
- {\coeffi \over 2} \, \Tr_g(T)\del_\beta \phi.
\endaligned
$$
Apply the above identity on $\givenCurv_{\alpha\beta}$, and the relation $\givennabla^{\alpha}\givenCurv_{\alpha\beta} = 0$.
\end{proof}


\paragraph{Equivalence in \eqref{eq2-10-mars-2023}.}

We claim that 
\bel{eq2-20-mars-2023}
-  \Ufrak'(\phi) = e^{-2 \coeffi \phi}\ffrak(\phi) + \coeffi \, \Ufrak(\phi).
\ee
Indeed, in view of $U(\givenR) = (f(\givenR) - \givenR f'(\givenR) / (f''(0) f'(\givenR)^2)$ the right-hand side reads 
$$
e^{- 2 \coeffi  \phi}\ffrak(\phi) + \coeffi \, \Ufrak(\phi) = {2 f(\givenR) - R f'(\givenR) \over f'(\givenR)^2} 
$$
while the left-hand side is computed by differentiating $\Ufrak = U\circ \hfrak$, indeed 
$$
 \Ufrak'(\phi) 
={\coeffi f'(\givenR) \over f''(\givenR)} {d \over d\givenR} \Big(\coeffi^{-1} (f(\givenR) - \givenR f'(\givenR)) / f'(\givenR)^2 \Big)
= - 
{2 f(R) - R f'(R) \over f'(R)^2}, 
$$
which establishes our claim.  


\section{Commutators, hierarchy, and functional inequalities} 
\label{appendix-B}

\paragraph{Commutator properties.}

The proof of the following statement can be found in~\cite[Part~1]{PLF-YM-main}. To deal with differential operators 
$\del^I L^J \Omega^K$ we use the notation $\ord(Z) = |I|+|J|+|K|$ (the order) and $\rank(Z) = |J|+|K|$ (the rank).  Given two integers $k \leq p$, it is convenient to introduce the notation 
\bel{equa-notation-pk}
|u|_{p,k} :=  \max_{\ord(Z) \leq p, \rank(Z) \leq k} |Z u|, 
\qquad 
|u|_{p} :=  \max_{\ord(Z) \leq p} |Z u|. 
\ee

\begin{lemma}[Estimates for linear commutators]
\label{prop-newpropo} 
For any admissible field $Z$ satisfying $\ord(Z) \leq p$ and $\rank(Z) \leq k$ one has 
\begin{subequations}
\be \label{eq2-31-01-2020}
|[Z, \del] u | \lesssim |\del u|_{p-1,k-1},
\ee
\be \label{eq1-31-01-2020}
| [Z, \del \del ] u | 
\lesssim |\del\del u|_{p-1,k-1} \lesssim |\del u|_{p,k-1}.
\ee
\end{subequations}
\end{lemma} 

It is convenient to denote by $\LOmega \in\{ L_a, \Omega_{ab}\}$ the collection of boosts and spatial rotations. 

\begin{proposition}[Hierarchy structure for quasi-linear commutators. Euclidean-merging domain]
\label{lm 2 dmpo-cmm-H} 
Let $Z$ be an admissible operator with ${\ord(Z) = p}$ and $\rank(Z) = k$ and let $H,u$ be functions defined in the  Euclidean-merging domain $\MME_s$.  Then, one has 
\be \label{eq8-14-03-2021}
|[Z,H]u|\lesssim 
\sum_{k_1+p_2=p\atop k_1+k_2=k } | \LOmega H|_{k_1-1} |u|_{p_2,k_2}
+ \sum_{p_1+p_2=p\atop k_1+k_2=k} |\del H|_{p_1-1,k_1} |u|_{p_2,k_2}, 
\ee 
\be \label{eq7-14-03-2021}
\aligned
 |[Z,H\del_{\alpha}\del_\beta ]u|
& \lesssim  \,  |H| \, |\del\del u|_{p-1,k-1} 
+ 
\sum_{k_1+p_2=p\atop k_1+k_2=k }
|\LOmega H|_{k_1-1} |\del\del u|_{p_2,k_2}
  +
\sum_{p_1+p_2=p\atop k_1+k_2=k}   |\del H|_{p_1-1,k_1} |\del\del u|_{p_2,k_2}. 
\endaligned
\ee 
\end{proposition}

\begin{proposition}[Hierarchy property for quasi-linear commutators. Hyperboloidal domain] 
\label{prop1-12-02-2020-interior} 
For any function $u$ defined in $\MH_{[s_0,s_1]}$ and for any admissible operator $Z$ with $\ord(Z) = p$ and $\rank(Z) = k$ one has 
\begin{subequations}
\bel{equa-hsgd57} 
\aligned
\, 
& |[Z,H^{\alpha\beta} \del_\alpha \del_\beta ]u|
\lesssim  T^\textbf{hier}  + T^\easy +  T^{\textbf{super}}, 
\endaligned
\ee
\be
\aligned
T^\textbf{hier} & 
|\Hu^{00} | \, | \del\del u|_{p-1,k-1} +\sum_{k_1+p_2=p\atop k_1+k_2=k} |L \Hu^{00} |_{k_1-1,k_1-1} |\del\del u|_{p_2,k_2}, 
\\
T^\easy
& :=
\sum_{p_1+p_2=p\atop k_1+k_2=k} |\del \Hu^{00} |_{p_1-1,k_1} |\del\del u|_{p_2,k_2}
+ t^{-1} |H| \, | \del u|_{p}, 
\\
T^{\textbf{super}} &
:= \sum_{p_1+p_2=p} |\delsH H|_{p_1-1} |\del u|_{p_2+1}
+ t^{-1} \!\!\!\!\sum_{p_1+p_2=p} \!\!\!\! |\del H|_{p_1-1} |\del u|_{p_2+1}.
\endaligned
\ee
\ese
\end{proposition}


\paragraph{Sobolev inequalities.} 

The following statements are taken from \cite[Part 1]{PLF-YM-main}. 

\begin{proposition}[Sup-norm Sobolev inequality. Hyperboloidal domain]
\label{prop:glol-Soin}
For any function defined on a hypersurface $\MH_s$, the following estimate holds (in which $t^2 = s^2+ |x|^2$): 
$$
\sup_{\MH_s} t^{3/2} \, |u(t,x)|
\lesssim 
\sum_{|J| \leq 2} \| L^J u\|_{L^2(\MH_s)}
\simeq 
\sum_{m=0,1,2} \| t^m (\slashed \del^\H)^m  u\|_{L^2(\MH_s)}.   
$$
\end{proposition}

\begin{proposition}[Sobolev decay for wave fields in the Euclidean-merging domain] 
\label{Sobolev-ext}
For all $\eta \in [0,1)$ and all functions $u$, one has (with $k \leq p$)
\begin{subequations}
\begin{equation} \label{eq 1 lem 2 d-e-I}
\big\| r  \, \crochet^\eta \, |\del u|_{p,k} \big\|_{L^\infty(\MME_s)}  
+ \big\|  r^{1+ \eta} \, | \delsN  u |_{p,k} \big\|_{L^\infty(\MME_s)} 
\lesssim (1-\eta)^{-1} \, \Fenergy_\eta^{\ME,p+3, k+3}(s,u)
\end{equation} 
and, for $1/2 < \eta = 1/2+\delta < 1$,
\begin{equation} \label{eq 1 lem 2 d-e-I-facile}
\| r \, \crochet^{-1+\eta} |u|_{N-2}\|_{L^\infty(\MME_s)} 
\lesssim \delta^{-1} \, \Fenergy_\eta^{\ME,N}(s,u) + \Fenergy_{\eta}^{0}(s,u). 
\end{equation}
\end{subequations}
\end{proposition}

\begin{proposition}[Sobolev decay for Klein-Gordon fields in the Euclidean-merging domain] 
\label{lem 2 d-e-II}
Fix some $\eta \in [0,1)$. 
For any function $v$ one has  (with $k \leq p$) 
\begin{equation} \label{eq decay-v-repeat000-two}
c \, \| r \, \crochet^{\eta} |v|_{p,k} \|_{L^\infty (\MME_s)}
\lesssim
\Fenergy^{\ME,p+2,k+2}_{\eta,c}(s,v).
\end{equation}
\end{proposition}


\paragraph{Linear estimates on Klein-Gordon fields}

\begin{proposition}[Pointwise decay of Klein-Gordon fields]
\label{lem 1 d-KG-e}
Given any exponent $\expeta \in (0,1)$, any  solution $v$ to $ 3\coeffi \Box \phi -  \Ufrak'(0) \phi = f$ defined in $\MME_{[s_0,s_1]}$ satisfies  
$$
\Ufrak'(0)|\phi|_{p,k} \lesssim 
\begin{cases}
r^{-2} \crochet^{1-\mu} \FCenergy_{\mu,\coeffi}^{\ME,p+4,k+4}(s,\phi) + |f|_{p,k}
&  \text{ in }\Mnear_{[s_0,s_1]},
\\
r^{-1-\mu}  \FCenergy_{\mu,\coeffi}^{\ME,p+2,k+2}(s,\phi)\quad 
& \text{ in }\Mfar_{[s_0,s_1]}.
\end{cases}
$$
In the same manner, any solution $\psi$ to $\Box \psi = V'(0)\psi = f$ defined in $M^{\EM}_{[s_0,s_1]}$ satisfies
$$
V'(0)|\psi|_{p,k}\lesssim 
\begin{cases}
r^{-2} \crochet^{1-\mu} \FSenergy_{\mu}^{\ME,p+4,k+4}(s,\phi) + |f|_{p,k}
&  \text{ in }\Mnear_{[s_0,s_1]},
\\
r^{-1-\mu}  \FSenergy_{\mu}^{\ME,p+2,k+2}(s,\phi)\quad 
& \text{ in }\Mfar_{[s_0,s_1]}.
\end{cases}
$$
\end{proposition}

\begin{proof} 
In $\MMEnear_{[s_0,s_1]}$ we consider a solution $\phi$ to  
$3\coeffi \Box \phi - \Ufrak'(0) \, \phi = f$ and start from the decomposition 
\begin{equation} \label{eq 1 d-KG-e}
\aligned
\Ufrak'(0) \phi 
& =  3\coeffi \big(r^2/t^2-1\big)\del_t \, \del_t \phi
\\
& \quad
- 3\coeffi t^{-1} \Big((2x^a/t)\del_tL_a - \sum_a\delsH_aL_a - (x^a/t)\delsH_a + \big(3+(r/t)^2 \big)\del_t\Big)\phi
+  f. 
\endaligned
\end{equation}
We write $\Ufrak'(0) \, |\phi| \lesssim \coeffi t^{-1} |r-t| \, | \del\del \phi| + \coeffi t^{-1} |\del \phi|_{1,1} + |f|$, and using this observation with $\phi$ replaced by $Z\phi$  
and recalling the ordering properties in \cite[Proposition~5.2]{PLF-YM-main}, we arrive at  
$$ 
\Ufrak'(0) \, |\phi|_{p,k} \lesssim \coeffi t^{-1} |r-t| \, | \del \phi|_{p+1,k} + \coeffi t^{-1} |\del \phi|_{p+1,k+1} + |f|_{p,k}. 
$$ 
together with the consequence \eqref{eq decay-v-repeat000-two} and 
substituting these bounds in the above inequality, we obtain 
$$
\Ufrak'(0) \, |\phi|_{p,k} \lesssim 
t^{-2} \crochet^{1-\mu} \, \FCenergy_{\mu,\coeffi}^{\ME,p+4,k+3}(s,\phi) + t^{-2} \crochet^{-\mu} \, \FCenergy_{\mu,\coeffi}^{\ME,p+4,k+4}(s,v) + |f|_{p,k}.
$$ 
This concludes the bound in $\Mnear_{[s_0,s_1]}$.
Finally, we again recall \eqref{eq decay-v-repeat000-two} which, in the far region $\MMEfar_s$, gives us the desired estimate.  
\end{proof}

 
\paragraph{Poincar\'e inequalities.} In our analysis the following functional inequalities are also useful, in which we recall that $\zeta$ was introduced in \eqref{equa-zeta}. 
 
\begin{proposition}[Poincar\'e-type inequalities in the Euclidean-merging domain] 
\label{propo-Poincare-ext}
\begin{subequations}
Fix an exponent {$\expeta =1/2 +\delta$} with $\delta > 0$. 
For any function $u$ defined in $\MME_s = \{(t,x)\in\Mcal_s\, / \, |x|\geq \rhoH(s)\}$, one has  
\be \label{Poincare-trex--zeta}
\aligned
\|  \crochet^{-1 + \expeta}  u\|_{L^2(\MME_s)}
& \lesssim
\big( 1+ \delta^{-1} \big) \|  \crochet^{\expeta} \delsME u\|_{L^2(\MME_s)}
+ 
\|  r^{-1} \crochet^{\expeta} u\|_{L^2(\MME_s)}, 
\\
\|\crochet^{-1 + \expeta}\zeta u\|_{L^2(\MME_s)}
& \lesssim
(1+\delta^{-1})\|\crochet^{\expeta}\zeta \delsME u\|_{L^2(\MME_s)}
+ 
\| r^{-1}\crochet^{\expeta}\zeta u\|_{L^2(\MME_s)}. 
\endaligned
\ee
\ese
\end{proposition}

\end{document}